\documentclass[acmsmall]{acmart}
\usepackage{graphicx}
\usepackage{booktabs} 
\usepackage{float}
\usepackage{amsmath}
\usepackage{subcaption}
\usepackage{slashbox}
\usepackage{hyperref} 
\usepackage{pbox}
\usepackage{paralist}
\usepackage[]{color}
\usepackage{algorithm}
\usepackage[noend]{algpseudocode}
\usepackage{comment}
\usepackage{arydshln}
\usepackage{wrapfig}
\usepackage{multirow}
\usepackage{balance}
\usepackage{epstopdf}
\usepackage{pifont}
\usepackage{wrapfig}
\newcommand{\cmark}{\ding{51}}%
\newcommand{\xmark}{\ding{55}}%
\newcommand{\fm}{f-measure}
\newcommand{\fmspace}{f-measure }

\usepackage{tikz}
\usetikzlibrary{arrows, decorations.markings}
\colorlet{mygray}{black!60}
\tikzset{thicker line small arrows m/.style args={#1in#2}{
		draw=#2,
		solid,
		line width=#1,
		shorten >=1mm,
		decoration={
			markings,
			mark=at position 1.0 with {\arrow[fill=#2,thin]{triangle 90}}
		},
		postaction={decorate}
}}
\usepackage{pgf-pie}
\usetikzlibrary{chains,matrix,positioning,scopes}
\makeatletter
\tikzset{
	join/.code=\tikzset{after node path={%
			\ifx\tikzchainprevious\pgfutil@empty\else(\tikzchainprevious)%
			edge[every join]#1(\tikzchaincurrent)\fi}
	},
	>=stealth',
	every on chain/.append style={join},
	every join/.style={->},
	labeled/.style={
		execute at begin node=$\scriptstyle, execute at end node=$
	}
}

\usetikzlibrary{shapes,arrows}
\usetikzlibrary{chains,shapes.multipart}\usetikzlibrary{automata,positioning}
\tikzstyle{block} = [rectangle, draw, fill=blue!20, text width=5em, text 
centered, rounded corners, minimum height=2em]
\tikzstyle{line} = [draw, -latex']
\tikzstyle{map} = [draw, node distance=3cm, minimum height=2em]
\tikzstyle{out2} = [node distance=3cm, minimum height=2em]
\tikzstyle{reg} = [diamond, draw, fill=red!20, text width=2.6em, text badly 
centered, node distance=3cm, inner sep=0pt]
\tikzset{every picture/.style={line width=.7pt}}
\usetikzlibrary{trees}
\tikzstyle{every node}=[draw=black,thick,anchor=west]
\tikzstyle{schema}=[draw=black,fill=gray!30, very thick]
\tikzstyle{table}=[dashed,fill=gray!15]
\tikzstyle{refkey}=[dashed]
\usepackage{color, colortbl}
\definecolor{Gray}{gray}{0.9}

\graphicspath{{./images/}}
\DeclareGraphicsExtensions{.pdf, .png, .jpg, .eps}

\newcommand{\add}[1]{\textcolor{black}{{#1}}} 

\newcommand{\sysName}{{\sf PoWareMatch}}
\newcommand{\sysNameSpace}{{\sf PoWareMatch }}
\newcommand{\deltacorr}{\Delta}

\newtheorem{example}{Example}
\newtheorem{prob}{Problem}
\newtheorem{definition}{Definition}
\newtheorem{prop}{Proposition}
\newtheorem{corol}{Corollary}
\newtheorem{lemma}{Lemma}
\newtheorem{theorem}{Theorem}
\newtheorem{calc}{Computation}

\def\@xfootnote[#1]{%
	\protected@xdef\@thefnmark{#1}%
	\@footnotemark\@footnotetext}
\makeatother

\newcommand{\TechReport}{}

\AtBeginDocument{%
	\providecommand\BibTeX{{%
			\normalfont B\kern-0.5em{\scshape i\kern-0.25em b}\kern-0.8em\TeX}}}

\setcopyright{acmcopyright}
\copyrightyear{2021}
\acmYear{2021}
\acmDOI{10.1145/1122445.1122456}

\acmJournal{JDIQ}
\acmVolume{xx}
\acmNumber{x}
\acmArticle{xxx}
\acmMonth{3}

\begin{document}

\ifdefined\TechReport
\title{{\sf PoWareMatch}: a Quality-aware Deep Learning Approach to Improve Human Schema Matching -- Technical Report}
\else 
\title{{\sf PoWareMatch}: a Quality-aware Deep Learning Approach\\ to Improve Human Schema Matching}
\fi
\author{Roee Shraga}
\email{shraga89@campus.technion.ac.il}
\author{Avigdor Gal}
\email{avigal@technion.ac.il}
\affiliation{
	\institution{Technion -- Israel Institute of Technology}
	\city{Haifa}
	\country{Israel}
}

\renewcommand{\shorttitle}{{\sf PoWareMatch}}
\begin{abstract}
Schema matching is a core task of any data integration process. Being investigated in the fields of databases, AI, Semantic Web and data mining for many years, the main challenge remains the ability to generate quality matches among data concepts (\emph{e.g.,} database attributes). In this work, we examine a novel angle on the behavior of humans as matchers, studying match creation as a process. We analyze the dynamics of common evaluation measures (precision, recall, and f-measure), with respect to this angle and highlight the need for unbiased matching to support this analysis. Unbiased matching, a newly defined concept that describes the common assumption that human decisions represent reliable assessments of schemata correspondences, is, however, not an inherent property of human matchers. In what follows, we design {\sf PoWareMatch} that makes use of a deep learning mechanism to calibrate and filter human matching decisions adhering the quality of a match, which are then combined with algorithmic matching to generate better match results. We provide an empirical evidence, established based on an experiment with more than 200 human matchers over common benchmarks, that {\sf PoWareMatch} predicts well the benefit of extending the match with an additional correspondence and generates high quality matches. In addition, {\sf PoWareMatch} outperforms state-of-the-art matching algorithms.
\end{abstract}

\maketitle

\section{Introduction}

\emph{Schema matching} is a core task of data integration for structured and semi-structured data. Matching revolves around providing correspondences between concepts describing the meaning of data in various heterogeneous, distributed data sources, such as SQL and XML schemata, entity-relationship diagrams, ontology descriptions, interface definitions, {\em etc.} The need for schema matching arises in a variety of domains including linking datasets and entities for data discovery~\cite{fernandez2018seeping,singh2017synthesizing,thirumuruganathan2020data}
, finding related tables in data lakes~\cite{zhang2020finding}, data enrichment~\cite{wang2019progressive}
, aligning ontologies and relational databases for the Semantic Web~\cite{EUZENAT2007a}
, and document format merging ({\em e.g.}, orders and invoices in e-commerce)~\cite{RAHM2001}. As an example, a shopping comparison app that supports queries such as ``the cheapest computer among retailers'' or ``the best rate for a flight to Boston in September'' requires integrating and matching several data sources of product orders and airfare forms.

Schema matching research originated in the database community~\cite{RAHM2001} and has been a focus for other disciplines as well, from artificial intelligence~\cite{halevy2003corpus}, to semantic web~\cite{EUZENAT2007a}, to data mining~\cite{he2005making,roee2018icdm}. Schema matching research has been going on for more than 30 years now, focusing on designing high quality matchers, automatic tools for identifying correspondences among database attributes. Initial heuristic attempts ({\em e.g.}, COMA~\cite{DO2002a} and Similarity Flooding~\cite{MELNIK2002}) were followed by theoretical grounding ({\em e.g.}, see~\cite{BELLAHSENE2011,GAL2011,DONG2009}). 


Human schema and ontology matching, the holy grail of matching, requires domain expertise~\cite{li2017human,dragisic2016user}. Zhang \emph{et al.} stated that users that match schemata are typically non experts, and may not even know what is a schema~\cite{Zhang2013}. Others, {\em e.g.},~\cite{Ross2010,HILDA18}, have observed the diversity among human inputs. Recently, human matching was challenged by the information explosion (a.k.a Big Data) that provided many novel sources for data and with it the need to efficiently and effectively integrate them. So far, challenges raised by human matching were answered by pulling further on human resources, using crowdsourcing ({\em e.g.},~\cite{NoyMMA13,Crowdmap,Zhang2013,zhang2018reducing,fan2014hybrid}) and pay-as-you-go frameworks ({\em e.g.}, \cite{McCann2008,Hung2014,pinkel2013incmap}).
However, recent research has challenged both traditional and new methods for human-in-the-loop matching, showing that humans have cognitive biases that decrease their ability to perform matching tasks effectively~\cite{ackerman2019cognitive}. For example, the study shows that over time, human matchers are willing to determine that an element pair matches despite their low confidence in the match, possibly leading to poor performance. 

Faced with the challenges raised by human matching, we offer a novel angle on the behavior of humans as matchers, analyzing {\em matching as a process}. We now motivate our proposed analysis using an example and then outline the paper's contribution.
\subsection{Motivating Example}
When it comes to humans performing a matching task, decisions regarding correspondences among data sources are made sequentially. To illustrate, consider Figure~\ref{fig:ex1}, presenting two simplified purchase order schemata adopted from~\cite{DO2002a}. {\sf PO$_1$} has four attributes (foreign keys are ignored for simplicity): purchase order's number ({\sf poCode}), timestamp ({\sf poDay} and {\sf poTime}) and shipment city ({\sf city}). {\sf PO$_2$} has three attributes: order issuing date ({\sf orderDate}), order number ({\sf orderNumber}), and shipment city ({\sf city}). A human matching sequence is given by the orderly annotated double-arrow edges. For example, a matching decision that {\sf poDay} in {\sf PO$_1$} corresponds to {\sf orderNumber} in {\sf PO$_2$} is the second decision made in the process.     

\begin{figure}[t]
	\centering
	\begin{subfigure}{.35\linewidth}
		\scalebox{0.7}{\begin{tikzpicture}[%
				grow via three points={one child at (0.25,-0.7) and
					two children at (0.25,-0.7) and (0.25,-1.4)},
				edge from parent path={(\tikzparentnode.south) |- (\tikzchildnode.west)}]
				\node [schema] (a) {\sf PO$_{1}$}
				child { node [table] {\sf City}
					child { node (pocode) {\sf poCode}}
					child { node (pocity) {\sf city}}}
				child [missing] {}				
				child [missing] {}	
				child { node[table] {\sf DateTime}
					child { node (poday) {\sf poDay}}
					child { node (potime) {\sf poTime}}
					child { node (pocode) {\sf poCode}}
				}
				;
				\node [schema, right of=a,node distance=4cm] {\sf PO$_{2}$}
				child { node [table] {\sf Order\_Details}
					child { node (ordernumber) {\sf orderNumber}}
					child { node (ordercity) {\sf city}}
					child { node (orderdate){\sf orderDate}}}
				;
				\draw[<->] (poday.east) -- ([xshift=1.5cm]poday) -| ([xshift=-1.5cm]orderdate.west) -- ([yshift=.1cm]orderdate.west) node[draw=none,pos=.5,font=\tiny,above] {$(1)$};
				\draw[<->] (potime.east) -- ([xshift=1cm]potime) -| ([xshift=-2cm, yshift=.1cm]ordernumber.west) -- ([yshift=.1cm]ordernumber.west) node[draw=none,pos=.65,font=\tiny,above] {$(2)$};
				\draw[<->] (potime.east) -- ([xshift=1.5cm,yshift=-.5cm]potime) -| ([xshift=-1.5cm, yshift=-.84cm]orderdate.west) -- ([yshift=-.1cm]orderdate.west) node[draw=none,pos=.5,font=\tiny,above] {$(3)$};
				\draw[<->] (pocity.east) -- ([yshift=-.1cm]ordercity.west) node[draw=none,pos=0.825,font=\tiny,above] {$(4)$};
				\draw[<->] (poday.east) -- ([xshift=.5cm,yshift=.5cm]poday) -| ([xshift=-1.75cm, yshift=-.3cm]ordernumber.west) -- ([yshift=-.1cm]ordernumber.west) node[draw=none,pos=.3,font=\tiny,above] {$(5)$}; 
		\end{tikzpicture}}
		\caption{Matching-as-a-Process example over two schemata. Decision ordering is annotated using the double-arrow edges.}
		\label{fig:ex1}
	\end{subfigure}
	\hfill
	\begin{subfigure}{.55\linewidth}
		\includegraphics[width=1.1\textwidth]{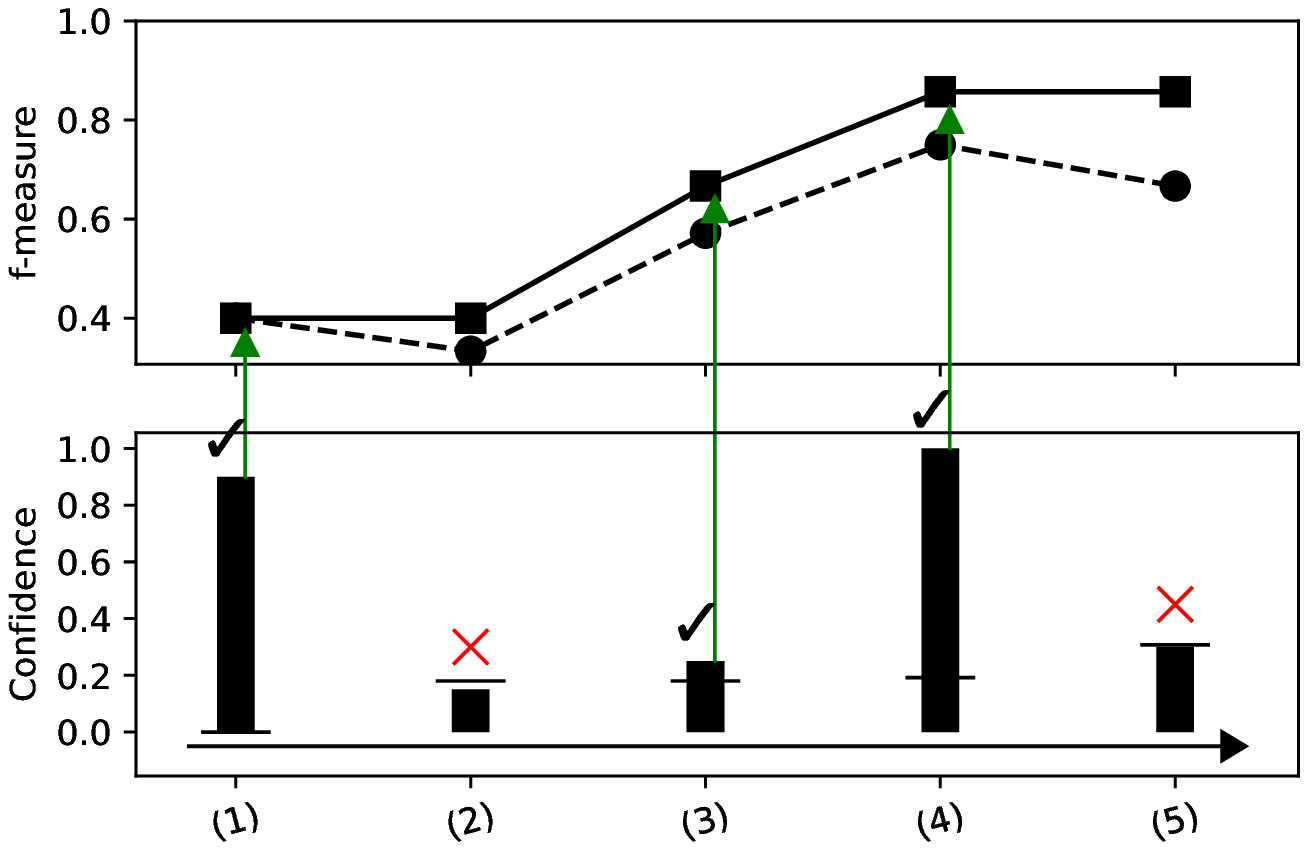}
		\caption{Matching-as-a-Process performance. Ordered decision of Figure~\ref{fig:ex1} on the x-axis, confidence associated with each decision is given at the bottom, and a dynamic threshold to accept/reject a human matching decision with respect to \fmspace (Section~\ref{sec:unbiasedprocessing}) is represented as horizontal lines. The top circle markers represent the performance of the traditional approach, unconditionally accepting human decisions, and the square markers represent a process-aware inference over the decisions based on the thresholds at the bottom.}
		\label{fig:exF}
	\end{subfigure}
	\caption{Matching-as-a-Process Motivating Example.}
	\label{fig:ex}
	\vskip -.1in
\end{figure}

The traditional view on human matching accepts human decisions as a ground truth (possibly subject to validation by additional human matchers) and thus the outcome match is composed of all the correspondences selected (or validated) by the human matcher. Figure~\ref{fig:exF} illustrates this view using the decision making process of Figure~\ref{fig:ex1}. The x-axis represents the ordering according to which correspondences were selected. The dashed line at the top illustrates the changes to the \fmspace (see Section~\ref{sec:smodel} for \fm 's definition) as more decisions are added, starting with an \fmspace of $0.4$, dropping slightly and then rising to a maximum of $0.75$ before dropping again as a result of the final human matcher decision.

In this work we offer an alternative approach that analyzes the sequential nature of human matching decision making and confidence values humans share to monitor performance and modify decisions accordingly. In particular, we set a dynamic threshold (marked as a horizontal line for each decision at the bottom of Figure~\ref{fig:exF}), that takes into account previous decisions with respect to the quality of the match (\fmspace in this case). Comparing the threshold to decision confidence (marked as bars in the figure), the algorithm determines whether a human decision is included in the match (marked using a \cmark sign and an arrow to indicate inclusion) or not (marked as a red $X$). A process-aware approach, for example, accepts the third decision that {\sf poTime} in {\sf PO$_1$} corresponds with {\sf orderDate} in {\sf PO$_2$} made with a confidence of $0.25$ while rejecting the fifth decision that {\sf poDay} in {\sf PO$_1$} corresponds with {\sf orderNumber} in {\sf PO$_2$}, despite the higher confidence of $0.3$. The \fmspace of this approach, as illustrated by the solid line at the top of the figure, demonstrates a monotonic non-decreasing behavior with a final high value of $0.86$.

More on the method, assumptions and how the acceptance threshold is set can be found in Section~\ref{sec:unbiasedprocessing}. Section~\ref{sec:poware} introduces a deep learning methodology to calibrate matching decisions, aiming to optimize the quality of a match with respect to a target evaluation measure.

\subsection{Contribution}
\add{In this work we focus on (real-world) human matching with the overarching goal of improving matching quality using deep learning. Specifically, we analyze a setting where human matchers interact directly with a pair of schemata, aiming to find accurate correspondences between them.} In such a setting, the matching decisions emerge as a process ({\em matching-as-a-process}, see Figure~\ref{fig:ex}) and can be monitored accordingly to advocate the quality of a current outcome match. In what follows, we characterize the dynamics of matching-as-a-process with respect to common quality evaluation measures, namely, precision, recall, and \fm, by defining a \emph{monotonic evaluation measure} and its probabilistic derivative. We show conditions under which precision, recall and \fmspace are monotonic and identify correspondences that their addition to a match improves on its quality. These conditions provide solid, theoretically grounded decision making, leading to the design of a step-wise matching algorithm that uses human confidence to construct a quality-aware match, taking into account the process of matching.

The theoretical setting described above requires human matchers to offer {\em unbiased matching} (which we formally define in this paper), assuming human matchers are experts in matching. This is, unfortunately, not always the case and human matchers were shown to have cognitive biases when matching~\cite{ackerman2019cognitive}, which may lead to poor decision making. Rather than aggregately judge human matcher proficiency and discard those that may seem to provide inferior decisions, we propose to directly overcome human biases and accept only high-quality decisions. 
To this end, we introduce \sysNameSpace ({\bf P}r{\bf o}cess a{\bf Ware} {\bf Match}er), a quality-aware deep learning approach to calibrate and filter human matching decisions and combine them with algorithmic matching to provide better match results. We performed an empirical evaluation with over $200$ human matchers to show the effectiveness of our approach in generating high quality matches. In particular, since we adopt a supervised learning approach, we also demonstrate the applicability of \sysNameSpace over a set of (unknown) human matchers over an (unfamiliar) matching problem. 

The paper offers the following four specific contributions: 
\begin{compactenum}
	\item A formal framework for evaluating the quality of (human) matching-as-a-process using the well known evaluation measures of precision, recall, and \fmspace (Section~\ref{sec:matching_eval}).
	\item A matching algorithm that uses confidence to generate a process-aware match with respect to an evaluation measure that dictates its quality (Section~\ref{sec:unbiasedprocessing}).
	\item \sysNameSpace (Section~\ref{sec:poware}), a matching algorithm that uses a deep learning model to calibrate human matching decisions (Section~\ref{sec:unbias}) and algorithmic matchers to complement human matching (Section~\ref{sec:output}).
	\item An empirical evaluation showing the superiority of our proposed solution over state-of-the-art in schema matching using known benchmarks (Section~\ref{sec:collbModeleval}).
\end{compactenum}

Section~\ref{sec:Background} provides a matching model and discusses algorithmic and human matching. The paper is concluded with related work (Section~\ref{sec:related}), concluding remarks and future work (Section~\ref{sec:discussion}).

\section{Model}
\label{sec:Background}
We now present the foundations of our work. Section~\ref{sec:smodel} introduces a matching model and the matching problem, followed by algorithmic (Section~\ref{ex:matchers}) and human matching (Section~\ref{sec:human_back}).

\subsection{Schema Matching Model}
\label{sec:smodel}

\begin{sloppypar}
	Let $S, S^{\prime}$ be two schemata with attributes $\{a_1, a_2, \dots , a_n\}$ and $\{b_1, b_2, \dots , b_m\}$, respectively.
	A matching model matches $S$ and $S^{\prime}$ by aligning their attributes using {\em matchers} that utilize matching cues such as attribute names, instances, schema structure, {\em etc.} (see surveys, {\em e.g.},~\cite{Bernstein2011} and books, {\em e.g.},~\cite{GAL2011}). 
\end{sloppypar}
A matcher's output is conceptualized as a matching matrix $M(S,S^{\prime})$ (or simply $M$), as follows.
\begin{definition}\label{def:match}
	$M(S,S^{\prime})$ is a {\em matching matrix}, having entry $M_{ij}$ (typically a real number in $[0,1]$) represent a measure of fit (possibly a similarity or a confidence measure) between $a_i\in{S}$ and $b_j\in{S^{\prime}}$. 
	
	\noindent $M$ is \emph{binary} if for all $1\leq i\leq n$ and $1\leq j\leq m$, $M_{ij}\in\left\{{0,1}\right\}$.
	
	\noindent A {\em match}, denoted $\sigma$, between $S$ and $S^{\prime}$ is a subset of $M$'s entries, each referred to as a {\em correspondence}. 
	
	\noindent $\Sigma = \mathcal{P}(S\,\times\,S^{\prime})$ is the set of all possible matches, where $\mathcal{P}(\cdot)$ is a power-set notation.
\end{definition}
Let $M^*$ be a {\em reference matrix}. $M^*$ is a binary matrix, such that $M^*_{ij}=1$ whenever $a_i\in{S}$ and $b_j\in{S^{\prime}}$ correspond and $M^*_{ij}=0$ otherwise. A {\em reference match}, denoted $\sigma^{*}$, is given by $\sigma^{*} = \{M^*_{ij}| M^*_{ij} = 1\}$. \add{Reference matches are typically compiled by domain experts over the years in which a dataset has been used for testing.} $G_{\sigma^{*}}:\Sigma \rightarrow [0,1]$ is an evaluation measure, assigning scores to matches according to their ability to identify correspondences in the reference match. Whenever the reference match is clear from the context, we shall refer to $G_{\sigma^{*}}$ simply as $G$. We define the precision ($P$) and recall ($R$) evaluation measures~\cite{BELLAHSENE2011a}, as follows:

\begin{equation}
	P(\sigma)=\frac{\mid\sigma\cap \sigma^{*}\mid}{\mid\sigma\mid}, 
	R(\sigma)=\frac{\mid\sigma\cap \sigma^{*}\mid}{\mid \sigma^{*}\mid}
	\label{eq:PandR}
\end{equation}
The \fmspace ($F_1$ score), $F(\sigma)$, is calculated as the harmonic mean of $P(\sigma)$ and $R(\sigma)$.



The schema matching problem is expressed as follows.

\begin{prob}[Matching]
	Let $S, S^{\prime}$ be two schemata and $G_{\sigma^{*}}$ be an evaluation measure wrt a reference match $\sigma^{*}$. We seek a match $\sigma\in\Sigma$, aligning attributes of $S$ and $S^{\prime}$, which maximizes $G_{\sigma^{*}}$.
	\label{prob:matching}
\end{prob}

\subsection{Algorithmic Schema Matching}\label{ex:matchers}

Matching is often a stepped procedure applying algorithms, rules, and constraints. Algorithmic matchers can be classified into those that are applied directly to the problem (first-line matchers -- 1LMs) and those that are applied to the outcome of other matchers (second-line matchers -- 2LMs). 1LMs receive (typically two) schemata and return a matching matrix, in which each entry $M_{ij}$ captures the similarity between attributes $a_i$ and $b_j$. 2LMs receive (one or more) matching matrices and return a matching matrix using some function $f(M)$~\cite{GAL2011}. Among the 2LMs, we term {\em decision makers} those that return a binary matrix as an output, from which a match $\sigma$ is derived, by maximizing $f(M)$, as a solution to Problem~\ref{prob:matching}.  


To illustrate the algorithmic matchers in the literature, consider three 1LMs, namely {\sf Term}, {\sf WordNet}, and {\sf Token Path} and three 2LMs, namely {\sf Dominants, Threshold($\nu$)}, and {\sf Max-Delta($\delta$)}. {\sf Term}~\cite{GAL2011} compares attribute names to identify syntactically similar attributes ({\em e.g.}, using edit distance and soundex). {\sf WordNet} uses abbreviation expansion and tokenization methods to generate a set of related words for matching attribute names~\cite{wordnet1}. {\sf Token Path}~\cite{PEUKERT2011} integrates node-wise similarity with structural information by comparing the syntactic similarity of full paths from root to a node. {\sf Dominants}~\cite{GAL2011} selects correspondences that dominate all other correspondences in their row and column.  {\sf Threshold($\nu$)} and {\sf Max-Delta($\delta$)} are selection rules, prevalent in many matching systems~\cite{DO2002a}. {\sf Threshold($\nu$)} selects those entries $(i,j)$ having $M_{ij}\geq{\nu}$. {\sf Max-Delta($\delta$)} selects those entries that satisfy: $M_{ij}+\delta\geq{\max_{i}}$, where $\max_{i}$ denotes the maximum match value in the $i$'th row.

\vspace{.2cm}
\hspace{-.3cm}\begin{minipage}{.47\textwidth}\begin{example}\label{ex:main}
		Figure~\ref{ex2} provides an example of algorithmic matching over the two purchase order schemata from Figure~\ref{fig:ex1}. The top right (and bottom left) matching matrix is the outcome of {\sf Term} and the bottom right is the outcome of {\sf Threshold($0.1$)}. The projected match is $\sigma_{alg} = \{M_{11}, M_{12}, M_{13}, M_{14}, M_{31}, M_{32}, M_{34}\}$.\footnotemark The reference match for this example is given by $\{M_{11}, M_{12}, M_{23}, M_{34}\}$ and accordingly $P(\sigma_{alg}) = 0.43$, $R(\sigma_{alg}) = 0.75$, and $F(\sigma_{alg}) = 0.55$.
\end{example}\end{minipage}\hspace{.3cm}
\begin{minipage}{.5\textwidth}
	\includegraphics[width=\columnwidth]{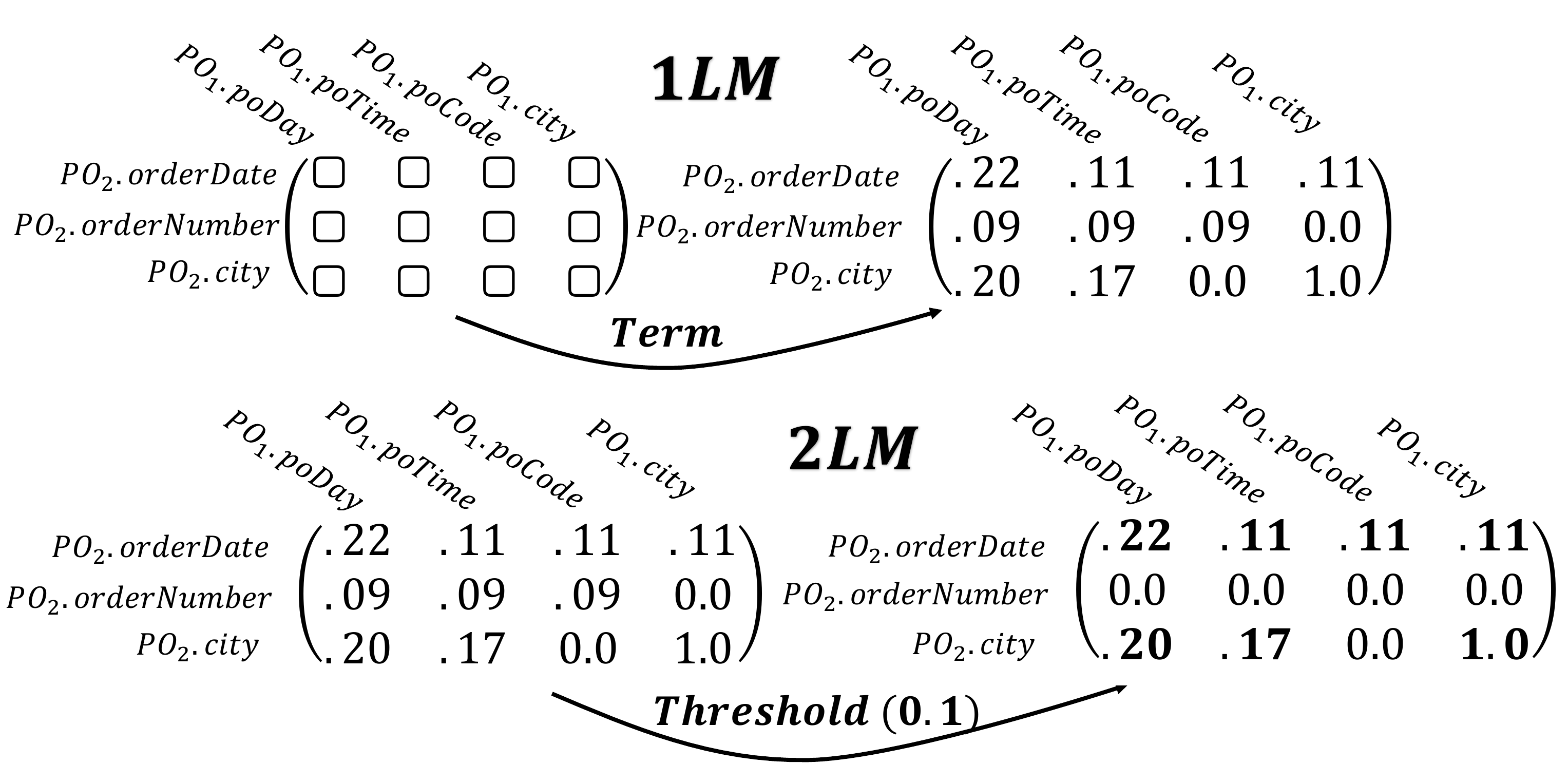}
	\captionof{figure}{Algorithmic matching example.}
	\label{ex2}
\end{minipage}
\footnotetext{Recall the $M_{ij}$ represents a correspondence between the $i$'th element in $S$ and the $j$'th element in $S^{\prime}$, \emph{e.g.,} $M_{11}$ means that {\sf PO$_1$.poDay} and {\sf PO$_2$.orderDate} correspond.}

\subsection{Human Schema Matching}\label{sec:human_back}

In this work, we examine human schema matching as a decision making process. Different from a binary crowdsourcing setting (\emph{e.g.,}~\cite{zhang2018reducing}), in which the matching task is typically reduced into a set of (binary) judgments regarding specific correspondences (see related work, Section~\ref{sec:related}), we focus on human matchers that perform a matching task in its entirety, as illustrated in Figure~\ref{fig:ex}.

Human schema matching is a complex sequential decision making process of interrelated decisions. Attribute of multiple schemata are examined to decide whether and which attributes correspond. Humans either validate an algorithmic result or locate a candidate attribute unassisted, possibly relying upon superficial information such as string similarity of attribute names 
or exploring information such as data-types, instances, and position within the schema tree. The decision whether to explore additional information relies upon self-monitoring of confidence.  




Human schema matching has been recently analyzed using metacognitive psychology~\cite{ackerman2019cognitive}, a discipline that investigates factors impacting humans when performing knowledge intensive tasks~\cite{Barsalou2014}. The metacognitive approach~\cite{Bjork2013}, traditionally applied for learning and answering knowledge questions, highlights the role of {\em subjective confidence} in regulating efforts while performing tasks. Metacognitive research shows that subjective judgments ({\em e.g.}, confidence) regulate the cognitive effort invested in each decision ({\em e.g.}, identifying a correspondence)~\cite{Metcalfe2008,Ackerman2017}. In what follows, we model human matching as a sequence of decisions regarding element pairs, each assigned with a confidence level. \add{In this work we assume that human matchers interact directly with the matching problem, selecting correspondences given a pair of schemata. Within this process, we directly query their (subjective) confidence level regarding a selected correspondence, which essentially reflects the ongoing monitoring and final subjective assessment of chances of success~\cite{Bjork2013,Ackerman2017}.} The dynamics of human matching decision making process is modeled using a \emph{decision history} $H$, as follows. 

\begin{definition}\label{def:history}
	\begin{sloppypar}
		Given two schemata $S, S^{\prime}$ with attributes $\{a_1, a_2, \dots , a_n\}$ and $\{b_1, b_2, \dots, b_m\}$, respectively, a history $H=\langle h_t\rangle_{t=1}^{T}$ is a sequence of triplets of the form $\langle e, c, t \rangle$, where $e = \left(a_i, b_j\right)$ such that $a_i\in{S}$, $b_j\in{S^{\prime}}$, $c\in [0,1]$ is a confidence value assigned to the correspondence of $a_i$ and $b_j$, and $t\in{\rm I\!R}$ is a timestamp, recording the time the decision was taken.
	\end{sloppypar}  
\end{definition}

Each decision $h_t\in H$ records a matching decision confidence ($h_t.c$) concerning an element pair ($h_t.e = \left(a_i, b_j\right)$) at time $t$ ($h_t.t$). Timestamps induce a total order over $H$'s elements. 

A matching matrix, which may serve as a solution of the human matcher to Problem~\ref{prob:matching}, can be created from a matching history by assigning the latest confidence to the respective matrix entry. Given an element pair $e = \left(a_i, b_j\right)$, we denote by $h_{max}^e$ the latest decision making in $H$ that refers to $e$ and compute a matrix entry as follows: 
\begin{equation}\label{eq:hutomat}
	M_{ij} =
	\begin{cases}
		h_{max}^e.c  & $\text{if }$ \exists h_t\in H |  h_t.e = \left(a_i, b_j\right) \\
		\varnothing & \text{otherwise}\\
	\end{cases} 
\end{equation}
where 
$M_{ij} = \varnothing$ means that $M_{ij}$ was not assigned a confidence value. Whenever clear from the context, we refer to a confidence value ($h_t.c$) assigned to an element pair $\left(a_i, b_j\right)$, simply as $M_{ij}$. 

\vspace{.2cm}
\addtocounter{example}{-1}
\hspace{-.3cm}\begin{minipage}{.5\textwidth}\begin{example}[cont.]
		Figure~\ref{ex3} (left) provides the decision history that corresponds to the matching process of Figure~\ref{fig:ex1} and the respective matching matrix (applying Eq.~\ref{eq:hutomat}) is given on the right. The projected match is $\sigma_{hum} = \{M_{11}, M_{22}, M_{12}, M_{34}, M_{21}\}$. Recalling the reference match, $\{M_{11}, M_{12}, M_{23}, M_{34}\}$, the projected match obtains $P(\sigma_{hum}) = 0.6$, $R(\sigma_{hum}) = 0.75$, and $F(\sigma_{hum}) = 0.67$.
\end{example}\end{minipage}\hspace{.3cm}
\begin{minipage}{.49\textwidth}
	\includegraphics[width=\columnwidth]{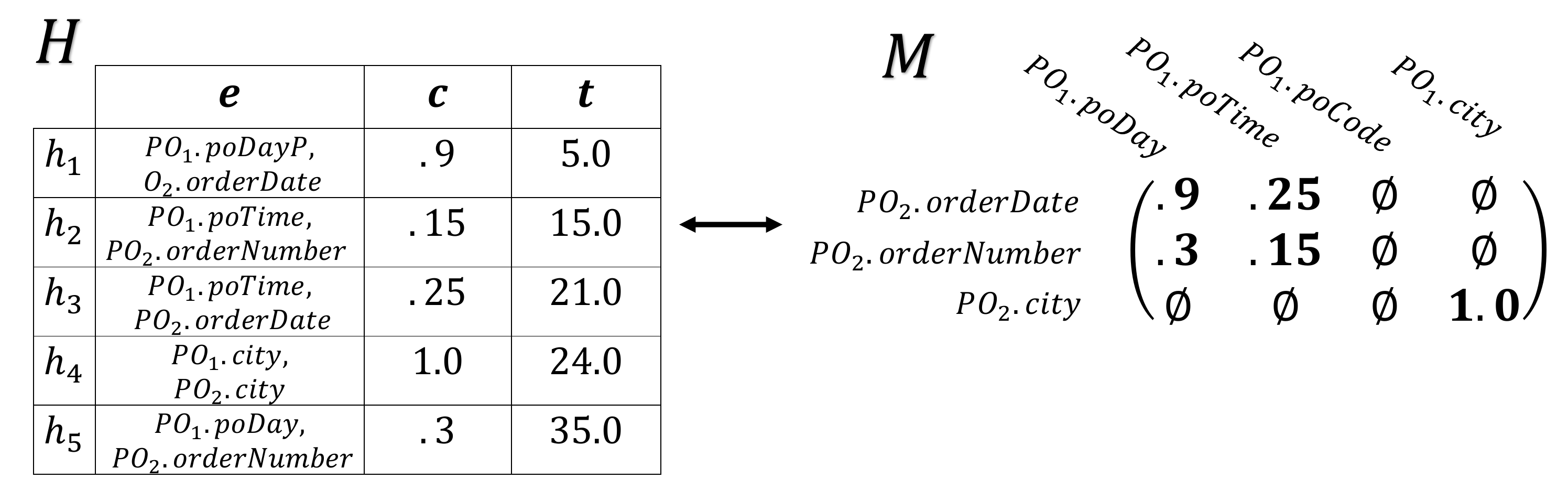}
	\captionof{figure}{Human matching example.}
	\label{ex3}
\end{minipage}
\vspace{.2cm}

In this work we seek a solution to Problem~\ref{prob:matching} that takes into account the quality of a match wrt an evaluation measure $G$ and a decision history $H$. Therefore, we next analyze the well-known measures, namely precision, recall, and \fmspace in the context of matching-as-a-process.
\section{Matching-as-a-Process and the Human Matcher}\label{sec:matching_eval}


We examine matching quality in the matching-as-a-process setting by analyzing the properties of matching evaluation measures (precision, recall, and \fm, see Eq.~\ref{eq:PandR}). The analysis uses regions of the match space $\Sigma\times\Sigma$ over which monotonicity can be guaranteed. To illustrate the method we first analyze evaluation measure monotonicity in a deterministic setting, repositioning well-known properties and showing that recall is always monotonic, while precision and \fmspace are monotonic only under strict conditions (Section~\ref{sec:monotonicEvaluation}). Then, we move to a more involved analysis of the probabilistic case, identifying which correspondence should be added to an existing partial match (Section~\ref{sec:matchAnnealing}). Finally, recalling that matching-as-a-process is a human matching characteristic, we tie our analysis to human matching and discuss the idea of unbiased matching (Section~\ref{sec:matchingBiases}).

\subsection{Monotonic Evaluation}
\label{sec:monotonicEvaluation}

Given two matches (Definition~\ref{def:match}) $\sigma$ and $\sigma^{\prime}$ that are a result of some sequential decision making such that $\sigma\subseteq\sigma^{\prime}$, we define their interrelationship using their monotonic behavior with respect to an evaluation measure $G$ ($P$, $R$, or $F$, see Eq.~\ref{eq:PandR}). For the remainder of the section, we denote by $\Sigma^{\subseteq}$ the set of all match pairs in $\Sigma\times\Sigma$ such that the first match is a subset of the second: $\Sigma^{\subseteq}=\{(\sigma,\sigma^{\prime})\in\Sigma\times\Sigma:\sigma\subseteq\sigma^{\prime}\}$ and use $\deltacorr_{\sigma,\sigma^{\prime}} = \sigma^{\prime}\setminus\sigma$ ($\deltacorr$, when clear from the context) to denote the set of correspondences that were added to $\sigma$ to generate $\sigma^{\prime}$. 

\begin{definition}[Monotonic Evaluation Measure]\label{def:monotonicity}
	Let $G$ be an evaluation measure and $\Sigma^2\subseteq\Sigma^{\subseteq}$ a set of match pairs in $\Sigma^{\subseteq}$. 
	\noindent $G$ is a \emph{monotonically increasing evaluation measure (MIEM) over $\Sigma^2$} if for all match pairs 
	$(\sigma, \sigma^{\prime})\in\Sigma^2$, $G(\sigma) \leq G(\sigma^{\prime})$.
\end{definition}


In Definition~\ref{def:monotonicity}, we use $\Sigma^2$ as a representative subspace of $\Sigma^{\subseteq}$, which was defined above. According to Definition~\ref{def:monotonicity}, an evaluation measure $G$ is monotonically increasing (MIEM) if by adding correspondences to a match, we do not reduce its value. It is fairly easy to infer that recall is an MIEM over all match pairs in $\Sigma^{\subseteq}$ and precision and \fmspace are not, unless some strict conditions hold. To guarantee such conditions, we define two subsets, $\Sigma^P=\{(\sigma, \sigma^{\prime})\in\Sigma^{\subseteq}:P(\sigma)\leq P(\deltacorr)\}$ and $\Sigma^F=\{(\sigma, \sigma^{\prime})\in\Sigma^{\subseteq}:0.5\cdot F(\sigma)\leq P(\deltacorr)\}$. The former represents a subset of all match pairs for which the precision of the added correspondences ($P(\deltacorr)$) is at least as high as the precision of the first match ($P(\sigma)$). The latter compares between the added correspondences ($P(\deltacorr)$) and the \fmspace of the first match ($F(\sigma)$). We use these two subspaces in the following theorem to summarize the main dynamic properties of the evaluation measures of precision, recall, and \fm.\footnote{The proof of Theorem~\ref{thm:MIEM} is given in Appendix~\ref{sec:proofs}.} 

\begin{theorem}\label{thm:MIEM}
	Recall ($R$) is a MIEM over $\Sigma^{\subseteq}$, Precision ($P$) is a MIEM over $\Sigma^P$, and \fmspace ($F$) is a MIEM over $\Sigma^F$.
\end{theorem}



\subsection{Local Match Annealing}
\label{sec:matchAnnealing}
The analysis of Section~\ref{sec:monotonicEvaluation} lays the groundwork for a principled matching process that continuously improves on the evaluation measure of choice. The conditions set forward use knowledge of, first, the evaluation outcome of the match performed thus far ($G(\sigma)$), and second, the evaluation score of the additional correspondences ($G(\deltacorr)$). While such knowledge can be extremely useful, it is rarely available during the matching process. Therefore, we next provide a relaxed setting, where $G(\sigma)$ and $G(\deltacorr)$ are probabilistically known (Section~\ref{sec:matchingBiases} provides an approximation using human confidence). For simplicity, we restrict our analysis to matching processes where a single correspondence is added at a time. We denote by $\Sigma^{\subseteq_1}=\{(\sigma,\sigma^{\prime})\in\Sigma^{\subseteq}:|\sigma^{\prime}|-|\sigma| = 1\}$ the set of all match pairs $(\sigma,\sigma^{\prime})$ in $\Sigma^{\subseteq}$ where $\sigma^{\prime}$ is generated by adding a single correspondence to $\sigma$. The discussion below can be extended (beyond the scope of this work) to adding multiple correspondences at a time.


We start with a characterization of correspondences whose addition to a match improves on the match evaluation. Recall that $\deltacorr$ represents the marginal set of correspondences that were added to the match. Following the specification of $\Sigma^{\subseteq_1}$, we let $\deltacorr$ represent a single correspondence ($|\deltacorr|=1$), which is the result of multiple match pairs $(\sigma,\sigma^{\prime})\in \Sigma^{\subseteq_1}$ such that $\deltacorr=\sigma^{\prime}\setminus\sigma$.

\begin{definition}[Local Match Annealing]\label{def:annealer}
	Let $G$ be an evaluation measure and $\deltacorr$ be a singleton correspondence set ($|\deltacorr|=1$). $\deltacorr$ is a \emph{local annealer with respect to $G$ over $\Sigma^2\subseteq\Sigma^{\subseteq_1}$} if for every $(\sigma,\sigma^{\prime})\in\Sigma^2$ s.t. $\deltacorr=\deltacorr_{(\sigma,\sigma^{\prime})}$:
	$G(\sigma) \leq G(\sigma^{\prime})$. 
\end{definition}

A local annealer is a single correspondence that is guaranteed to improve the performance of any match within a specific match pair subset. We now connect the MIEM property of an evaluation measure $G$ (Definition~\ref{def:monotonicity}) with the annealing property of a match delta $\deltacorr$ (Definition~\ref{def:annealer}) with respect to a specific match $\sigma^{\prime}$.\footnote{All proofs are provided in Appendix~\ref{sec:proofs}}

\begin{prop}\label{prop:annealer}
	Let $G$ be an evaluation measure. If $G$ is a MIEM over $\Sigma^2\subseteq\Sigma^{\subseteq_1}$, then $\forall(\sigma,\sigma^{\prime})\in\Sigma^2$~: $\deltacorr=\sigma^{\prime}\setminus\sigma$ is a local annealer with respect to $G$ over $\Sigma^2\subseteq\Sigma^{\subseteq_1}$.
\end{prop}


Proposition~\ref{prop:annealer} demonstrates the importance of defining the appropriate subset of matches that satisfy monotonicity. Together with Theorem~\ref{thm:MIEM}, the following immediate corollary can be deduced. 
\begin{corol}\label{corol:annealer}
	Any singleton correspondence set $\deltacorr$ ($|\deltacorr|=1$) is a local annealer with respect to 1) $R$ over $\Sigma^{\subseteq_1}$, 2) $P$ over $\Sigma^P\cap\Sigma^{\subseteq_1}$, and 3) $F$ over $\Sigma^F\cap\Sigma^{\subseteq_1}$.
\end{corol}
\noindent where $\Sigma^P$ and $\Sigma^F$ are the subspaces for which precision and \fmspace are monotonic as defined in Section~\ref{sec:monotonicEvaluation}.

Corollary~\ref{corol:annealer} indicates which correspondence should be added to a match to improve its quality with respect to an evaluation measure of choice, according to the respective condition. 


Assume now that the value of $G$, applied to a match $\sigma$, is not deterministically known. Rather, $G(\sigma)$ is a random variable with an expected value of $E(G(\sigma))$. We extend Definition~\ref{def:annealer} as follows.

\begin{definition}[Probabilistic Local Match Annealing]\label{def:probAnnealer}
	Let $G$ be a random variable, whose values are taken from the domain of an evaluation measure, and $\deltacorr$ be a singleton correspondence set ($|\deltacorr|=1$). $\deltacorr$ is a \emph{probabilistic local annealer with respect to $G$ over $\Sigma^2\subseteq\Sigma^{\subseteq_1}$} if for every $(\sigma,\sigma^{\prime})\in\Sigma^2$ s.t. $\deltacorr=\deltacorr_{(\sigma,\sigma^{\prime})}$:
	$E(G(\sigma)) \leq E(G(\sigma^{\prime}))$. 
\end{definition} 

We now define conditions (match subspaces) under which a correspondence is a probabilistic local annealer for recall ($R$), precision ($P$), and \fmspace ($F$). Let ${\rm I\!I}_{\{\deltacorr\in\sigma^*\}}$ be an indicator function, returning a value of $1$ whenever $\deltacorr$ is a part of the reference match and $0$ otherwise.
$${\small{\rm I\!I}_{\{\deltacorr\in\sigma^*\}} = 
	\begin{cases}
		1 & \text{ if }\deltacorr\in\sigma^*\\
		0 & \text{otherwise}
\end{cases}}$$ 
Using ${\rm I\!I}_{\{\deltacorr\in\sigma^*\}}$, we define the probability that $\deltacorr$ is correct: 
$Pr\{\deltacorr\in\sigma^*\} \equiv Pr\{{\rm I\!I}_{\{\deltacorr\in\sigma^*\}} = 1\}$.

Lemma~\ref{lemma:p_f_ineqProb} is the probabilistic counterpart of Lemma~\ref{lemma:p_f_ineq} (see Appendix~\ref{sec:proofs}).
\begin{lemma}\label{lemma:p_f_ineqProb}
	For $(\sigma, \sigma^{\prime})\in\Sigma^{\subseteq_1}$:
	\begin{itemize}
		\item $E(P(\sigma)) \leq E(P(\sigma^{\prime}))$ iff $E(P(\sigma))\leq Pr\{\deltacorr\in\sigma^*\}$
		\item $E(F(\sigma)) \leq E(F(\sigma^{\prime}))$ iff $0.5\cdot E(F(\sigma))\leq Pr\{\deltacorr\in\sigma^*\}$
	\end{itemize}
\end{lemma}

Using the following subsets: $\Sigma^{E(P)}=\{(\sigma, \sigma^{\prime})\in\Sigma^{\subseteq_1}:E(P(\sigma))\leq Pr\{\deltacorr\in\sigma^*\}\}$ and $\Sigma^{E(F)}=\{(\sigma, \sigma^{\prime})\in\Sigma^{\subseteq_1}:0.5\cdot E(F(\sigma))\leq Pr\{\deltacorr\in\sigma^*\}\}$, we extend Theorem~\ref{thm:MIEM} to the probabilistic setting.\footnote{The proof of Theorem~\ref{thm:probLocalAnneal} is given in Appendix~\ref{sec:proofs}.}

\begin{theorem}\label{thm:probLocalAnneal}
	Let $R/P/F$ be a random variable, whose values are taken from the domain of $[0,1]$, and $\deltacorr$ be a singleton correspondence set ($|\deltacorr|=1$). $\deltacorr$ is a \emph{probabilistic local annealer with respect to $R/P/F$ over $\Sigma^{\subseteq_1}/\Sigma^{E(P)}/\Sigma^{E(F)}$}.
\end{theorem}

Intuitively, using Theorem~\ref{thm:probLocalAnneal} one can infer those correspondences, whose addition to the current match improves its quality. For example, if a reference match contains $12$ correspondences and a current match contains $9$ correct correspondences (\emph{i.e.,} correspondences that are included in the reference match) and $1$ incorrect correspondences. The quality of the current match is $0.75$ with respect to $R$, $0.9$ with respect to $P$, and $0.82$ with respect to $F$. Any potential addition to the match does not decrease $R$. One should add correspondences associated with a probability higher than $0.9$ to probabilistically guarantee $P$ improvement and correspondences associated with a probability higher than $0.5\cdot0.82 = 0.41$ to probabilistically guarantee $F$ improvement. 

\add{Our analysis focuses on adding correspondences, the most typical action in creating a match. We can fit revisit actions in out analysis as follows. A change in confidence level of a match can be considered as a newly suggested correspondence associated with a new confidence level, removing the former decision from the current match. Similarly, if a user un-selects a correspondence, it can be treated as a new assignment associated with a confidence level of 0. The latter fits well with our model. If a human matcher decides to delete a decision, we do not ignore this decision but rather use our methodology to decide whether we should accept it.}

\subsection{Evaluation Approximation: the Case of a Human Matcher}
\label{sec:matchingBiases}
Theorem~\ref{thm:probLocalAnneal} offers a probabilistic interpretation 
by defining $\Sigma^{E(P)}$ and $\Sigma^{E(F)}$ over which matches are probabilistic local annealers (Definition~\ref{def:probAnnealer}). A key component to generating these subsets is the computation of $Pr\{\deltacorr\in\sigma^*\}$ that, in most real-world scenarios, is likely unavailable during the matching process or even after it concludes. To overcome this hurdle, we discuss next the possibility to judicially make use of human matching to assign a probability to the inclusion of a correspondence in the reference match.

The traditional view of human matchers in schema matching is that they offer a reliable assessment on the inclusion of a correspondence in a match. Given a matching decision by a human matcher, we formulate this view, as follows.

\begin{definition}[Unbiased Matching]\label{def:unbiasedMatching}
	Let $M_{ij}$ be a confidence value assigned to an element pair $\left(a_i, b_j\right)$ and $\sigma^{*}$ a reference match. $M_{ij}$ is \emph{unbiased} (with respect to $\sigma^{*}$) if $M_{ij} = Pr\{M_{ij}\in\sigma^*\}$
\end{definition}

Unbiased matching allows the use of a matching confidence to assess the probability of a correspondence to be part of a reference match. Using Definition~\ref{def:unbiasedMatching}, we define an unbiased matching matrix $M$ such that $\forall M_{ij}\in M: M_{ij} = Pr\{M_{ij}\in\sigma^*\}$ and an unbiased matching history $H$ such that $\forall h_t\in H: h_t.c = Pr\{M_{ij}\in\sigma^*\}$. 

Given an unbiased matching matrix $M$, a reference match $\sigma^{*}$, a match $\sigma \subseteq M$ and a candidate correspondence $M_{ij}\in M$, we can, using Definition~\ref{def:unbiasedMatching} and the definition of expectation, compute
\begin{equation}\label{eq:unbiasedMatchingEstim}
	Pr\{M_{ij} \in\sigma^*\} = M_{ij}, \text{ }E(P(\sigma)) = \frac{\sum_{M_{ij}\in\sigma}M_{ij}}{|\sigma|}, \text{ }E(F(\sigma)) =\frac{2\cdot\sum_{M_{ij}\in\sigma}M_{ij}}{|\sigma| + |\sigma^{*}|}
\end{equation}
and check whether $(\sigma,\sigma\cup\{M_{ij}\})\in\Sigma^{E(P)}$ and $(\sigma,\sigma\cup\{M_{ij}\})\in\Sigma^{E(F)}$. In case the size of the reference match $|\sigma^{*}|$ is unknown, it needs to be estimated, \emph{e.g.,} using $1$:$1$ matching, $|\sigma^{*}| = min(S, S^{\prime})$.\footnote{Details of Eq.~\ref{eq:unbiasedMatchingEstim} computation are given in Appendix~\ref{sec:proofs}.} 

\subsubsection{Biased Human Matching}
\label{sec:RWHMatching}

In Section~\ref{sec:human_back} we presented a human matching decision process as a history, from which a matching matrix may be derived. The decisions in the history represent corresponding element pairs chosen by the human matcher and their assigned confidence level (see Definition~\ref{def:history}). Assuming unbiased human matching (Definition~\ref{def:unbiasedMatching}), the assigned confidence can be used to determine which of the selected correspondences should be added to the current match, given an evaluation measure of choice (recall, precision, or \fm), using Eq.~\ref{eq:unbiasedMatchingEstim}.

An immediate question that comes to mind is whether human matching is indeed unbiased. Figure~\ref{fig:confacc} illustrates of the relationship between human confidence in matching and two derivations of an empirical probability distribution estimation of $Pr\{M_{ij}\in\sigma^*\}$. The results are based on our {
	\parfillskip=0pt
	\parskip=0pt
	\par}

\begin{wrapfigure}{R}{0.5\textwidth}
	\centering
	\includegraphics[width=.9\linewidth]{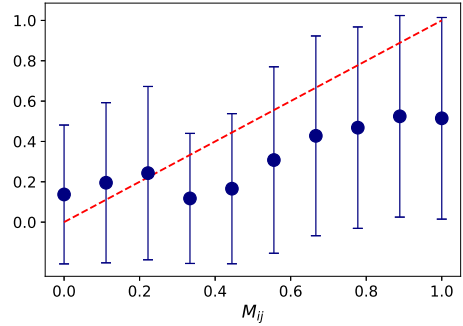}
	\caption{Is human matching biased? Confidence by correctness partitioned to 0.1 buckets.}
	\label{fig:confacc}
\end{wrapfigure}
\noindent experiments (see Section~\ref{sec:collbModeleval} for details). We partitioned the confidence levels into 10 buckets (x-axis) to allow an estimation of an expected value (blue dots) and standard deviation (vertical lines) of the accuracy results of decisions within a bucket of confidence. Each bucket includes at least 500 examples and the estimation within each bucket was calculated as the proportion of correct correspondences out of all correspondences in the bucket. The red dotted line represents theoretical unbiased matching. It is clearly illustrated that human matching {\bf is} biased. 
Therefore, human subjective confidence is unlikely to serve as a (single) good predictor to matching correctness. Ackerman \emph{et al.} reported several possible biases in human matching which affect their confidence and their ability to provide accurate matching~\cite{ackerman2019cognitive}. The interested reader to~Appendix~\ref{app:biases} for additional details. These observations and the analysis of matching-as-a-process are used in our proposed algorithmic solution to create quality matches.

\section{Process-Aware Matching Collaboration with \sysName}\label{sec:model}
Section~\ref{sec:unbiasedprocessing} introduces a decision history processing strategy under the (utopian) unbiased matching assumption (Definition~\ref{def:unbiasedMatching})  following the observations of Section~\ref{sec:matching_eval}. Then, we describe \sysName, our proposed deep learning solution to improve the quality of typical (biased) human matching (Section~\ref{sec:poware}), and detail its components (sections~\ref{sec:unbias}-\ref{sec:output}).

\subsection{History Processing with Unbiased Matching}\label{sec:unbiasedprocessing} 
Ideally, human matching is unbiased. That is, each matching decision $h_t\in H$ in the decision history (Definition~\ref{def:history}) is accompanied by an unbiased (accurate) confidence value $h_t.c = Pr\{M_{ij}\in\sigma^*\}$ (Definition~\ref{def:unbiasedMatching}). Accordingly, we can use the observations from Section~\ref{sec:matching_eval} to produce a better match as a solution to Problem~\ref{prob:matching} (with respect to some evaluation measure) out of a decision history.

Let $h_t\in H$ be a matching decision made at time $t$, such that $h_{t}.e = \left(a_i, b_j\right)$. Aiming to generate a match, we have two options, either adding the correspondence $M_{ij}$ to the match or not. Targeting recall, an MIEM over $\Sigma^{\subseteq}$ (Theorem~\ref{thm:MIEM}), the choice is clear. $\{M_{ij}\}$ is a local annealer with respect to recall over $\Sigma^{\subseteq_1}$ (Corollary~\ref{corol:annealer}) and we always benefit by adding it. Focusing on precision or \fm, the decision depends on the current match (termed $\sigma_{t-1}$). With unbiased matching, we can estimate the values of $Pr\{M_{ij}\in\sigma^*\}$, $E(P(\sigma_{t-1}))$, and $E(F(\sigma_{t-1}))$ (Eq.~\ref{eq:unbiasedMatchingEstim}) and use Theorem~\ref{thm:probLocalAnneal} as follows. 
\begin{description}
	\item[Targeting precision,] if $E(P(\sigma_{t-1}))\leq M_{ij}$, then $\{M_{ij}\}$ is a probabilistic local annealer with respect to $P$, and we can increase precision by adding $M_{ij}$ to the match ($\sigma_{t} = \sigma_{t-1}\cup\{M_{ij}\}$).
	\item[Targeting \fm,] if $0.5\cdot E(F(\sigma_{t-1}))\leq M_{ij}$, then $\{M_{ij}\}$ is a probabilistic local annealer with respect to $F$, and adding $M_{ij}$ to the match does not decrease it's quality.
\end{description}

\begin{example}\label{ex:process}
	
	Figure~\ref{fig:HPex} illustrates history processing over the history from Figure~\ref{ex3}. The left column presents the targeted evaluation measure and at the bottom, a not-to-scale timeline lays out the decision history. Along each row, \cmark and \xmark represent whether $M_{ij}$ is added to the match or not, respectively. Targeting recall (top row), all decisions are accepted and a high final recall value of $0.75$ is obtained. 
	
	\tikzstyle{every node}=[draw=none,anchor=west]

	\begin{figure}[h]
		\begin{tikzpicture}
			\matrix (m)[
			matrix of math nodes,
			nodes in empty cells, 
			row sep=0.01em, column sep=3.2em,
			nodes={font={\vphantom{S,}\small}}, 
			row 6/.style={text height=0pt,text depth=0pt} 
			]
			{ 
				\text{\small{$R$:}} & \text{\small{\cmark}} & \text{\small{\cmark}} & \text{\small{\cmark}} & \text{\small{\cmark}} & \text{\small{\cmark}} & \operatornamewithlimits{\{M_{11}, M_{22}, M_{12}, M_{34},  M_{21}\}}\limits_{(P=0.6, R=0.75, F=0.67)}\\
				\text{\small{$P$:}} & \text{\small{\cmark}} & \text{\small{\xmark}} & \text{\xmark} & \text{\small{\cmark}} & \text{\small{\xmark}} & \operatornamewithlimits{\{M_{11}, M_{34}\}}\limits_{(P=1.0, R=0.5, F=0.67)}\\
				\text{\small{$F$:}} & \text{\small{\cmark}} & \text{\small{\xmark}} & \text{\small{\cmark}} & \text{\small{\cmark}} & \text{\small{\xmark}} & \operatornamewithlimits{\{M_{11}, M_{12}, M_{34}\}}\limits_{(P=1.0, R=0.75, F=0.86)}\\
				&&&&&& \\
			};
			{ [start chain] \chainin (m-1-1);
				\chainin (m-1-2);
				\chainin (m-1-3);
				\chainin (m-1-4);
				\chainin (m-1-5); 
				\chainin (m-1-6); 
				\chainin (m-1-7);
			}
			
			{ [start chain] \chainin (m-2-1);
				\chainin (m-2-2)[join={node[above,labeled] {0.0}}];
				\chainin (m-2-3)[join={node[above,labeled] {0.9}}];
				\chainin (m-2-4)[join={node[above,labeled] {0.9}}];
				\chainin (m-2-5)[join={node[above,labeled] {0.9}}];
				\chainin (m-2-6)[join={node[above,labeled] {0.95}}]; 
				\chainin (m-2-7);
			}
			
			{ [start chain] \chainin (m-3-1);
				\chainin (m-3-2)[join={node[above,labeled] {0.0}}];
				\chainin (m-3-3)[join={node[above,labeled] {0.18}}];
				\chainin (m-3-4)[join={node[above,labeled] {0.18}}];
				\chainin (m-3-5)[join={node[above,labeled] {0.19}}];
				\chainin (m-3-6)[join={node[above,labeled] {0.31}}]; 
				\chainin (m-3-7);
			}

			\draw [->] (m-4-1 -| m-4-1.east) -- (m-4-7);
			
			\foreach [count=\i from 2] \txt in {
				\footnotesize{$M_{11}$=$0.9$}\\ $t=5$,
				\footnotesize{$M_{22}$=$0.15$}\\ $t=15$,
				\footnotesize{$M_{12}$=$0.25$}\\ $t=21$,
				\footnotesize{$M_{34}$=$1.0$}\\ $t=24$,
				\footnotesize{$M_{21}$=$0.3$}\\ $t=35$} {
				\draw ([yshift=2pt]m-4-\i.center) -- ++(0,-4pt)
				node [below,align=center,font=\footnotesize] {\txt};
			}
		\end{tikzpicture}
		\caption{History Processing Example.}
		\label{fig:HPex}
	\end{figure}
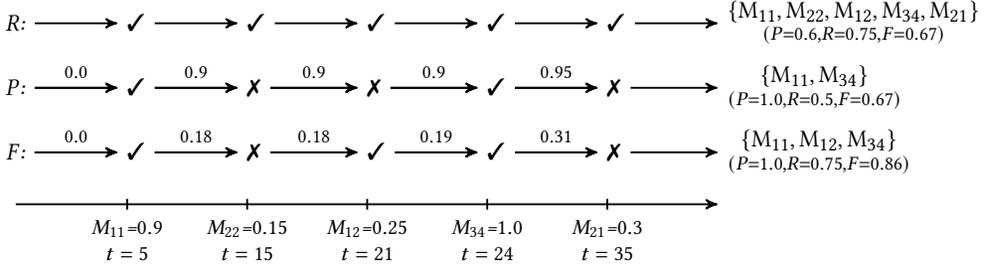
	
	
	For precision and \fm, each arrow is annotated with a decision threshold, which is set by $E(P(\sigma_{t-1}))$ and $0.5\cdot E(F(\sigma_{t-1}))$, whose computation is given in Eq.~\ref{eq:unbiasedMatchingEstim}. At the beginning of the process, $\sigma_0=\emptyset$ and $E(P(\sigma_{0}))$ is set to 0 by definition. $E(F(\sigma_{0})) = 0$ since there are no correspondences in $\sigma_0$. To illustrate the process consider,
	for example, the second threshold, computed based on the match $\{M_{11}\}$, whose value is $\frac{0.9}{1.0}=0.9$ in the second row and $0.5\cdot\frac{2\cdot0.9}{1+4}=0.18$ in the third row.  
\end{example}

\subsubsection{Setting a Static (Global) Threshold}\label{sec:glob}
The decision making process above assumes the availability of (an unlabeled) $\sigma_{t-1}$. Whenever we do not know $\sigma_{t-1}$, {\em e.g.}, when partitioning the matching task over multiple matchers~\cite{Crowdmap}, we can set a static (global) threshold. Adding $M_{ij}$ is always guaranteed to improve recall and, thus, the strategy targeting recall remains the same, \emph{i.e.,} a static threshold is set to $0$. Note that both precision and \fmspace are bounded from above by $1$. Thus, by setting $E(P(\sigma_{t-1}))$ and $E(F(\sigma_{t-1}))$ to their upper bound ($1$), we obtain a global condition to add $M_{ij}$ if $1\leq M_{ij}$ and $0.5\leq M_{ij}$, respectively. To summarize, targeting $R/F/P$ with a static threshold is done by adding a correspondence to a match if its confidence exceeds $0/0.5/1$, respectively. 


Ongoing decisions are not taken into account when setting a static threshold. Recalling Example~\ref{ex:process}, using a static threshold for \fmspace will reject the third decision ($0.5> M_{12} = 0.25$) (unlike the case of using the estimated value of $E(P(\sigma_{t-1}))$), resulting in a lower final \fmspace of $0.67$.

\subsection{\sysNameSpace Architecture Overview}\label{sec:poware}

In Section~\ref{sec:RWHMatching} we demonstrated that human matching may be biased, in which case Eq.~\ref{eq:unbiasedMatchingEstim} cannot be used as is. Therefore, we next present \sysName, 
aiming to calibrate biased matching decisions and to predict the values of $P$ and $F$. \sysNameSpace is a matching algorithm that enriches the representation of each matching decision with cognitive aspects, algorithmic assessment and neural encoding of prior decisions using LSTM. Compensating for (possible) lack of evaluation regarding former decisions, \sysNameSpace repeatedly predicts missing precision and \fmspace values (learned during a training phase) that are used to monitor the decision making process (see Section~\ref{sec:matching_eval}). Finally, to reach a complete match and boost recall, \sysNameSpace uses algorithmic matching to complete missing values that were not inserted by human matchers.

\begin{figure}[h]
	\begin{center}
		\includegraphics[width=.75\columnwidth]{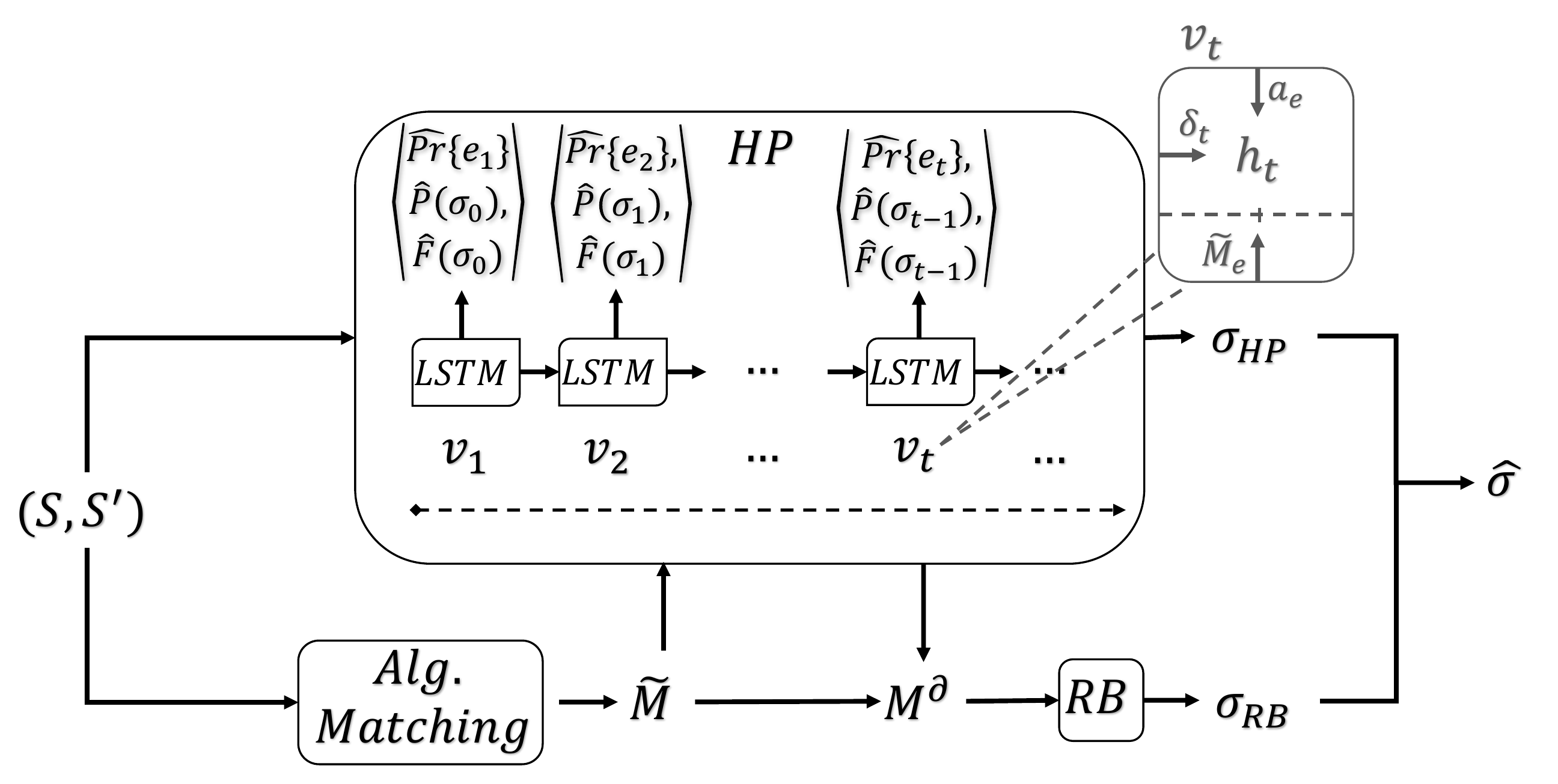}
	\end{center}
	\caption{\sysNameSpace Framework.}
	\label{fig:framework}
\end{figure}

The flow of \sysNameSpace is illustrated in Figure~\ref{fig:framework}. Its input is a schema pair $(S, S^{\prime})$ and a decision history $H$ (Definition~\ref{def:history}) and its output is a match $\hat{\sigma}$. \sysNameSpace is composed of two components, history processing ($HP$), aiming to calibrate matching decisions (Section~\ref{sec:unbias}) and recall boosting ($RB$), focusing on improving the low recall intrinsically obtained by human matchers (Section~\ref{sec:output}).

\subsection{Calibrating Matching Decisions using History Processing ($HP$)}\label{sec:unbias}

The main component of \sysNameSpace calibrates matching decisions by history processing ($HP$). We recall again that the matching history (Definition~\ref{def:history}) records matching decisions of human matchers interacting with a pair of schemata, according to the order in which they were taken. In what follows, $HP$ uses the history to process the matching decisions in the \emph{order} they were assigned by the human matcher. The goal of $HP$ is to improve the estimation of $Pr\{M_{ij}\in\sigma^*\}$ of common (biased) matching beyond using the confidence value, which is an accurate assessment only for unbiased matching. We use $Pr\{e_{t}\}$ here as a shorthand writing for the probability of an element pair assigned at time $t$ to be correct ($Pr\{h_t.e\in\sigma^*\}$).

We first propose a feature representation of matching history decisions (Section~\ref{sec:features}) to be processed using a recurrent neural network that is trained in a supervised manner to capture latent temporal properties of the matching process. Once trained, the network predicts a set of labels $\langle \hat{Pr}\{e_{t}\}, \hat{P}(\sigma_{t-1}), \hat{F}(\sigma_{t-1})\rangle$ regarding each decision $h_{t}$, which is used to generate a match $\sigma_{HP}$ (Section~\ref{sec:learn}).

\subsubsection{Turning Biases into Features}\label{sec:features}

A feature encoding of a human matcher decision $h_{t}$ uses a $4$-dimensional feature vector, composed of the reported confidence, allocated time, consensual agreement, and an algorithmic similarity score. Allocated time and consensual agreement (along with control, the extent of which the matcher was assisted by an algorithmic solution, which we do not use here since it was shown to be less predictive in our experiments) are human biases studied by Ackerman {\em et al.}~\cite{ackerman2019cognitive} (see~Appendix~\ref{app:biases} for more details). Allocated time is measured directly using history timestamps. \add{An $n\times m$ consensual agreement matrix $A$ is constructed using the decisions of all available human matchers (in our experiments, those in the training set), such that $a_{ij}\in A$ is calculated as the number of matchers that determine $a_i\in{S}$ and $b_j\in{S^{\prime}}$ to correspond, \emph{i.e.,} $\exists h_t\in H: h_t.e = \left(a_i, b_j\right)$}. An algorithmic matching result is given as an $n\times m$ similarity matrix $\tilde{M}$.

Let $h_{t}\in H$ be a matching decision at time $t$ (Definition~\ref{def:history}) regarding entry $h_{t}.e=\left(a_i, b_j\right)$. We create a feature encoding $v_{t}\in {\rm I\!R}^4$ given by $v_{t} = \langle h_t.c, \delta_{t}, a_{e}, \tilde{M}_{e} \rangle$, where 
\begin{itemize}
	\item $h_t.c$ is the confidence value associated with $h_{t}$,
	\item $\delta_{t} = h_{t}.t - h_{t-1}.t$ is the time spent until determining $h_{t}$,
	\item $a_{e} = A_{ij}$ is the consensus regarding the entry assigned in the decision $h_{t}$,
	\item $\tilde{M}_{e} = \tilde{M}_{ij}$, an algorithmic similarity score regarding $\left(a_i, b_j\right)$.
\end{itemize}  

$v_{t}$ enriches the reported confidence ($h_t.c$) with additional properties of a matching decision including possible biases and an 
algorithmic opinion. A simple solution, which we analyze as a baseline in our experiments ($ML$, Section~\ref{sec:collbModeleval}) applies an out-of-the-box machine learning method to predict $\langle \hat{Pr}\{e_{t}\}, \hat{P}(\sigma_{t-1}), \hat{F}(\sigma_{t-1})\rangle$. \add{To encode the sequential nature of a (human) matching process, we explore next the use of recurrent neural networks (LSTM), assuming a relation between prior decisions and the current decision.}

\subsubsection{History Processing with Biased Matching}\label{sec:learn}

Using the initial decision encoding $v_{t}$, we now turn our effort to improve biased matching decisions. 
To this end, we train a neural network to estimate $Pr\{M_{ij}\in\sigma^*\}$, $E(P(\sigma_{t-1}))$ and $E(F(\sigma_{t-1}))$ instead of using Eq.~\ref{eq:unbiasedMatchingEstim} (Section~\ref{sec:unbiasedprocessing}). Therefore, the $HP$ component consists of three (supervised) neural models, a classifier predicting $\hat{Pr}\{e_{t}\}$ and two regressors predicting $\hat{P}(\sigma_{t-1})$ and $\hat{F}(\sigma_{t-1})$ (see illustration of $HP$ component in Figure~\ref{fig:framework}).   

\add{With deep learning~\cite{lecun2015deep}, it is common to expect a large training dataset. While human matching often offers scarce training data, the sequential decision making process (formalized as a decision history), which is unique to human (as opposed to algorithmic) matchers, offers an opportunity to collect a substantial size training set for training recurrent neural networks, and specifically long short-term memory (LSTMs). LSTM is a natural processing tool of sequential data~\cite{gers2000learning}, using a gating scheme of hidden states to control the amount of information to preserve at each timestamp. Using LSTM, we train a model that, when assessing a matching decision $h_{t}$, takes advantage of matcher's previous decisions ($\{h_{t^{\prime}}\in H| t^{\prime} < t\}$).}

Applying LSTM for a decision vector $v_{t}$ yields a recurrent representation of the decision, after which, we apply $tanh$ activation to obtain a hidden representation $h^{LSTM}_t$. Intuitively, $h^{LSTM}_t$ provides a process-aware representation over $v_{t}$, encoding the decision $h_t$. 

After the LSTM utilization, the network is divided into the three tasks by applying a linear fully connected layer over $h^{LSTM}_t$, reducing the dimension to the label dimension ($2$ for the classier and $1$ for the regressors) and producing $h^{Pr}_{t}$, $h^{P}_{t}$ and $h^{F}_{t}$ corresponding to the task of estimating the probability of the $t$'th decision to be correct (included in the reference match) and the task of predicting the current match precision and \fmspace values, respectively. Note that since $h^{LSTM}_t$ latently encodes the timestamps $\{t^{\prime}| t^{\prime} < t\}$, it also assists in predicting $\sigma_{t-1}$. Then, the classifier applies a $softmax(\cdot)$ function and the regressors apply a $sigmoid(\cdot)$. The softmax function returns a two-dimensional probabilities vector, \emph{i.e.,} the $i$'th entry represents a probability that the classification is $i$, for which we take the entry representing $1$ as $\hat{Pr}\{e_{t}\}$. At each time $t$ we obtain $\hat{Pr}\{e_{t}\}$, $\hat{P}(\sigma_{t-1}) = sigmoid(h^{P}_{t})$ and $\hat{F}(\sigma_{t-1}) = sigmoid(h^{F}_{t})$ whose labels during training are ${\rm I\!I}_{\{\{M_{ij}: e_{t} = \left(a_i, b_j\right)\}\in\sigma^*\}}$ (an indicator stating whether the decision is a part of the reference match, see Section~\ref{sec:matchAnnealing}), $P(\{M_{ij}|e_{t^{\prime}} = \left(a_i, b_j\right) \wedge h_{t^{\prime}}\in H \wedge t^{\prime} < t\})$ and $F(\{M_{ij}|e_{t^{\prime}} = \left(a_i, b_j\right) \wedge h_{t^{\prime}}\in H \wedge t^{\prime} < t\})$ (computed according to Eq.~\ref{eq:PandR} over a match composed of prior decisions), respectively.

\begin{wrapfigure}[8]{R}{0.5\textwidth}
	\vspace*{-.4cm}
	\begin{minipage}{0.5\textwidth}
		\begin{table}[H]
			\caption{History processing by target and threshold.}
			\label{tab:HPthresholds}
			\centering 
			\begin{tabular}{l c c c}
				\multicolumn{1}{r}{\textbf{target}$\rightarrow$} & \multirow{2}{*}{$R$}  & \multirow{2}{*}{$P$} & \multirow{2}{*}{$F$}\\
				$\downarrow$\textbf{threshold} & & & \\
				\toprule
				dynamic & $0.0$ & $\hat{P}(\sigma_{t-1})$ & $0.5\cdot \hat{F}(\sigma_{t-1})$\\
				static & $0.0$ & $1.0$ & $0.5$ \\
			\end{tabular}
		\end{table}
	\end{minipage}
\end{wrapfigure}

Once trained, $H$ processing is similar to the one described in Section~\ref{sec:unbiasedprocessing}. Let $h_t\in H$ be a matching decision made at time $t$ and $\sigma_{t-1}$ the match {\bf generated until} time $t$. When targeting recall we accept every decision, to produce $\sigma_{HP}$.
Targeting precision (\fm), we predict $\hat{Pr}\{e_{t}\}$ and $\hat{P}(\sigma_{t-1})$ ($\hat{F}(\sigma_{t-1})$) and add $\{M_{ij}| e_{t} = \left(a_i, b_j\right)\}$ to the match ($\sigma_{t} = \sigma_{t-1}\cup\{M_{ij}\}$) if $\hat{P}(\sigma_{t-1})\leq \hat{Pr}\{e_{t}\}$ ($0.5\cdot\hat{F}(\sigma_{t-1})\leq \hat{Pr}\{e_{t}\}$). A static threshold may also be applied here (see Section~\ref{sec:glob}). Table~\ref{tab:HPthresholds} summarizes the discussion above, specifying the conditions needed to accept a correspondence to a match. Finally, $HP$ returns the match $\sigma_{HP} = \sigma_{T}$.

\subsection{Recall Boosting}\label{sec:output}
The $HP$ component improves on the precision of the human matcher and is limited to correspondences that were chosen by a human matcher. In what follows, the resulted match may have a low recall, simply because the decision history as provided by the human matcher did not contain enough accurate correspondences to choose from. To better cater to recall, we present next a complementary component, $RB$, to boost recall using adaptation of existing 2LMs (see Section~\ref{ex:matchers}).

\add{$RB$ uses algorithmic match results for correspondences that were not assigned by a human matcher, formally stated as $M^{\partial} = \{\tilde{M}_{ij}|\forall h\in H, h.e \not= (a_i,b_j)\}$. In our experiments, we use {\sf ADnEV}~\cite{shraga2020} to produce $M^{\partial}$ (see Section~\ref{sec:meth})}. When a human matcher does not offer an opinion on a correspondence, it may be intentional or caused inadvertently and the two are indistinguishable. $RB$ complements human matching by analyzing all unassigned correspondences.

$RB$ is a general threshold-based 2LM, from which multiple 2LMs can be derived. Unlike traditional 2LMs, which are applied over a whole similarity matrix, $RB$ is applied to $M^{\partial}$ only, letting $HP$ handle human decisions as discussed in Section~\ref{sec:unbias}. Given a partial matrix $M^{\partial}$ and a set of thresholds $\nu = \{\nu_{ij}\}$, $RB$ is defined as follows:
\begin{equation}\label{eq:2LMg1}
	RB(M^{\partial}, \nu) = \{M^{\partial}_{ij}| M^{\partial}_{ij}\geq\nu_{ij}\}
\end{equation}

$RB$ forms a natural implementation of {\sf Threshold}~\cite{DO2002a} by defining $\nu^{th}$ to assign a single threshold value to all entries such that $\nu_{i,j}=\nu^{th}$ for all matrix entries $(i,j)$. 
To adapt {\sf Max-Delta} to work with $M^{\partial}$, we separate the decision-making process by cases. Let $M_{ij}$ be the entry under consideration. In the case that a human matcher did not choose any value in row $i$, $RB$ adds $M_{ij}$ to $\sigma_{RB}$ if $M_{ij}>0$. There are two possible scenarios when a human matcher chose an entry in row $i$. If none of the entries a human matcher chose in row $i$ are included in $\sigma_{HP}$, then we are back to the first case, adding $M_{ij}$ to $\sigma_{RB}$ if $M_{ij}>0$. Otherwise, some entry in row $i$ has been included in $\sigma_{HP}$ and by using a hyper-parameter $\theta$,\footnote{we use $\theta$ here to avoid confusion with a $\delta$ function notation.} $M_{ij}$ is added to $\sigma_{RB}$ if it satisfies the {\sf Max-Delta} revised condition: $M_{ij}>1-\theta$.  

In what follows, each threshold in the set $\nu^{md}(\theta)$ is formally defined as follows:\footnote{$\varepsilon\sim 0$ is a very small number assuring a strict inequality.\label{fn:epsilon}}
\begin{equation}\label{eq:2LMg2}
	\begin{split}
		\nu_{ij}^{md}(\theta) & = \mathtt{I}_{\{\exists j^{\prime}\in \{1,2,\dots,m\}: M_{ij^{\prime}}\in\sigma_{HP}\}}\cdot\varepsilon\\
		& - \theta\cdot(1 - \mathtt{I}_{\{\nexists j^{\prime}\in \{1,2,\dots,m\}:  M_{ij^{\prime}}\in\sigma_{HP}\}})
	\end{split}
\end{equation}

A straightforward variation of {\sf Max-Delta} can be derived by looking at the columns instead of the rows of each entry, as initially suggested by Do \emph{et al.}~\cite{DO2002a}. 

The implementation of such variant of Eq.~\ref{eq:2LMg2} is given by 
\begin{equation}\label{eq:2LMg22}
	\begin{split}
		\nu_{ij}^{md(col)}(\theta) & = \mathtt{I}_{\{\exists i^{\prime}\in \{1,2,\dots,n\}: M_{ij^{\prime}}\in\sigma_{HP}\}}\cdot\varepsilon\\
		& - \theta\cdot(1 - \mathtt{I}_{\{\nexists i^{\prime}\in \{1,2,\dots,n\}:  M_{ij^{\prime}}\in\sigma_{HP}\}})
	\end{split}
\end{equation}

Finally, we offer an extension to the original implementation of {\sf Dominants} by introducing a tuning window (distance from maximal value) similar to the one defined for {\sf Max-Delta}. The inference is similar to the one applied for $\nu^{md}(\theta)$ and each threshold in the set $\nu^{dom}(\theta)$ can be defined as follows:
\begin{equation}\label{eq:2LMg3}
	\begin{split}
		\nu_{ij}^{dom}(\theta) & = \mathtt{I}_{\{\exists i^{\prime}\in \{1,\dots,n\}: M_{ij^{\prime}}\in\sigma_{HP} \lor \exists j^{\prime}\in \{1,\dots,m\}: M_{ij^{\prime}}\in\sigma_{HP}\}}\cdot\varepsilon\\
		& - \theta\cdot(1 - \mathtt{I}_{\{\nexists i^{\prime}\in \{1,\dots,n\}:  M_{ij^{\prime}}\in\sigma_{HP}\wedge \exists j^{\prime}\in \{1,\dots,m\}: M_{ij^{\prime}}\in\sigma_{HP}\}})
	\end{split}
\end{equation}

The $RB$ methods introduced above are all accompanied by a hyper-parameter (a uniform $\nu^{th}$ for {\sf Threshold} and $\theta$ for the {\sf Max-Delta}'s variations and {\sf Dominants}) controlling the threshold. In this work we focus on a uniform threshold, which also yielded the best results in our empirical evaluation (see Section~\ref{sec:rbres}).

The two generated matches $\sigma_{HP}$ and $\sigma_{RB}$, using $HP$ and $RB$, are combined by \sysNameSpace to create the final match $\hat{\sigma} = \sigma_{HP}\cup \sigma_{RB}$.

\section{Empirical Evaluation}
\label{sec:collbModeleval}

In this work we focus on human matching as a decision making process. Accordingly, different from a typical crowdsourcing setting, the human matchers in our experiments are free to choose their own order of matching elements and to decide on which correspondences to report (as illustrated in Figure~\ref{fig:ex1}). Section~\ref{sec:dataset} describes human matching datasets using two matching tasks and Section~\ref{sec:setup} details the experimental setup. Our analysis shows that
\begin{enumerate}
	\item The \emph{matching quality of human matchers is improved} using (a fairly simple) process-aware inference that takes into account self-reported confidence (Section~\ref{sec:hpres}).
	\item \emph{\sysNameSpace effectively calibrates human matching} (Section~\ref{sec:unbiasingres}) to provide decisions of higher quality, which \emph{produce high precision} values (Section~\ref{sec:precision}), \emph{even if confidence is not reported by human matchers} (Section~\ref{sec:ablres}).
	\item \emph{\sysNameSpace efficiently boosts human matching recall} performance to produce \emph{improved overall matching quality} (Section~\ref{sec:rbres}).
	\item \emph{\sysNameSpace generalizes (without training a new model) beyond the domain of schema matching} to the domain of ontology alignment (Section~\ref{sec:mainresOA}).  
\end{enumerate}

\subsection{Human Matching Datasets}
\label{sec:dataset} 

We created two human matching datasets for our experiments, gathered via a controlled experiment. To simulate a real-world setting, the participants (human matchers) in our study were asked to match two schemata (which they had never seen before) by locating correspondences between them using an interface. Participants were briefed in matching prior to the task and were given a time to prepare on a pair of small schemata (taken from the \emph{Thalia} dataset~\cite{HammerST05}) prior to performing the main matching tasks. Participants in our experiments were Science/Engineering undergraduates who studied database management course. The study was approved by the institutional review board and four pilot participants completed the task prior to the study to ensure its coherence and instruction legibility. A subset of the dataset is available at~\cite{data}.\footnote{We intend to make the full datasets public upon acceptance.}

The main matching tasks were chosen from two domains, one of a schema matching task and the other of an ontology alignment task (which is used to demonstrate generalizability, see Section~\ref{sec:mainresOA}). \add{Reference matches for these tasks were manually constructed by domain experts over the years and considered as ground truth for our empirical evaluation.} The schema matching task was taken from the \emph{Purchase Order} (PO) dataset~\cite{DO2002a} with schemata of medium size, having 142 and 46 attributes (6,532 potential correspondences), and with high information content (labels, data types, and instance examples). A total of 7,618 matching decisions from 175 human matchers were gathered for the PO dataset. The ontology alignment~\cite{EUZENAT2007a} task was taken from the OAEI 2011 and 2016 competitions~\cite{oaeiURL}, containing ontologies with 121 and 109 elements (13,189 potential correspondences) with high information content as well. A total of 1,562 matching decisions from 34 human matchers were gathered for the OAEI dataset. Schema matching and ontology alignment offer different challenges, where ontology elements differ in their characteristics from schemata attributes. Element pairs vary in their difficulty level, introducing a mix of both easy and complex matches. 

\begin{figure}[htpb]
	\centering
	\includegraphics[width=\columnwidth]{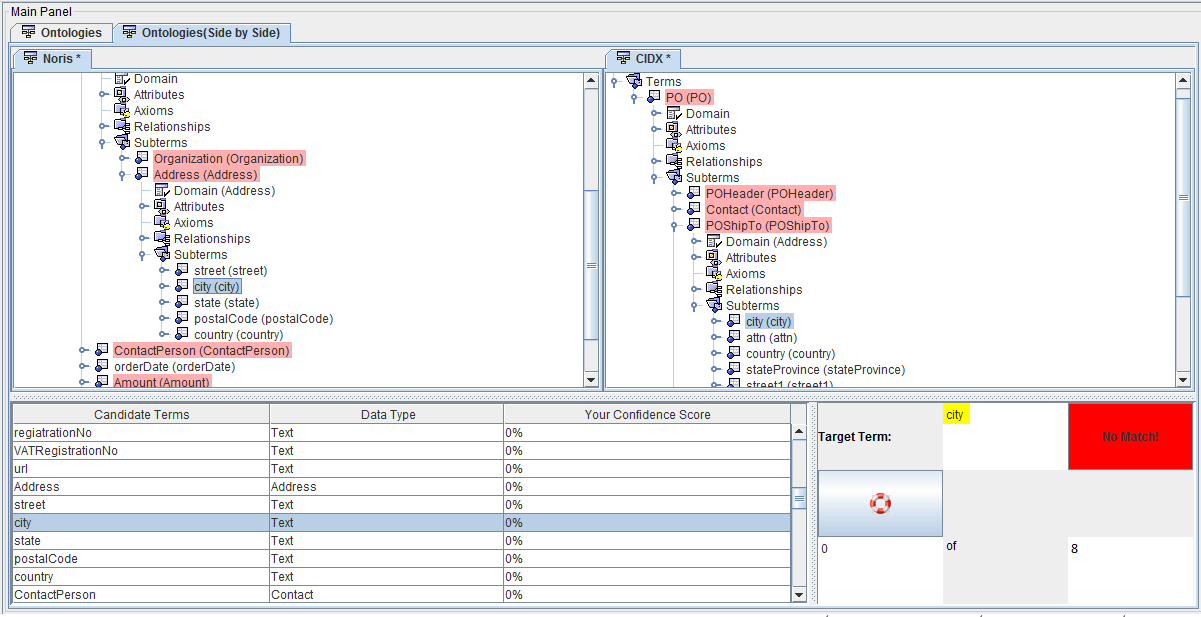}
	\caption{User Interface Example.}
	\label{fig:mipanel1}
\end{figure}

The interface that was used in the experiments is an upgraded version of the Ontobuilder research environment~\cite{MODICA2001}, an open source prototype~\cite{oreURL}. An illustration of the user interface is given in Figure~\ref{fig:mipanel1}. Schemata are presented as foldable trees of terms (attributes). After selecting an attribute from the target schema, the user can choose an attribute from the candidate schema tree. In addition, a list of candidate attributes (from the candidate schema tree), is presented for the user. Selecting an element reveals additional information about it in a properties box. After selecting a pair of elements, match confidence is inserted by participants as a value in $[0,1]$ and timestamped to construct a history. 

\add{The human matchers that participated in the experiments were asked to self-report personal information prior to the experiment. The gathered information includes gender, age, psychometric exam\footnote{\url{https://en.wikipedia.org/wiki/Psychometric_Entrance_Test}} score, English level (scale of 1-5), knowledge in the domain (scale of 1-5) and basic database management education (binary). The human matchers that participated in the experiments reported on average psychometric exam scores that are higher than the general population. While the general population's mean score is 533, participants average is 678. In addition, 88\% of human matchers consider their English level to be at least 4 out of 5 and their knowledge of the domain is low (only 14\% claim their knowledge to be above 1). To sum, the participating human matchers represent academically oriented audience with a proper English level, yet with lack of any significant knowledge in the domain of the task.}

\add{As a final note, we observe a correlation between reported English level and Recall and between reported psychometric exam score and Precision. These results can be justified as better English speakers read faster and can cover more element pairs (Recall) and people that are predicted to have a higher likelihood of academic success at institutions of higher education (higher psychometric score) can be expected to be accurate (Precision). It is noteworthy that these are the only significant correlations found with personal information. This, in turn, emphasizes the importance of understanding the behavior of humans when seeking quality matching, even when personal information is readily available.}

\subsection{Experimental Setup}
\label{sec:setup}

Evaluation was performed on a server with $2$ Nvidia GTX 2080 Ti and a CentOS 6.4 operating system. Networks were implemented using PyTorch~\cite{torchURL} and the code repository is available online~\cite{gitURL}. The $HP$ component of \sysNameSpace was implemented according to Section~\ref{sec:learn}, using an LSTM hidden layer of $64$ nodes and a $128$ nodes fully connected layer. We used Adam optimizer with default configuration ($\beta_1 = 0.9$ and $\beta_2 = 0.999$) with cross entropy and mean squared error loss functions for classifiers ($\hat{Pr}\{e_{t}\}$) and regressors ($\hat{P}(\sigma_{t-1})$ and $\hat{F}(\sigma_{t-1})$), respectively, during training. 

\subsubsection{Evaluation Measures}\label{sec:measures}

We evaluate the matching quality using precision ($P$), recall ($R$), and \fmspace ($F$) (see Section~\ref{sec:smodel}). In addition, we introduce four measures to account for biased matching (Section~\ref{sec:matchingBiases}), which are used in Section~\ref{sec:unbiasingres}. Specifically, we use two correlation measures and two error measures to compute the bias between the estimated values and the true values.
For the former we use Pearson correlation coefficient ($r$) measuring linear correlation and Kendall rank correlation coefficient ($\tau$) measuring the ordinal association. When measuring the bias of $M_{ij}$, the $r$ and $\tau$ values for a match $\sigma$ are given by:
\begin{equation}
	\label{eq:corr_r}
	r = \frac{\sum_{M_{ij}\in\sigma} (M_{ij} - \bar{\sigma})\cdot ({\rm I\!I}_{\{M_{ij}\in\sigma^*\}} - P(\sigma))}{\sqrt{\sum_{M_{ij}\in\sigma} (M_{ij} - \bar{\sigma})^{2}}\cdot \sqrt{\sum_{M_{ij}\in\sigma}({\rm I\!I}_{\{M_{ij}\in\sigma^*\}} - P(\sigma))^{2}}}
\end{equation}\noindent where $\bar{\sigma} = \frac{\sum_{M_{ij}\in\sigma} M_{ij}}{|\sigma|}$ represents the average of a match $\sigma$ and $P(\sigma)$ is the precision of $\sigma$ (see Eq.~\ref{eq:PandR}). 

\begin{equation}
	\label{eq:corr_tau}
	\tau = \frac{C-D}{C+D}
\end{equation}\noindent where $C$ and $D$ represent the number of concordant and discordant pairs.

When measuring the bias of $P$ and $F$ (for convenience we use $G$ in the formulas), $r$ is given by:
\begin{equation}
	\label{eq:corr_r2}
	r = \frac{\sum_{k=1}^{K} (\hat{G}(\sigma_{k}) - \bar{\hat{G}}(\sigma))\cdot (G(\sigma_{k}) - \bar{G}(\sigma_{k}))}{\sqrt{\sum_{k=1}^{K} (\hat{G}(\sigma_{k}) - \bar{\hat{G}}(\sigma))^{2}}\cdot \sqrt{\sum_{k=1}^{K}(G(\sigma_{k}) - \bar{G}(\sigma_{k}))^{2}}}
\end{equation}\noindent where $\hat{G}$ is the estimated value of $G$ and $\bar{G}(\sigma) = \frac{\sum_{k=1}^{K} G(\sigma_{k})}{K}$ is the average of $G$.

For the latter, we use root mean squared error (RMSE) and mean absolute error (MAE), computed as follows (for the match $\sigma$):

\begin{equation}
	\label{eq:error}
	RMSE = \sqrt{\frac{1}{|\sigma|}\sum_{M_{ij}\in\sigma} ({\rm I\!I}_{\{M_{ij}\in\sigma^*\}} - M_{ij})^{2}}, MAE = \frac{1}{|\sigma|}\sum_{M_{ij}\in\sigma} |{\rm I\!I}_{\{M_{ij}\in\sigma^*\}} - M_{ij}|
\end{equation} 
and for an evaluation measure $G$, as follows:
\begin{equation}
	\label{eq:error_2}
	RMSE = \sqrt{\frac{1}{K}\sum_{k=1}^{K} (G(\sigma_{k}) - \hat{G}(\sigma_{k}))^{2}}, MAE =\frac{1}{K}\sum_{k=1}^{K} |G(\sigma_{k}) - \hat{G}(\sigma_{k})|
\end{equation} 

\subsubsection{Methodology}\label{sec:meth}

\begin{sloppypar}

	Sections~\ref{sec:hpres}-\ref{sec:mainresOA} provide an analysis of \sysName's performance. We analyze the ability of \sysNameSpace to improve on decisions taken by human matchers (Section~\ref{sec:hpres}) and the overall final matching (Section~\ref{sec:rbres}). We further analyze \sysName's generalizability to the domain of ontology alignment (Section~\ref{sec:mainresOA}). The experiments were conducted as follows: 
	
	\vspace{0.05in}
	\noindent
	{\bf Human Matching Improvement (Section~\ref{sec:hpres}):}  Using 5-fold cross validation over the human schema matching dataset (PO task, see Section~\ref{sec:dataset}), we randomly split the data into $5$ folds and repeat an experiment $5$ times with $4$ folds for training ($140$ matchers) and the remaining fold ($35$ matchers) for testing. We report on average performance over all human matchers from the $5$ experiments. Matches for each human matcher are created according to Section~\ref{sec:unbias}. We report on \sysName's ability to calibrate biased matching (Section~\ref{sec:unbiasingres}). An ablation study is reported in Section~\ref{sec:ablres}, for which we trained and tested 6 additional $HP$ implementations using a $5$-fold cross validation as before. In this analysis we explore the feature representation of a matching decision (see Section~\ref{sec:features}) by either solely using or discarding 1) confidence ($h_t.c$), 2) cognitive aspects ($\delta_{t}$ and $a_{e}$), or 3) algorithmic input ($\tilde{M}_{e}$).
\end{sloppypar}

\vspace{0.05in}

\noindent
\textbf{\sysName's Overall Performance (Section~\ref{sec:rbres}):} 
We assess the overall performance of \sysName by adding the $RB$ component. In Section~\ref{sec:skills}, we also report on two subgroups of human matchers representing top 10\% (Top-10) and bottom 10\% (Bottom-10) performing human matchers. \add{We selected the RB threshold (in $[0.0, 0.05, \dots, 1.0]$) that yielded the best results during \textbf{training} and analyze this selection in Section~\ref{sec:threshold}.}

\vspace{0.05in}

\noindent{\bf Generalizing to Ontology Alignment (Section~\ref{sec:mainresOA}):} In our final analysis we aim to demonstrate the generalizability of \sysNameSpace in practice. To do so, we assume that we have at our disposal the set of human schema matchers that performed the PO task and we aim to improve on the performance of (different) human matchers on a new (similar) task of ontology alignment over the OAEI task. Specifically, we use the $175$ schema matchers as our training set and generate a \sysNameSpace model. The generated model is than tested on the $34$ ontology matchers.

\vspace{0.05in}

\begin{figure*}[t]
	\begin{subfigure}{.32\linewidth}
		\centering
		\includegraphics[width=\linewidth]{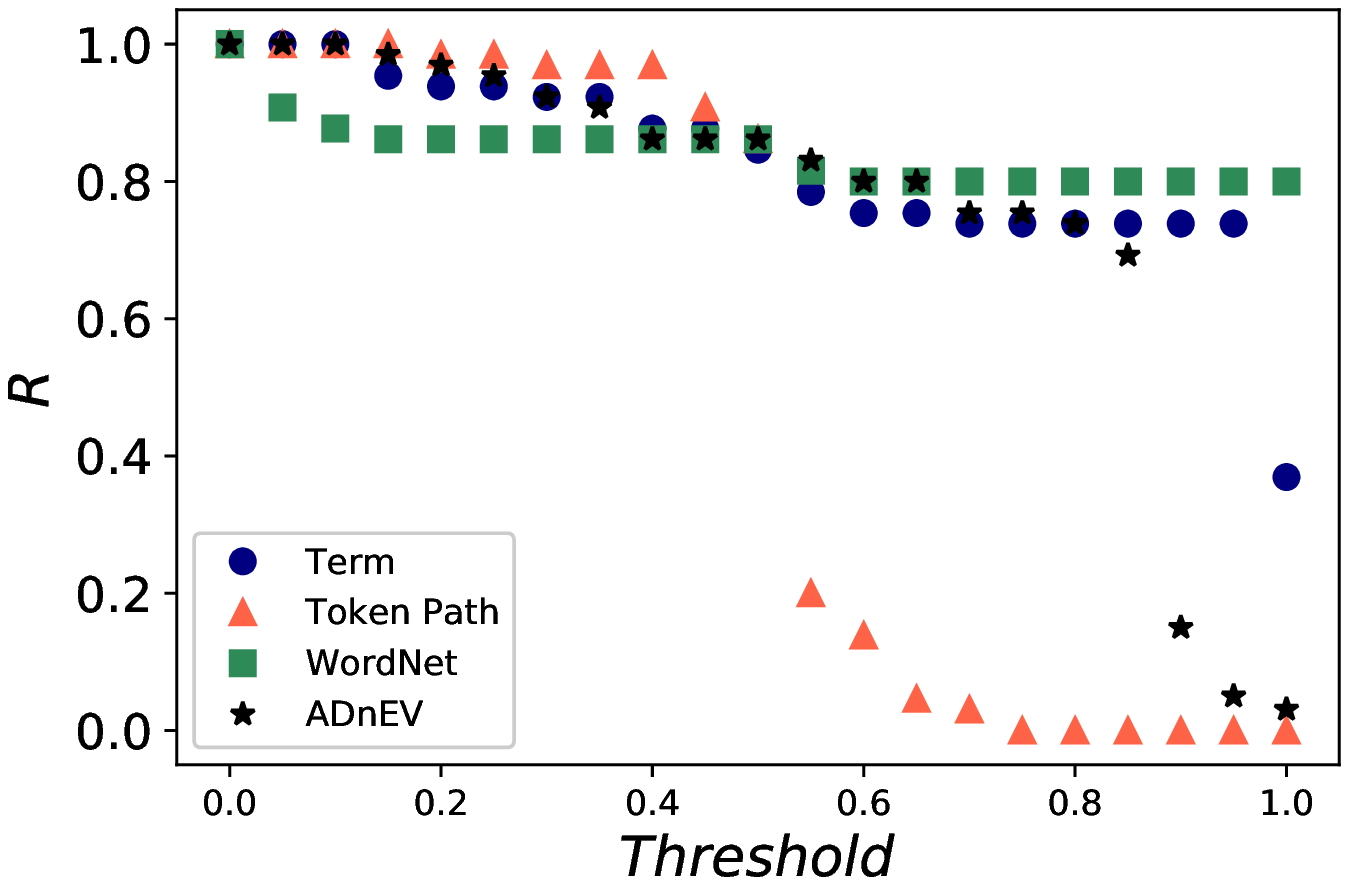}
	\end{subfigure}
	\begin{subfigure}{.32\linewidth}
		\centering
		\includegraphics[width=\linewidth]{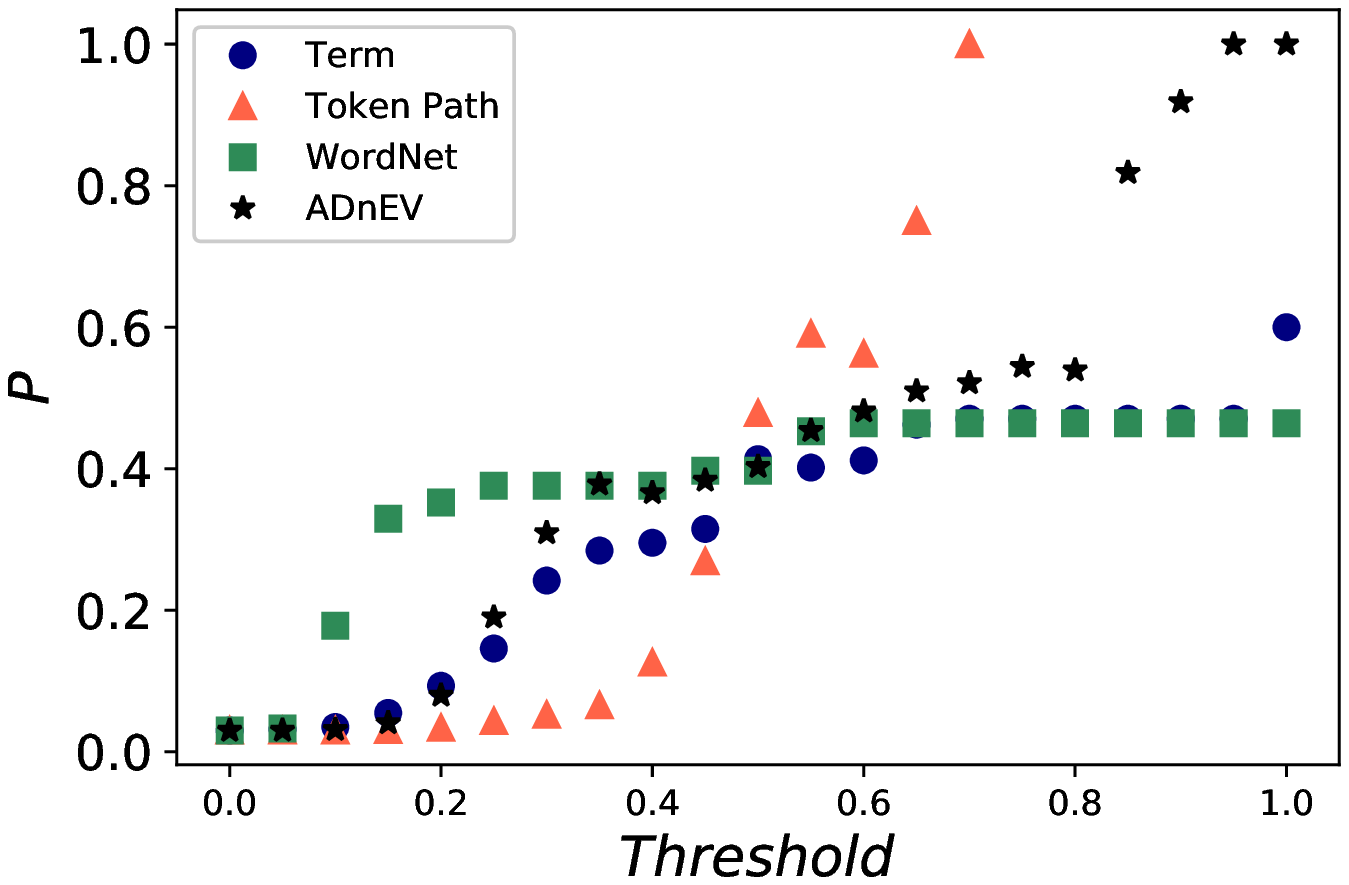}
	\end{subfigure}
	\begin{subfigure}{.32\linewidth}
		\centering
		\includegraphics[width=\linewidth]{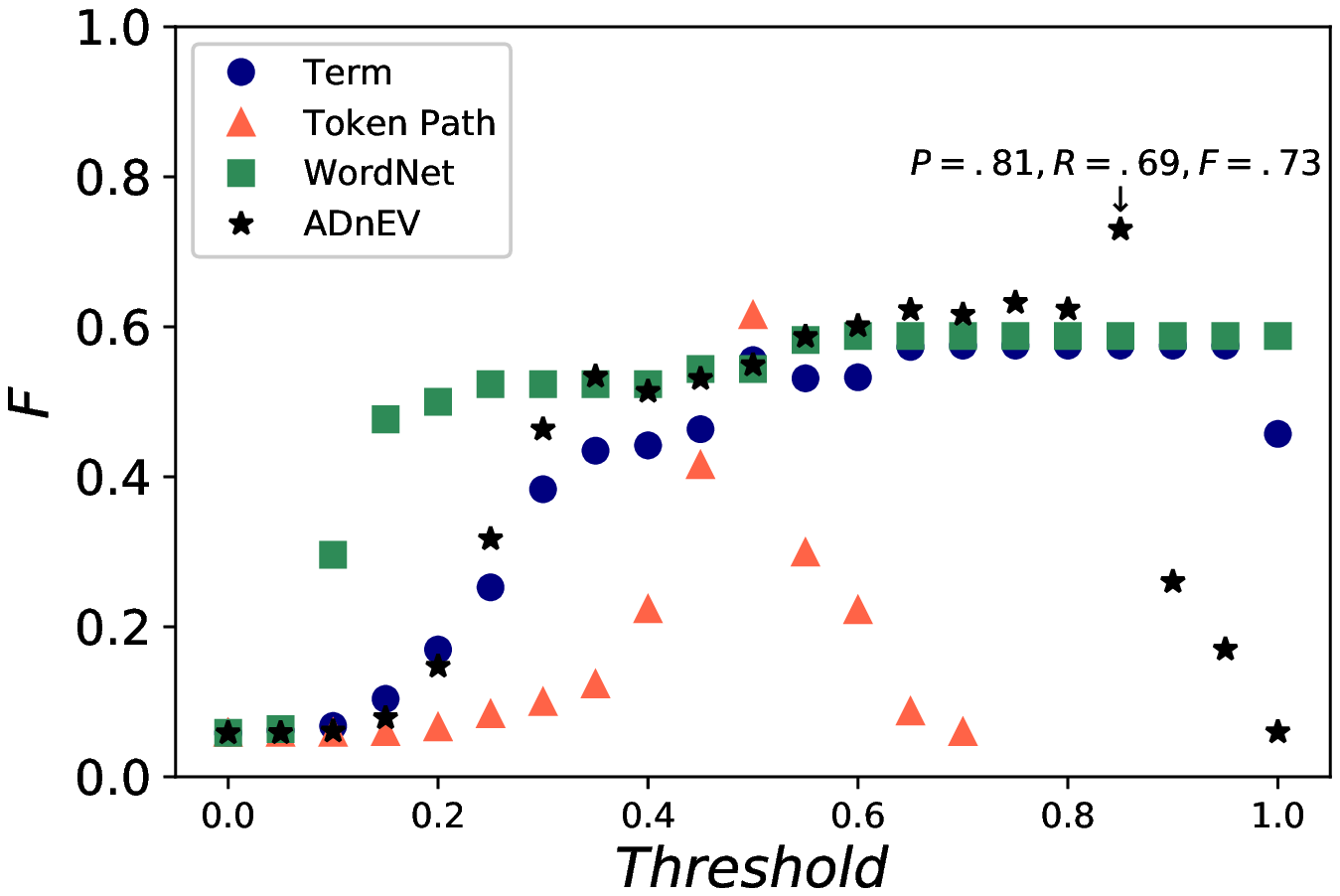}
	\end{subfigure}
	\centering
	\caption{Precision ($P$), recall ($R$), and the \fmspace ($F$) of the algorithmic matchers used in the experiments.}
	\label{fig:alg}
\end{figure*}

We produce $\tilde{M}$ (an algorithmic match, see Section~\ref{sec:model}) using the matchers presented in Section~\ref{ex:matchers} and {\sf ADnEV}, a state-of-the-art aggregated matcher~\cite{shraga2020}. A comparison of the algorithmic matchers performance over several thresholds in terms of precision, recall and F1 measure is given in Figure~\ref{fig:alg}. The comparison shows the superiority of {\sf ADnEV} in high threshold levels. Thus, we present the results of \sysNameSpace using {\sf ADnEV} for recall boosting.

Statistically significant differences in performance are tested using a paired two-tailed t-test with Bonferroni correction for 95\% confidence level, and marked with an asterisk.

\subsubsection{Baselines}\label{sec:base} 

Two types of baselines were used in the experiments we conducted, as follows:
\vspace{0.05in}

\noindent\textbf{Human analysis and improvement (Section~\ref{sec:hpres}):} First, when \sysNameSpace targets recall, it accepts all human judgments (see Theorems~\ref{thm:MIEM} and~\ref{thm:probLocalAnneal}). This also represents traditional methods using human input as ground truth (in this work referred to as ``unbiased matching assumption''). We use two additional types of baselines to evaluate \sysNameSpace ability to improve human matching: 

\begin{enumerate}
	\item  $\mathbf{raw}$: human confidence with threshold filtering that follows Section~\ref{sec:unbiasedprocessing}. For example, targeting $F$ with static threshold ($0.5$, see Table~\ref{tab:HPthresholds}) represents a likelihood-based baseline that accepts decisions assigned with a confidence level greater than $0.5$.
	\item $\mathbf{ML}$: non process-aware machine learning. We experimented with several common classifiers (\emph{e.g.,} SVM) and regressors (\emph{e.g.,} Lasso), and selected the top performing one during training.\footnote{see the full list of $ML$ models (including their configuration) in~\cite{gitConfig}.} The selected machine learning models were then used to replace the neural process-aware classifiers and regressors in the $HP$ component of \sysName.
\end{enumerate}

\vspace{0.05in}

\noindent\textbf{Matching improvement (Sections~\ref{sec:rbres}-\ref{sec:mainresOA}):} We use four algorithmic matching baselines, namely, {\sf Term}, {\sf WordNet}, {\sf Token Path}, and {\sf ADnEV}, the state-of-the-art deep learning algorithmic matching~\cite{shraga2020}. For each, we applied several thresholds (see Figure~\ref{fig:alg}) during training and report on the top performing threshold. Since {\sf ADnEV} shows the best overall performance ($F=0.73$), we combine it with human matching ($raw$) to create a human-algorithm baseline ($raw$-{\sf ADnEV}).

\subsection{Improving Human Matching}
\label{sec:hpres}

We analyze the ability of \sysName's $HP$ component to improve the precision of human matching decisions (Section~\ref{sec:precision}) and to calibrate human confidence, potentially achieving unbiased matching (Section~\ref{sec:unbiasingres}). We then provide feature analysis via an ablation study in Section~\ref{sec:ablres}. 

\subsubsection{Precision}
\label{sec:precision}

Table~\ref{tab:mainhp} provides a comparison of results in terms of precision ($P$) and its computational components ($|\sigma\cap \sigma^{*}|$ and $|\sigma|$, see Eq.~\ref{eq:PandR}) when applying \sysNameSpace to two baselines, namely $raw$ (assuming unbiased matching) and $ML$ (applying non-process aware learning), see details in Section~\ref{sec:base}. We split the comparison by target measure (see Section~\ref{sec:model}), namely targeting 1) recall ($R$),\footnote{Since recall is always monotonic (Theorem~\ref{thm:MIEM}), a dominating strategy adds all correspondences (first row in Table~\ref{tab:mainhp})\label{fn:recall}.} 2) precision ($P$) with static ($1.0$) and dynamic ($\hat{P}(\sigma_{t-1})$) thresholds, and 3) \fmspace ($F$) with static ($0.5$) and dynamic ($0.5\cdot\hat{F}(\sigma_{t-1})$) thresholds (as illustrated in Table~\ref{tab:HPthresholds}). The best results within each target measure (+threshold) are marked in bold.

\begin{table}[t]
	\caption{True positive ($|\sigma\cap \sigma^{*}|$), match size ($|\sigma|$), and precision ($P$) by target measure, applying history processing ($HP$)}
	\label{tab:mainhp}
	\scalebox{1}{\begin{tabular}{|l|l|l|c|c|c|}
			\hline
			Target         & (threshold) & $HP$        & $\mid\sigma\cap \sigma^{*}\mid$              & $\mid\sigma\mid$              & $P$               \\ \hline
			$R$ & $0.0$ & -        & 19.02 & 43.53 & 0.549           \\ \toprule[0.15em]
			\multirow{6}{*}{$P$}       & \multirow{3}{*}{$1.0$} & $raw$  & 10.79          & 19.91          & 0.656           \\ 
			&& $ML$     & 14.01*         & 14.02          & 0.999*          \\ 
			&& \cellcolor{Gray}\sysName    & \cellcolor{Gray}\textbf{15.80*}         & \cellcolor{Gray}15.81         & \cellcolor{Gray}\textbf{0.999*} \\\cline{2-5}  
			&\multirow{3}{*}{$\hat{P}(\sigma_{t-1})$} & $raw$  & 11.34          & 21.69          & 0.612           \\ 
			&& $ML$     & 16.21*         & 18.63          & 0.858*          \\ 
			&& \cellcolor{Gray}\sysName    & \cellcolor{Gray}\textbf{16.50*}         & \cellcolor{Gray}16.62          & \cellcolor{Gray}\textbf{0.987*}          \\ \toprule[0.15em]
			\multirow{6}{*}{$F$}       & \multirow{3}{*}{$0.5$} & $raw$  & \textbf{18.19}          & 36.90          & 0.574           \\ 
			&& $ML$      & 16.16          & 17.97          & 0.890*          \\ 
			&& \cellcolor{Gray}\sysName     & \cellcolor{Gray}16.65          & \cellcolor{Gray}16.73          & \cellcolor{Gray}\textbf{0.993*}          \\\cline{2-5}
			&\multirow{3}{*}{$0.5\cdot\hat{F}(\sigma_{t-1})$} & $raw$ & \textbf{18.05}          & 39.65          & 0.552           \\ 
			&& $ML$      & 16.95          & 22.16          & 0.759           \\ 
			&& \cellcolor{Gray}\sysName     & \cellcolor{Gray}16.97          & \cellcolor{Gray}17.32          & \cellcolor{Gray}\textbf{0.968*}          \\ \hline
	\end{tabular}}
\end{table}

The first row of Table~\ref{tab:mainhp} represents the decision making when targeting recall, \emph{i.e.,} accepting {\bf all} human decisions.$^{\ref{fn:recall}}$ A clear benefit of treating matching \emph{as a process} is observed when comparing this na\"ive baseline of accepting human decisions as ground truth (first row of Table~\ref{tab:mainhp}) to a fairly simple process-aware inference (as in Section~\ref{sec:unbiasedprocessing}) using the reported confidence ($raw$) aiming to improve a quality measure of choice (target measure). Empirically, even when assuming unbiased matching ($raw$), as in Section~\ref{sec:unbiasedprocessing}, we achieve average precision improvement of 9\%. 

\sysNameSpace achieves a statistically significant precision improvement over $raw$, with an average improvement of 65\%, targeting both $P$ and $F$ with both static and dynamic thresholds. \sysNameSpace also outperforms the learning-based baseline ($ML$), achieving 13.6\% higher precision on average. Even when the precision is similar, \sysNameSpace generates larger matches ($|\sigma|$), which can improve recall and \fm, as discussed in Section~\ref{sec:rbres}. Since \sysNameSpace performs process-aware learning (LSTM), its improvement over $ML$ supports the use of matching as a process. 

\sysNameSpace achieves highest precision when targeting $P$ with static threshold, forcing the algorithm to accept only decisions for which \sysNameSpace is fully confident ($\hat{Pr}\{e_{t}\} = 1$). Alas, it comes at the cost match size (less than $16$ correspondences on average). This indicates that being conservative is beneficial for precision, yet it has a disadvantage when it comes to recall (and accordingly \fm). Targeting $F$, especially with dynamic thresholds, balances match size and precision. This observation becomes essential when analyzing the full performance of \sysNameSpace (Section~\ref{sec:rbres}).

\begin{figure*}[t]
	
	\begin{subfigure}{.49\linewidth}
		\centering
		\caption{Precision ($P$) vs. True positive ($|\sigma\cap \sigma^{*}|$)}
		\includegraphics[width=\linewidth]{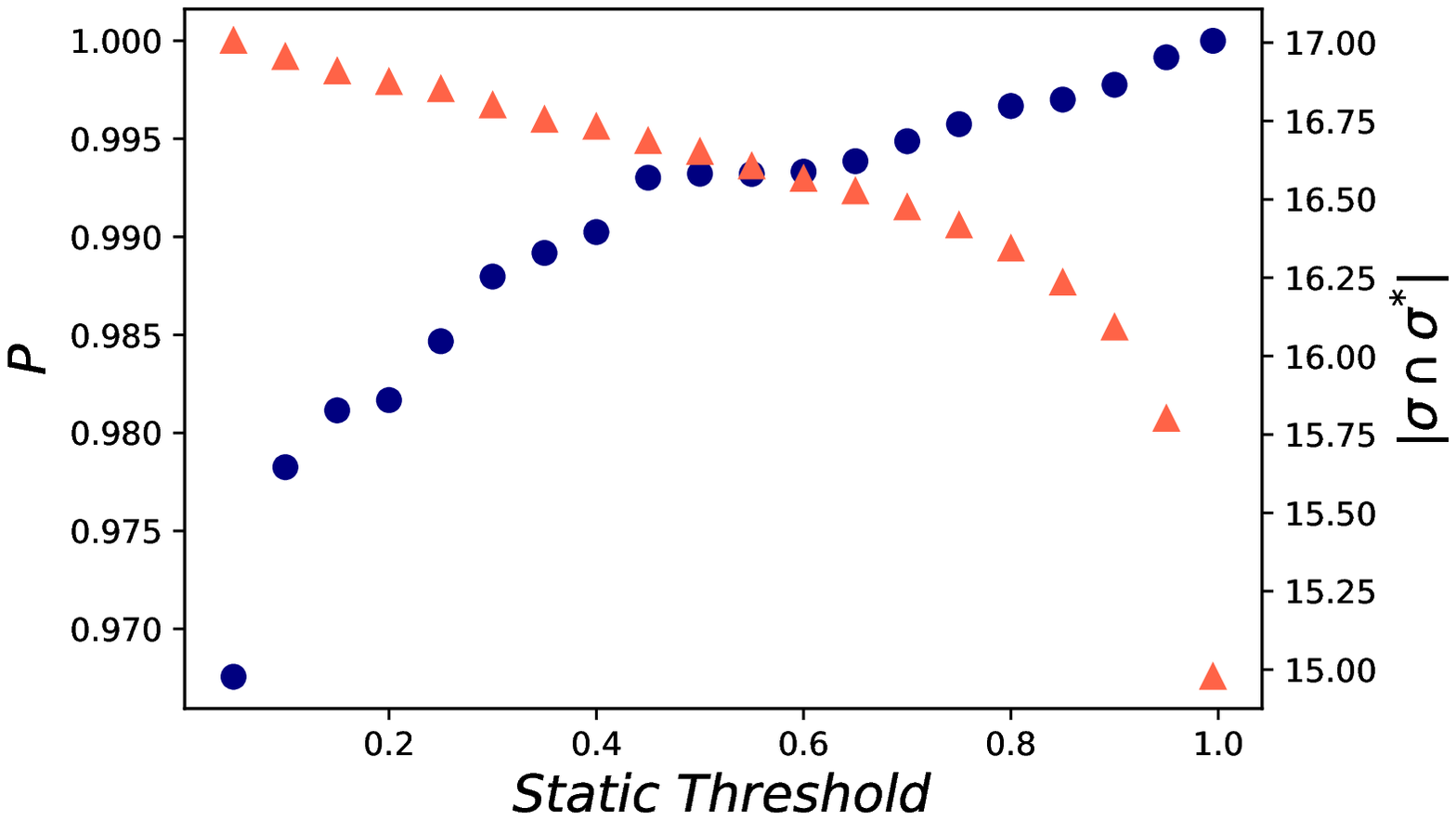}
		\label{fig:th1}
	\end{subfigure}
	\begin{subfigure}{.49\linewidth}
		\centering
		\caption{Match size ($|\sigma|$) vs. True positive ($|\sigma\cap \sigma^{*}|$)}
		\includegraphics[width=\linewidth]{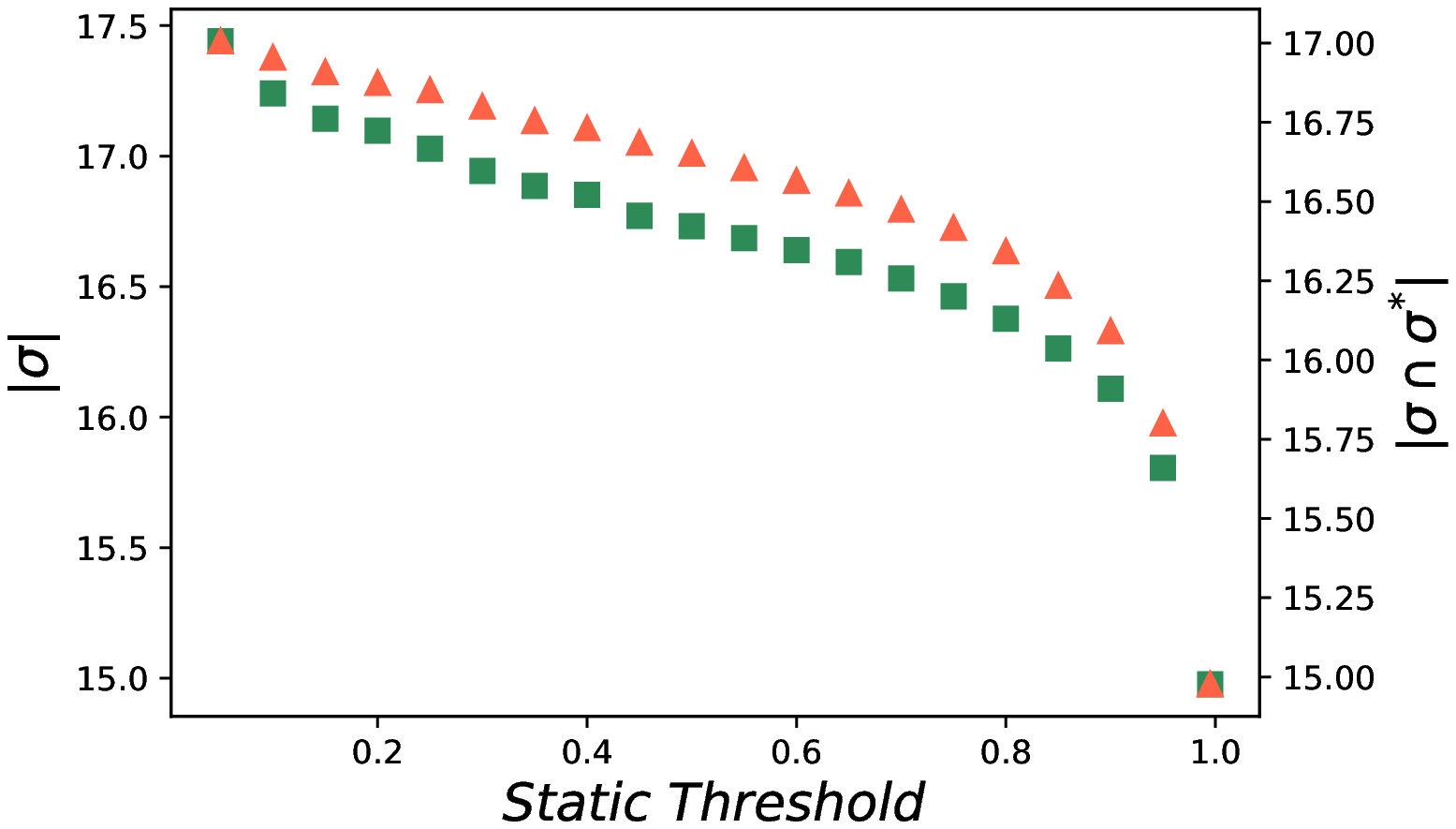}
		\label{fig:th2}
	\end{subfigure}
	\caption{Static threshold analysis}
	\label{fig:hp_th}
\end{figure*}

\add{Figure~\ref{fig:hp_th} analyzes \sysNameSpace precision ($P$) and its computational components ($|\sigma\cap \sigma^{*}|$ and $|\sigma|$) over varying static thresholds. We discard the threshold level of $0.0$ (associated with targeting recall) for readability. As illustrated, as the the static threshold increases, the precision increases while the size of the match and number of true positives decreases. In other words, when using a static threshold the selecting also depends on the target measure. Targeting high precision is associated with high, conservative, thresholds, while targeting large matches (and accordingly recall and \fm) is associated with selecting lower thresholds.}

\subsubsection{Calibration}
\label{sec:unbiasingres}

In the final analysis of the $HP$ component, we quantify its ability to calibrate human confidence to potentially achieve unbiased matching. Table~\ref{tab:calibhp} compares the results of \sysNameSpace to three baselines, namely  $raw$ (assuming unbiased matching), {\sf ADnEV} (representing state-of-the-art algorithmic matching) and $ML$ (applying non-process aware learning), in terms of correlation ($r$ and $\tau$) and error (RMSE and MAE), see Section~\ref{sec:measures}, with respect to the three estimated quantities, $\hat{Pr}\{e_{t}\}$, $\hat{P}(\sigma_{t-1})$ and $\hat{F}(\sigma_{t-1})$. For $raw$ and {\sf ADnEV} the quantities are calculated by Eq.~\ref{eq:unbiasedMatchingEstim} and for $ML$ and \sysNameSpace they are predicted using the $HP$ component. The best results within each quantity are marked in bold.

\begin{table}[h!]
	\caption{Correlation ($r$ and $\tau$) and error (RMSE and MAE) in estimating $\hat{Pr}\{M_{ij}\in\sigma^{*}\}$, $\hat{P}(\sigma_{t-1})$ and $\hat{F}(\sigma_{t-1})$}
	\label{tab:calibhp}
	\begin{tabular}{|l|l|cc|cc|}
		\hline
		Measure               & Method & $r$      & $\tau$       & RMSE          & MAE           \\ \hline
		\multirow{4}{*}{$\hat{Pr}\{M_{ij}\in\sigma^{*}\}$} & $raw$ & 0.29          & 0.26          & 0.59          & 0.45          \\ 
		& {\sf ADnEV}      & 0.27          & 0.22          & 0.24         & 0.21         \\ 
		& $ML$      & 0.76          & 0.62          & 0.30          & 0.13          \\ 
		& \cellcolor{Gray}\sysName     & \cellcolor{Gray}\textbf{0.78} & \cellcolor{Gray}\textbf{0.64} & \cellcolor{Gray}\textbf{0.23} & \cellcolor{Gray}\textbf{0.06} \\ \toprule[0.15em]
		\multirow{4}{*}{$\hat{P}(\sigma_{t-1})$}    & $raw$ & 0.23          & 0.17          & 0.99          & 0.50          \\ 
		& {\sf ADnEV}      & -          & -          & 0.19         & 0.19         \\ 
		& $ML$      & 0.00          & -0.01         & 0.34          & 0.29          \\ 
		& \cellcolor{Gray}\sysName     & \cellcolor{Gray}\textbf{0.90} & \cellcolor{Gray}\textbf{0.72} & \cellcolor{Gray}\textbf{0.13} & \cellcolor{Gray}\textbf{0.10} \\ \toprule[0.15em]
		\multirow{4}{*}{$\hat{F}(\sigma_{t-1})$}    & $raw$ & 0.21          & 0.15          & 0.77          & 0.39          \\ 
		& {\sf ADnEV}      & -          & -         & 0.37         & 0.37         \\ 
		& $ML$      & -0.06         & -0.04         & 0.27          & 0.23          \\ 
		& \cellcolor{Gray}\sysName     & \cellcolor{Gray}\textbf{0.80} & \cellcolor{Gray}\textbf{0.60} & \cellcolor{Gray}\textbf{0.12} & \cellcolor{Gray}\textbf{0.09} \\ \hline
	\end{tabular}
\end{table}

\add{As depicted in Table~\ref{tab:calibhp}, both human and algorithmic matching is biased and \sysNameSpace performs well in calibrating it.} Specifically, \sysNameSpace improves $raw$ human decision confidence correlation with decision correctness by 169\% and 146\% in terms of $r$ and $\tau$, respectively, and lowers the respective error by 0.36 and 0.39 in terms of $RMSE$ and $MAE$, respectively. Algorithmic matching ({\sf ADnEV}) exhibits similar low correlation while reducing the error as well. A possible explanation involves the significant number of non-corresponding element pairs (non-matches). While human matchers avoid assigning confidence values to non-matches (and therefore such element pairs are not included in the error computation), algorithmic matchers assign (very) low similarity scores to non-corresponding element pairs. For example, none of the human matchers selected (and assigned confidence score to) the (incorrect) correspondence between {\sf Contact.e-mail} and {\sf POBillTo.city}, while all algorithmic matchers assigned a similarity score of less than 0.05 to this correspondence. This observation may also serve as an explanation to the proximity between the RMSE and MAE values of {\sf ADnEV}. Finally, when compared to non-process aware learning ($ML$), \sysNameSpace achieves only a slight correlation ($r$ and $\tau$) improvement, yet $ML$'s error values ($RMSE$ and $MAE$) are significantly higher. These higher error values demonstrate that process aware-learning, as applied by \sysName, is better in accurately predicting the probability of a decision in the history to be correct ($\hat{Pr}\{e_{t}\}$) and explains the superiority of \sysNameSpace in providing precise results (Table~\ref{tab:mainhp}).

\subsubsection{Ablation Study}
\label{sec:ablres}

After empirically validating that applying a process-aware learning (as in \sysName) is better than assuming unbiased matching ($raw$) and unordered learning ($ML$), we next analyze the various representations of matching decisions (Section~\ref{sec:features}), namely, 1) \textbf{conf}idence ($h_t.c$), 2) \textbf{cognitive} aspects ($\delta_{t}$ and $a_{e}$), or 3) \textbf{alg}orithmic input ($\tilde{M}_{e}$). Similar to Table~\ref{tab:mainhp}, Table~\ref{tab:ablhp} presents precision ($P$) and its computational components ($|\sigma\cap \sigma^{*}|$ and $|\sigma|$, see Eq.~\ref{eq:PandR}, with respect to a target measure (without recall$^{\ref{fn:recall}}$). Table~\ref{tab:ablhp} compares \sysNameSpace to: 1) using each decision representation element by itself (\emph{only}) and 2) removing it one at a time (\emph{w/o}). Boldface entries indicate the higher importance (for \emph{only} higher quality and for \emph{w/o} lower quality).

\begin{table*}[t!]
	\caption{Decision representation ablation study. {\em only} refers to training using only one decision representation element while {\em w/o} refers to the exclusion of a decision representation element at a time.}
	\label{tab:ablhp}
	\scalebox{.75}{\begin{tabular}{|l|l|c|c|c|c|c|c|c|c|c|c|c|c|}
			\hline
			\multicolumn{2}{|l|}{Target $\rightarrow$}  & \multicolumn{6}{c|}{$P$}  & \multicolumn{6}{c|}{$F$}   \\ 
			\multicolumn{2}{|l|}{(threshold)}  & \multicolumn{3}{c}{$1.0$}  & \multicolumn{3}{c|}{$\hat{P}(\sigma_{t-1})$} & \multicolumn{3}{c}{$0.5$} & \multicolumn{3}{c|}{$0.5\cdot\hat{F}(\sigma_{t-1})$}   \\\hline
			\multicolumn{2}{|l|}{$\downarrow$ Decision Rep.}    & $\mid\sigma\cap \sigma^{*}\mid$ & $\mid\sigma\mid$ & $P$                             & $\mid\sigma\cap \sigma^{*}\mid$ & $\mid\sigma\mid$ & $P$                             & $\mid\sigma\cap \sigma^{*}\mid$ & $\mid\sigma\mid$ & $P$                             & $\mid\sigma\cap \sigma^{*}\mid$ & $\mid\sigma\mid$ & $P$                             \\ \hline
			\multirow{3}{*}{only} & conf      & \textbf{18.78} & 19.94            & 0.94                           & \textbf{19.02} & 20.53            & 0.93                           & \textbf{18.72} & 20.03            & 0.93                           & \textbf{17.92} & 19.53            & 0.92                           \\ \cline{2-14} 
			& cognitive & 15.29                           & 15.80            & \textbf{0.96} & 15.35                           & 16.53            & \textbf{0.93} & 16.11                           & 16.91            & \textbf{0.95} & 16.27                           & 17.31            & \textbf{0.94} \\ \cline{2-14} 
			& alg       & 14.15                           & 16.93            & 0.82                           & 16.86                           & 21.17            & 0.75                           & 17.24                           & 21.71            & 0.76                           & 18.35                           & 24.48            & 0.70                           \\ \toprule[0.15em]
			\multirow{3}{*}{w/o}  & conf      & \textbf{13.36} & 15.76            & 0.85                           & \textbf{16.14} & 20.43            & 0.79                           & \textbf{16.39} & 20.57            & 0.797                           & \textbf{13.53} & 18.23            & 0.74                           \\ \cline{2-14} 
			& cognitive & 14.10                           & 16.86            & \textbf{0.81} & 16.79                           & 20.89            & \textbf{0.77} & 17.05                           & 21.18            & \textbf{0.77} & 18.18                           & 23.66            & \textbf{0.72} \\ \cline{2-14} 
			& alg       & 15.46                           & 15.66            & 0.98                           & 16.42                           & 16.54            & 0.98                           & 16.58                           & 16.66            & 0.98                           & 16.68                           & 17.28            & 0.95                           \\ \hline
	\end{tabular}}
\end{table*}

Examining Table~\ref{tab:ablhp}, we observe that two aspects of the decision representation are predominant, namely confidence and cognitive aspects. Using confidence features only yields the highest proportion of correct correspondences ($|\sigma\cap \sigma^{*}|$) and using only cognitive features offers the best precision values. 
The ablation study shows that even when self-reported confidence is absent from the decision representation, \sysNameSpace can provide precise results using other aspects of the decision. For example, by using only the decision time (the time it took for the matcher to make the decision) and consensus (the amount of matchers who assigned the current correspondence) to represent decisions (only cognitive aspects, second row of Table~\ref{tab:ablhp}), the process-aware learning of \sysNameSpace (targeting $F$) achieves a precision of $0.94$. 

Cognitive aspects and confidence together, without an algorithmic similarity (bottom row of Table~\ref{tab:ablhp}), achieves comparable results to the ones reported in Table~\ref{tab:mainhp}, while eliminating either confidence or cognitive features reduces performance. This observation may indicate that the algorithmic matcher is not as important to correspondences that are assigned by human matchers. 
Recalling that $HP$ is designed to consider only the set of correspondences that were originally assigned by the human matcher during the decision history, in the following section we show the importance of algorithmic results in complementing human decisions.

\subsection{Improving Matching Outcome}
\label{sec:rbres}

\begin{sloppypar}
	We now examine the overall performance of \sysNameSpace and the ability of $RB$ to boost recall. The $RB$ thresholds (see Section~\ref{sec:output}) for \sysName($HP$+$RB$) and $raw$-{\sf ADnEV} were set to the top performing thresholds \textbf{during training} ($0.9$ and $0.85$, respectively). Table~\ref{tab:mainrb} compares, for each target measure, results of \sysNameSpace with and without recall boosting (\sysName($HP$+$RB$) and \sysName($HP$), respectively) and $raw$-{\sf ADnEV} (see Section~\ref{sec:base}). In addition, the four last rows of Table~\ref{tab:mainrb} exhibit the results of algorithmic matchers, for which we present the threshold yielding the best performance in terms of $F$. Best results for each quality measure are marked in bold.
\end{sloppypar}

\begin{table}[t]
	\caption{Precision ($P$), recall ($R$), \fmspace ($F$) of \sysNameSpace by target measure compared to baselines (PO task)}
	\label{tab:mainrb}
	\scalebox{1}{\begin{tabular}{|l|l|l|ccc|}
			\hline
			Target                 & (threshold)     & Method  & $P$              & $R$              & $F$              \\ \hline
			\multirow{3}{*}{$R$}    & \multirow{3}{*}{$0.0$}    & \sysName($HP$)      & 0.549          & 0.293          & 0.380          \\ \cline{3-5} 
			&& \cellcolor{Gray}\sysName($HP$+$RB$)      & \cellcolor{Gray}0.776          & \cellcolor{Gray}0.844          & \cellcolor{Gray}0.797          \\ 
			&& $raw$-{\sf ADnEV}      & 0.731          & 0.807          & 0.756          \\ \hline
			\multirow{6}{*}{$P$} & \multirow{3}{*}{$1.0$} & \sysName($HP$)      & \textbf{0.999} & 0.230          & 0.364          \\ \cline{3-6} 
			&& \cellcolor{Gray}\sysName($HP$+$RB$)        & \cellcolor{Gray}\textbf{0.999*} & \cellcolor{Gray}0.782          & \cellcolor{Gray}0.876*          \\ 
			&& $raw$-{\sf ADnEV}      & 0.824          & 0.680          & 0.729          \\ \cline{2-6}
			& \multirow{3}{*}{$\hat{P}(\sigma_{t-1})$}& \sysName($HP$)      & 0.987          & 0.254          & 0.391          \\ \cline{3-6} 
			&& \cellcolor{Gray}\sysName($HP$+$RB$)        & \cellcolor{Gray}0.998*          & \cellcolor{Gray}0.805*         & \cellcolor{Gray}0.889*          \\ 
			&& $raw$-{\sf ADnEV}    & 0.808          & 0.688          & 0.730          \\ \hline
			\multirow{6}{*}{$F$}  & \multirow{3}{*}{$0.5$} & \sysName($HP$)      & 0.993          & 0.256          & 0.393          \\ \cline{3-6} 
			&& \cellcolor{Gray}\sysName($HP$+$RB$)        & \cellcolor{Gray}0.998*          & \cellcolor{Gray}0.807          & \cellcolor{Gray}0.891*          \\ 
			&& $raw$-{\sf ADnEV}     & 0.754          & 0.794          & 0.764          \\ \cline{2-6} 
			& \multirow{3}{*}{$0.5\cdot\hat{F}(\sigma_{t-1})$} & \sysName($HP$)    & 0.968          & 0.261          & 0.398          \\ \cline{3-6} 
			&& \cellcolor{Gray}\sysName($HP$+$RB$)        & \cellcolor{Gray}0.993*          & \cellcolor{Gray}0.812          & \cellcolor{Gray}\textbf{0.892*} \\ 
			&& $raw$-{\sf ADnEV}     & 0.741          & 0.792          & 0.754          \\ \hline
			\multicolumn{2}{|c|}{} & {\sf ADnEV}~\cite{shraga2020}      & 0.810          & 0.692          & 0.730          \\ \cline{3-6} 
			\multicolumn{2}{|c|}{}& {\sf Term}~\cite{GAL2011}       & 0.471          & 0.738          & 0.575          \\ \cline{3-6} 
			\multicolumn{2}{|c|}{-}& {\sf Token Path}~\cite{PEUKERT2011}       & 0.479          & \textbf{0.862} & 0.615          \\ \cline{3-6} 
			\multicolumn{2}{|c|}{}& {\sf WordNet}~\cite{wordnet1}      & 0.453          & 0.815          & 0.582          \\ \hline
	\end{tabular}}
\end{table}

Evidently, $RB$ improves (mostly in a statistical significance manner) recall and F1 measure over \sysName's $HP$. On average, the recall boosting phase improves recall by 214\% and the F1 measure by 125\%. Compared to the baselines, \sysNameSpace outperforms $raw$-{\sf ADnEV} by 23\%, 8\%, and 17\% in terms of $P$, $R$, and $F$ on average, respectively, and performs better than {\sf ADnEV}, {\sf Term}, {\sf Token Path}, {\sf WordNet} by 19\%, 51\%, 41\%, and 49\% in terms of $F$, on average.

Compared to the baselines, \sysNameSpace outperforms $raw$-{\sf ADnEV} by 23\%, 8\%, and 17\% in terms of $P$, $R$, and $F$ on average, respectively, and performs better than {\sf ADnEV}, {\sf Term}, {\sf Token Path}, {\sf WordNet} by 19\%, 51\%, 41\%, and 49\% in terms of $F$, on average. We next dive into more detailed analysis, namely skill-based performance (Section~\ref{sec:skills}) and precision-recall tradeoff (Section~\ref{sec:threshold}). 

\subsubsection{$RB$'s Effect via Skill-based Analysis}
\label{sec:skills}

We note again that poor recall is a feature of human matching and while raw human matching oftentimes also suffers from low precision, \sysNameSpace can boost the precision to obtain reliable human matching results (Section~\ref{sec:hpres}). We now analyze the $RB$'s ability to boost recall. 

Overall, \sysNameSpace aims to screen low quality \textbf{matching decisions} rather than low quality \textbf{matchers}, acknowledging that low quality matchers can sometimes make valuable decisions while high quality matcher may slip and provide erroneous decisions. To illustrate the effect of \sysNameSpace with respect to the varying abilities of human matchers to perform high quality matching, in addition to all human matchers (All), we also investigate high-quality (Top-10) and low-quality matchers (Bottom-10). 

\begin{wrapfigure}[8]{R}{0.6\textwidth}
	\vspace*{-1.1cm}
	\begin{minipage}{0.6\textwidth}
		\begin{table}[H]
			\caption{History sizes ($|H|$), match sizes ($|\sigma|$), true positive number ($|\sigma\cap \sigma^{*}|$), and Precision ($P$), recall ($R$), and \fmspace ($F$) improvement achieved by $RB$ by matchers subgroup}
			\label{tab:skillrb}
			\scalebox{.825}{\begin{tabular}{|l|c|c|c|c|c|c|}
					\hline
					& $|H|$ & \multicolumn{2}{|c|}{$|\sigma|$ ($|\sigma\cap \sigma^{*}|$)} & \multicolumn{3}{|c|}{$RB$ \% Improvement}\\\hline
					$\downarrow$ Group & & $HP$ & $RB$  & $P$ & $R$ & $F$ \\\hline
					All & 43.5 & 17.3 (17) & 36.2 (35.5) & 2\% & 211\% & 122\%\\\hline 
					Top-10 & 54.5 & 43 (43) & 6.9 (6.9) & 0\% & 16\% & 12\%\\\hline
					Bottom-10 & 26.2 &  2.7 (2.3) & 46.6 (43.2) & 4\% & 3,230\% & 1,900\%\\\hline
			\end{tabular}}
		\end{table}
	\end{minipage}
	\vspace*{.3cm}
\end{wrapfigure}

On average, using Eq.~\ref{eq:hutomat}, all human matchers yielded matches with $P$=$.55$, $R$=$.29$, $F$=$.38$, the Top-10 group matches with $P$=$.91$, $R$=$.74$, $F$=$.81$, and the Bottom-10 group matches with $P$=$.14$, $R$=$.06$, $F$=$.08$. Table~\ref{tab:skillrb} compares the three groups in terms of history size, \emph{i.e.,} average number of human decisions, match size (and number of true positive) of the $HP$ phase, the average number of (correct) correspondences added in the $RB$ phase, and the improvement in terms of $P$, $R$, and $F$ of the recall boosting using the $RB$ component.

$RB$ significantly improves, over all human matchers, recall ($211\%$ on average) and \fmspace ($122\%$ on average) and slightly improves precision. When it comes to low-quality matchers $RB$ has a considerable role (bottom row, Table~\ref{tab:skillrb}) while for high-quality matchers, RB only provides a slight recall boost (middle row, Table~\ref{tab:skillrb}).

\sysNameSpace is judicious even when calibrating the results of the high quality matchers. While on average, $49.8$ of the $54.5$ raw decisions of high-quality human matchers are correct, \sysNameSpace only uses an average of $43$ (correct) correspondences when processing history, omitting, on average, $6.8$ correct correspondences from the final match (recall that $RB$ considers only $M^{\partial}$, see Section~\ref{sec:output}). However, a state-of-the-art algorithmic matcher enables recall boosting, adding an average of $6.9$ (other) correct correspondences to the final match, improving both recall and \fm.

\subsubsection{\sysNameSpace Precision - Recall Tradeoff}
\label{sec:threshold}

\begin{figure}[h]
	\begin{subfigure}{.32\linewidth}
		\centering
		\caption{Target measure: recall ($R$)}
		\includegraphics[width=\linewidth]{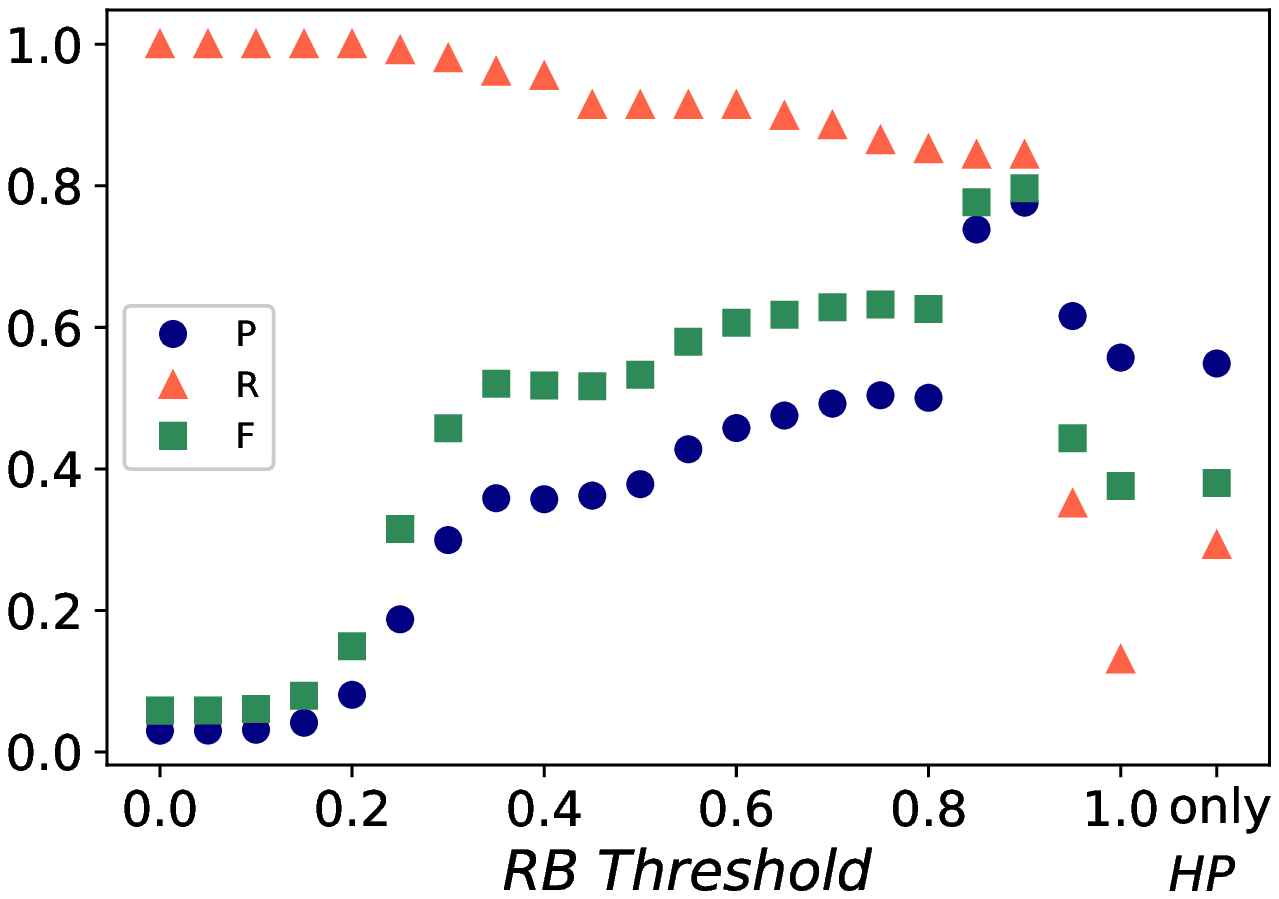}
		\label{fig:mainR}
	\end{subfigure}
	\begin{subfigure}{.32\linewidth}
		\centering
		\caption{Target measure: precision ($P$)}
		\includegraphics[width=\linewidth]{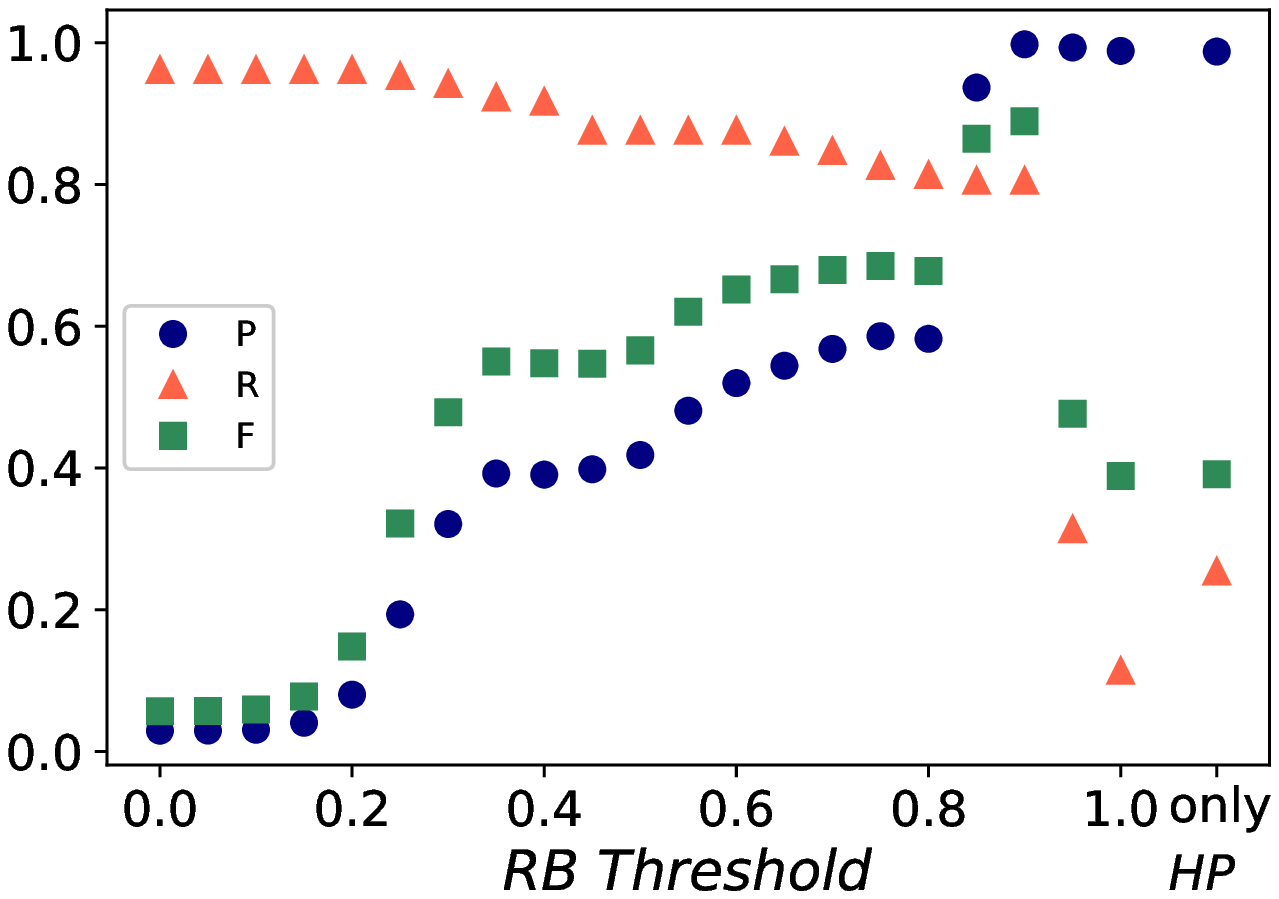}
		\label{fig:mainP}
	\end{subfigure}
	\begin{subfigure}{.32\linewidth}
		\centering
		\caption{Target measure: \fmspace ($F$)}
		\includegraphics[width=\linewidth]{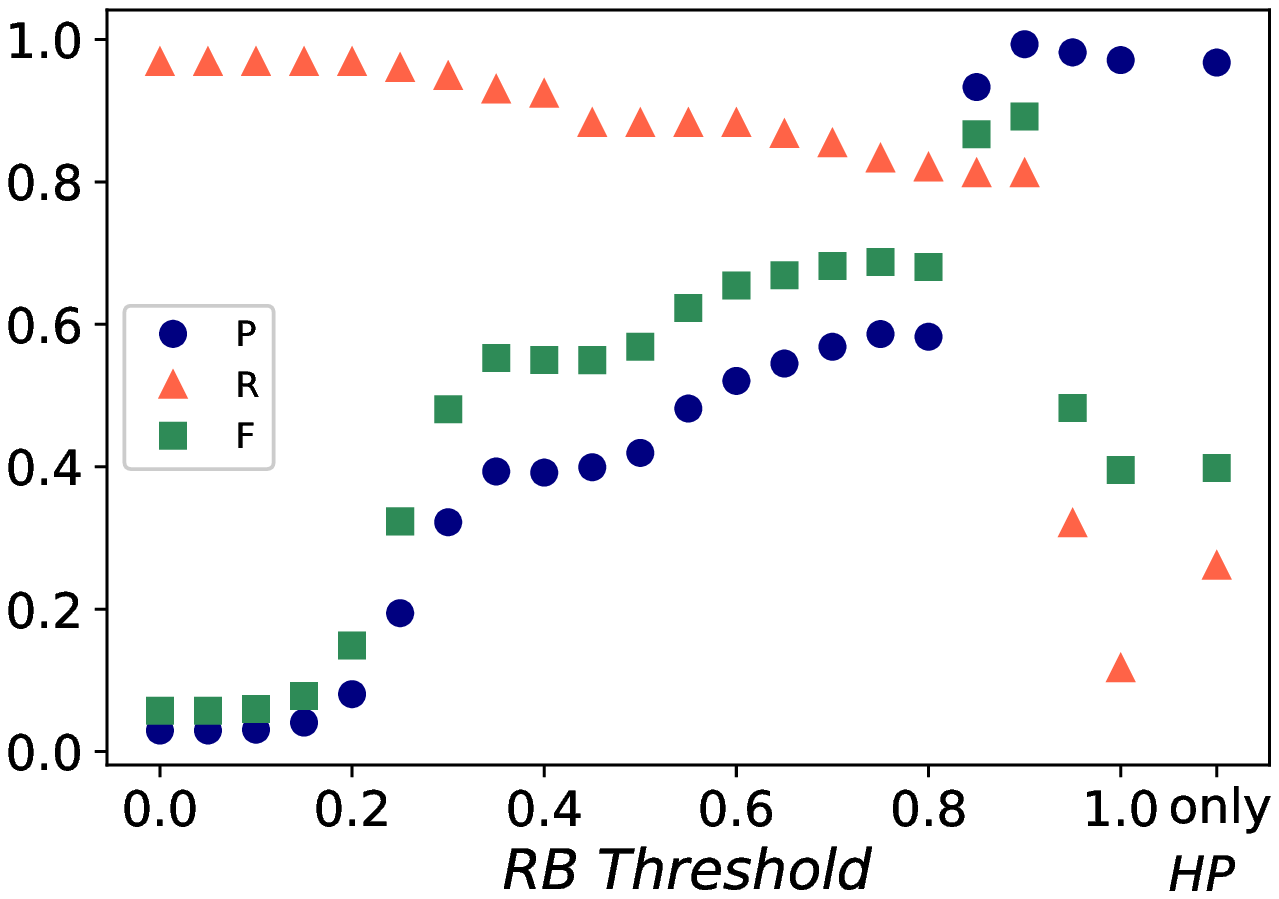}
		\label{fig:mainF}
	\end{subfigure}
	\centering
	\caption{\sysNameSpace performance in terms of recall ($R$, red triangles), precision ($P$, blue dot), and \fmspace ($F$, green squares) as a function of $RB$ threshold by target measure (using dynamic thresholds).$^{\ref{fn:th}}$}
	\label{fig:main}
\end{figure}

Our analysis thus far used an $RB$ threshold of $0.9$, which yielded the best performance during training. We next turn to examine how the tradeoff between precision and recall changes with the $RB$ threshold. Figure~\ref{fig:main} illustrates precision ($P$), recall ($R$), and \fmspace ($F$), targeting recall (Figure~\ref{fig:mainR}), precision (with dynamic thresholds,\footnote{\label{fn:th}Static thresholds yielded similar results, which can be found in an online repository~\cite{graphs}.} Figure~\ref{fig:mainP}), and \fmspace (with dynamic thresholds, Figure~\ref{fig:mainF}). The far right values in each graph represent using $HP$ only, allowing no algorithmic results into the output match $\hat{\sigma}$. Values at the far left, setting $RB$ threshold to $0$, includes all algorithmic results in  $\hat{\sigma}$.
Overall, the three graphs demonstrate a similar trend. Primarily, regardless of the target measure, a $0.9$ $RB$ threshold yields the best results in terms of \fmspace (as was set during training). Recall is at its peak when adding all correspondences human matchers did not assign ($RB$ threshold = 0). A conservative approach of adding only correspondences the algorithmic matcher is fully confident about ($RB$ threshold = 1) results in a very low recall.

\subsection{\sysNameSpace Generalizability}
\label{sec:mainresOA}

The core idea of this section is to provide empirical evidence to the applicability of \sysName. Sections~\ref{sec:hpres}-\ref{sec:rbres} demonstrate the effectiveness of \sysNameSpace to improve the matching quality of human schema matchers in the (standard) domain of purchase orders (PO)~\cite{DO2002a}. In such a setting, we trained \sysNameSpace over a given set of human matchers and show that it can be used effectively on a test set of \textbf{unseen} human matchers (in 5-fold cross validation fashion, see Section~\ref{sec:meth}), all in the PO domain. In the following analysis we use different (unseen) human matchers, performing a (slightly) different matching task (ontology alignment) from a different domain (bibliographic references~\cite{oaeiURL}), to show the power of \sysNameSpace in generalizing beyond schema matching. Please refer to Section~\ref{sec:meth} for additional details. Table~\ref{tab:mainOA} presents results on human matching dataset of OAEI (see Section~\ref{sec:dataset}) in a similar fashion to Table~\ref{tab:mainrb}. 

\begin{table}[h]
	\caption{Precision ($P$), recall ($R$), \fmspace ($F$) of \sysNameSpace by target measure compared to baselines (OAEI task)}
	\label{tab:mainOA}
	\scalebox{1}{\begin{tabular}{|l|l|l|ccc|}
			\hline
			Target                 & (threshold)     & Method  & $P$              & $R$              & $F$              \\ \hline
			\multirow{3}{*}{$R$}    & \multirow{3}{*}{$0.0$}         & \sysName($HP$)      &   0.616      &  0.348      &    0.447      \\ \cline{3-6} 
			&& \cellcolor{Gray}\sysName($HP$+$RB$)      & \cellcolor{Gray} 0.554          & \cellcolor{Gray}0.896           & \cellcolor{Gray} 0.749*          \\ 
			&& $raw$-{\sf ADnEV}      &  0.484        &  0.861        &   0.638      \\ \cline{1-6}
			\multirow{6}{*}{$P$} & \multirow{3}{*}{$1.0$}  & \sysName($HP$)      & 0.878  &  0.315    &  0.432         \\ \cline{3-6} 
			&& \cellcolor{Gray}\sysName($HP$+$RB$)        & \cellcolor{Gray}\textbf{0.912} & \cellcolor{Gray}0.746*          & \cellcolor{Gray}0.792*          \\ 
			&& $raw$-{\sf ADnEV}      &    0.812       &   0.582       & 0.688          \\ \cline{2-6}
			& \multirow{3}{*}{$\hat{P}(\sigma_{t-1})$} & \sysName($HP$)      &  0.867 &  0.318      &    0.436         \\ \cline{3-6} 
			&& \cellcolor{Gray}\sysName($HP$+$RB$)        & \cellcolor{Gray} 0.889        & \cellcolor{Gray} 0.776*         & \cellcolor{Gray}0.810*          \\ 
			&& $raw$-{\sf ADnEV}    & 0.785       &   0.567       & 0.674          \\ \cline{1-6}
			\multirow{6}{*}{$F$}  & \multirow{3}{*}{$0.5$}   & \sysName($HP$)      &   0.824       &   0.327     &   0.442         \\ \cline{3-6} 
			&& \cellcolor{Gray}\sysName($HP$+$RB$)        & \cellcolor{Gray} 0.896*         & \cellcolor{Gray}0.747*           & \cellcolor{Gray}0.809*           \\ 
			&& $raw$-{\sf ADnEV}     &   0.745       &   0.582        &  0.681         \\ \cline{2-6}
			& \multirow{3}{*}{$0.5\cdot\hat{F}(\sigma_{t-1})$} & \sysName($HP$)    &    0.811       &   0.319     &   0.437       \\ \cline{3-6} 
			&& \cellcolor{Gray}\sysName($HP$+$RB$)        & \cellcolor{Gray} 0.892*        & \cellcolor{Gray}0.772*           & \cellcolor{Gray}\textbf{0.825*} \\\ 
			&& $raw$-{\sf ADnEV}     &    0.743       &   0.586         &   0.682       \\ \hline
			\multicolumn{2}{|c|}{} & {\sf ADnEV}~\cite{shraga2020}     &  0.677        &  0.656       &  0.667  \\ \cline{3-6} 
			\multicolumn{2}{|c|}{} & {\sf Term}~\cite{GAL2011}       &  0.266        &   0.750       &  0.393        \\ \cline{3-6} 
			\multicolumn{2}{|c|}{-} & {\sf Token Path}~\cite{PEUKERT2011}       &  0.400        & 0.250 &  0.307         \\ \cline{3-6} 
			\multicolumn{2}{|c|}{} & {\sf WordNet}~\cite{wordnet1}       &    0.462      & \textbf{0.937}   &   0.618        \\ \hline
	\end{tabular}}
\end{table}

Matching results are slightly lower than in the PO task. However, the tendency is the same, demonstrating that a \sysNameSpace trained on the domain of schema matching can align ontologies well. The main difference between the performance of \sysNameSpace on the PO task and the OAEI task is in terms of precision, where the results of the latter failed to reach the 0.999 precision value of the former (when targeting precision). This may not come as a surprise since the trained model affects only the $HP$ component, which is also in charge of providing high precision.

\section{Related work}
\label{sec:related}
Human-in-the-loop in schema matching typically uses either crowdsourcing~\cite{NoyMMA13,Crowdmap,Zhang2013,zhang2018reducing,fan2014hybrid} or pay-as-you-go~\cite{McCann2008,Hung2014,pinkel2013incmap} to reduce the demanding cognitive load of this task. The former slices the task into smaller sized tasks and spread the load over multiple matchers. The latter partitions the task over time, aiming at minimizing the matching task effort at each point in time. From a usability point-of-view, Noy, Lanbrix, and Falconer~\cite{Falconer2007,Lambrix2003,Noy2002} investigated ways to assist humans in validating results of computerized matching systems. In this work, we provide an alternative approach, offering an algorithmic solution that is shown to improve on human matching performance. Our approach takes a human matcher's input and boosts its performance by analyzing the process a human matcher followed and complementing it with an algorithmic matcher. 

Also using a crowdsourcing technique, Bozovic and Vassalos~\cite{Bozovic2015} proposed a combined human-algorithm matching system where limited user feedback is used to weigh the algorithmic matching. In our work we offer an opposite approach, according to which a human match is provided, evaluated and modified, and then extended with algorithmic solutions.  

Human matching performance was analyzed in both schema matching and the related field of ontology alignment~\cite{dragisic2016user,li2019user,zhang2018reducing}, acknowledging that humans can err while matching due to biases~\cite{ackerman2019cognitive}. Our work turns such bugs (biases in our case) into features, improving human matching performance by assessing the impact of biases on the quality of the match.

The use of deep learning for solving data integration problems becomes widespread~\cite{ebraheem2017deeper,mudgal2018deep,kolyvakis2018deepalignment,thirumuruganathan2020data,li2021deep}. Chen {\em et al.}~\cite{chen2018biggorilla} use instance data to apply supervised learning for schema matching. Fernandez \emph{et al.}~\cite{fernandez2018seeping} use embeddings to identify relationships between attributes, which was extended by Cappuzzo \emph{et al.}~\cite{cappuzzo2020creating} to consider instances in creating local embeddings. Shraga \emph{et al.}~\cite{shraga2020} use a neural network to improve an algorithmic schema matching result. In our work, we use an LSTM to capture the time-dependent decision making of human matching and complement it with a state-of-the-art deep learning-based algorithmic matching based on~\cite{shraga2020}.

\section{Conclusions and Future Work}
\label{sec:discussion}

This work offers a novel approach to address matching, analyzing it as a process and improving its quality using machine learning techniques. We recognize that human matching is basically a sequential process and define a matching sequential process using matching history (Definition~\ref{def:history}) and monotonic evaluation of the matching process (Section~\ref{sec:monotonicEvaluation}). We show conditions under which precision, recall and \fmspace are monotonic (Theorem~\ref{thm:MIEM}). Then, aiming to improve on the matching quality, we tie the monotonicity of these measures to the ability of a correspondence to improve on a match evaluation and characterize such correspondences in probabilistic terms (Theorem~\ref{thm:probLocalAnneal}). Realizing that human matching is biased (Section~\ref{sec:RWHMatching}) we offer \sysNameSpace to calibrate human matching decisions and compensate for correspondences that were left out by human matchers using algorithmic matching. Our empirical evaluation shows a clear benefit in treating matching as a process, confirming that \sysNameSpace improves on both human and algorithmic matching. \add{We also provide a proof-of-concept, showing that \sysNameSpace generalizes well to the closely domain of ontology alignment.} \add{An important insight of this work relates to the way training data should be obtained in future matching research. The observations of this paper can serve as a guideline for collecting (query user confidence, timing the decisions, \emph{etc.}), managing (using a decision history instead of similarity matrix), and using (calibrating decisions using \sysNameSpace or a derivative) data from human matchers.}

In future work, we aim to extend \sysNameSpace to additional platforms, \emph{e.g.,} crowdsourcing, where several additional aspects, such as crowd workers heterogeneity~\cite{Ross2010}, should be considered. Interesting research directions involve experimenting with additional matching tools and analyzing the merits of LSTM in terms of overfitting and sufficient training data.

\bibliographystyle{ACM-Reference-Format.bst}
\bibliography{ltsLong}


\begin{thebibliography}{63}


\ifx \showCODEN    \undefined \def \showCODEN     #1{\unskip}     \fi
\ifx \showDOI      \undefined \def \showDOI       #1{#1}\fi
\ifx \showISBNx    \undefined \def \showISBNx     #1{\unskip}     \fi
\ifx \showISBNxiii \undefined \def \showISBNxiii  #1{\unskip}     \fi
\ifx \showISSN     \undefined \def \showISSN      #1{\unskip}     \fi
\ifx \showLCCN     \undefined \def \showLCCN      #1{\unskip}     \fi
\ifx \shownote     \undefined \def \shownote      #1{#1}          \fi
\ifx \showarticletitle \undefined \def \showarticletitle #1{#1}   \fi
\ifx \showURL      \undefined \def \showURL       {\relax}        \fi
\providecommand\bibfield[2]{#2}
\providecommand\bibinfo[2]{#2}
\providecommand\natexlab[1]{#1}
\providecommand\showeprint[2][]{arXiv:#2}

\bibitem[\protect\citeauthoryear{??}{dat}{2021}]%
        {data}
 \bibinfo{year}{2021}\natexlab{}.
\newblock \bibinfo{title}{Data}.
\newblock
  \bibinfo{howpublished}{\url{https://github.com/shraga89/PoWareMatch/tree/master/DataFiles}}.
\newblock


\bibitem[\protect\citeauthoryear{??}{gra}{2021}]%
        {graphs}
 \bibinfo{year}{2021}\natexlab{}.
\newblock \bibinfo{title}{Graphs}.
\newblock
  \bibinfo{howpublished}{\url{https://github.com/shraga89/PoWareMatch/tree/master/Eval_graphs}}.
\newblock


\bibitem[\protect\citeauthoryear{??}{oae}{2021}]%
        {oaeiURL}
 \bibinfo{year}{2021}\natexlab{}.
\newblock \bibinfo{title}{OAEI benchmark}.
\newblock
  \bibinfo{howpublished}{\url{http://oaei.ontologymatching.org/2011/benchmarks/}}.
\newblock


\bibitem[\protect\citeauthoryear{??}{ore}{2021}]%
        {oreURL}
 \bibinfo{year}{2021}\natexlab{}.
\newblock \bibinfo{title}{Ontobuilder research environment}.
\newblock
  \bibinfo{howpublished}{\url{https://github.com/shraga89/Ontobuilder-Research-Environment}}.
\newblock


\bibitem[\protect\citeauthoryear{??}{git}{2021a}]%
        {gitConfig}
 \bibinfo{year}{2021}\natexlab{a}.
\newblock \bibinfo{title}{PoWareMatch Configuration}.
\newblock
  \bibinfo{howpublished}{\url{https://github.com/shraga89/PoWareMatch/blob/master/RunFiles/config.py}}.
\newblock


\bibitem[\protect\citeauthoryear{??}{git}{2021b}]%
        {gitURL}
 \bibinfo{year}{2021}\natexlab{b}.
\newblock \bibinfo{title}{PoWareMatch repository}.
\newblock
  \bibinfo{howpublished}{\url{https://github.com/shraga89/PoWareMatch}}.
\newblock


\bibitem[\protect\citeauthoryear{??}{tor}{2021}]%
        {torchURL}
 \bibinfo{year}{2021}\natexlab{}.
\newblock \bibinfo{title}{PyTorch}.
\newblock \bibinfo{howpublished}{\url{https://pytorch.org/}}.
\newblock


\bibitem[\protect\citeauthoryear{??}{tec}{2021}]%
        {tech}
 \bibinfo{year}{2021}\natexlab{}.
\newblock \bibinfo{title}{Technical Report}.
\newblock
  \bibinfo{howpublished}{\url{https://github.com/shraga89/PoWareMatch/blob/master/PoWareMatch_Tech.pdf}}.
\newblock


\bibitem[\protect\citeauthoryear{Ackerman, Gal, Sagi, and Shraga}{Ackerman
  et~al\mbox{.}}{2019}]%
        {ackerman2019cognitive}
\bibfield{author}{\bibinfo{person}{Rakefet Ackerman}, \bibinfo{person}{Avigdor
  Gal}, \bibinfo{person}{Tomer Sagi}, {and} \bibinfo{person}{Roee Shraga}.}
  \bibinfo{year}{2019}\natexlab{}.
\newblock \showarticletitle{A Cognitive Model of Human Bias in Matching}. In
  \bibinfo{booktitle}{\emph{Pacific Rim International Conference on Artificial
  Intelligence}}. Springer, \bibinfo{pages}{632--646}.
\newblock


\bibitem[\protect\citeauthoryear{Ackerman and Thompson}{Ackerman and
  Thompson}{2017}]%
        {Ackerman2017}
\bibfield{author}{\bibinfo{person}{Rakefet. Ackerman} {and}
  \bibinfo{person}{Valerie Thompson}.} \bibinfo{year}{2017}\natexlab{}.
\newblock \showarticletitle{Meta-Reasoning: Monitoring and control of thinking
  and reasoning}.
\newblock \bibinfo{journal}{\emph{Trends in Cognitive Sciences}}
  \bibinfo{volume}{21}, \bibinfo{number}{8} (\bibinfo{year}{2017}),
  \bibinfo{pages}{607--617}.
\newblock


\bibitem[\protect\citeauthoryear{Barsalou}{Barsalou}{2014}]%
        {Barsalou2014}
\bibfield{author}{\bibinfo{person}{Lawrence~W Barsalou}.}
  \bibinfo{year}{2014}\natexlab{}.
\newblock \bibinfo{booktitle}{\emph{Cognitive psychology: An overview for
  cognitive scientists}}.
\newblock \bibinfo{publisher}{Psychology Press}.
\newblock


\bibitem[\protect\citeauthoryear{Bellahsene, Bonifati, Duchateau, and
  Velegrakis}{Bellahsene et~al\mbox{.}}{2011b}]%
        {BELLAHSENE2011a}
\bibfield{author}{\bibinfo{person}{Zohra Bellahsene}, \bibinfo{person}{Angela
  Bonifati}, \bibinfo{person}{Fabien Duchateau}, {and} \bibinfo{person}{Yannis
  Velegrakis}.} \bibinfo{year}{2011}\natexlab{b}.
\newblock \showarticletitle{On Evaluating Schema Matching and Mapping}.
\newblock In \bibinfo{booktitle}{\emph{Schema Matching and Mapping}}.
  \bibinfo{publisher}{Springer Berlin Heidelberg}, \bibinfo{pages}{253--291}.
\newblock
\showISBNx{978-3-642-16518-4}


\bibitem[\protect\citeauthoryear{Bellahsene, Bonifati, and Rahm}{Bellahsene
  et~al\mbox{.}}{2011a}]%
        {BELLAHSENE2011}
\bibfield{editor}{\bibinfo{person}{Zohra Bellahsene}, \bibinfo{person}{Angela
  Bonifati}, {and} \bibinfo{person}{Erhard Rahm}} (Eds.).
  \bibinfo{year}{2011}\natexlab{a}.
\newblock \bibinfo{booktitle}{\emph{Schema Matching and Mapping}}.
\newblock \bibinfo{publisher}{Springer}.
\newblock
\showISBNx{978-3-642-16517-7}


\bibitem[\protect\citeauthoryear{Benesty, Chen, Huang, and Cohen}{Benesty
  et~al\mbox{.}}{2009}]%
        {benesty2009pearson}
\bibfield{author}{\bibinfo{person}{Jacob Benesty}, \bibinfo{person}{Jingdong
  Chen}, \bibinfo{person}{Yiteng Huang}, {and} \bibinfo{person}{Israel Cohen}.}
  \bibinfo{year}{2009}\natexlab{}.
\newblock \showarticletitle{Pearson correlation coefficient}.
\newblock In \bibinfo{booktitle}{\emph{Noise reduction in speech processing}}.
  \bibinfo{publisher}{Springer}, \bibinfo{pages}{1--4}.
\newblock


\bibitem[\protect\citeauthoryear{Bernstein, Madhavan, and Rahm}{Bernstein
  et~al\mbox{.}}{2011}]%
        {Bernstein2011}
\bibfield{author}{\bibinfo{person}{Philip~A. Bernstein},
  \bibinfo{person}{Jayant Madhavan}, {and} \bibinfo{person}{Erhard Rahm}.}
  \bibinfo{year}{2011}\natexlab{}.
\newblock \showarticletitle{Generic Schema Matching, Ten Years Later}.
\newblock \bibinfo{journal}{\emph{{PVLDB}}} \bibinfo{volume}{4},
  \bibinfo{number}{11} (\bibinfo{year}{2011}), \bibinfo{pages}{695--701}.
\newblock


\bibitem[\protect\citeauthoryear{Bjork, Dunlosky, and Kornell}{Bjork
  et~al\mbox{.}}{2013}]%
        {Bjork2013}
\bibfield{author}{\bibinfo{person}{Robert~A Bjork}, \bibinfo{person}{John
  Dunlosky}, {and} \bibinfo{person}{Nate Kornell}.}
  \bibinfo{year}{2013}\natexlab{}.
\newblock \showarticletitle{Self-regulated learning: Beliefs, techniques, and
  illusions}.
\newblock \bibinfo{journal}{\emph{Annual Review of Psychology}}
  \bibinfo{volume}{64} (\bibinfo{year}{2013}), \bibinfo{pages}{417--444}.
\newblock


\bibitem[\protect\citeauthoryear{Bozovic and Vassalos}{Bozovic and
  Vassalos}{2015}]%
        {Bozovic2015}
\bibfield{author}{\bibinfo{person}{Nikolaos Bozovic} {and}
  \bibinfo{person}{Vasilis Vassalos}.} \bibinfo{year}{2015}\natexlab{}.
\newblock \showarticletitle{Two Phase User Driven Schema Matching}. In
  \bibinfo{booktitle}{\emph{Advances in Databases and Information Systems}}.
  \bibinfo{pages}{49--62}.
\newblock


\bibitem[\protect\citeauthoryear{Cappuzzo, Papotti, and
  Thirumuruganathan}{Cappuzzo et~al\mbox{.}}{2020}]%
        {cappuzzo2020creating}
\bibfield{author}{\bibinfo{person}{Riccardo Cappuzzo}, \bibinfo{person}{Paolo
  Papotti}, {and} \bibinfo{person}{Saravanan Thirumuruganathan}.}
  \bibinfo{year}{2020}\natexlab{}.
\newblock \showarticletitle{Creating embeddings of heterogeneous relational
  datasets for data integration tasks}. In \bibinfo{booktitle}{\emph{SIGMOD}}.
  \bibinfo{pages}{1335--1349}.
\newblock


\bibitem[\protect\citeauthoryear{Chen, Golshan, Halevy, Tan, and Doan}{Chen
  et~al\mbox{.}}{2018}]%
        {chen2018biggorilla}
\bibfield{author}{\bibinfo{person}{Chen Chen}, \bibinfo{person}{Behzad
  Golshan}, \bibinfo{person}{Alon~Y Halevy}, \bibinfo{person}{Wang-Chiew Tan},
  {and} \bibinfo{person}{AnHai Doan}.} \bibinfo{year}{2018}\natexlab{}.
\newblock \showarticletitle{BigGorilla: An Open-Source Ecosystem for Data
  Preparation and Integration.}
\newblock \bibinfo{journal}{\emph{IEEE Data Eng. Bull.}} \bibinfo{volume}{41},
  \bibinfo{number}{2} (\bibinfo{year}{2018}), \bibinfo{pages}{10--22}.
\newblock


\bibitem[\protect\citeauthoryear{Do and Rahm}{Do and Rahm}{2002}]%
        {DO2002a}
\bibfield{author}{\bibinfo{person}{Hong-Hai Do} {and} \bibinfo{person}{Erhard
  Rahm}.} \bibinfo{year}{2002}\natexlab{}.
\newblock \showarticletitle{COMA—a system for flexible combination of schema
  matching approaches}. In \bibinfo{booktitle}{\emph{VLDB'02: Proceedings of
  the 28th International Conference on Very Large Databases}}. Elsevier,
  \bibinfo{pages}{610--621}.
\newblock


\bibitem[\protect\citeauthoryear{Dong, Halevy, and Yu}{Dong
  et~al\mbox{.}}{2009}]%
        {DONG2009}
\bibfield{author}{\bibinfo{person}{Xin Dong}, \bibinfo{person}{Alon Halevy},
  {and} \bibinfo{person}{Cong Yu}.} \bibinfo{year}{2009}\natexlab{}.
\newblock \showarticletitle{Data integration with uncertainty}.
\newblock \bibinfo{journal}{\emph{The VLDB Journal}}  \bibinfo{volume}{18}
  (\bibinfo{year}{2009}), \bibinfo{pages}{469--500}.
\newblock
\showISSN{1066-8888}


\bibitem[\protect\citeauthoryear{Dragisic, Ivanova, Lambrix, Faria,
  Jim{\'e}nez-Ruiz, and Pesquita}{Dragisic et~al\mbox{.}}{2016}]%
        {dragisic2016user}
\bibfield{author}{\bibinfo{person}{Zlatan Dragisic}, \bibinfo{person}{Valentina
  Ivanova}, \bibinfo{person}{Patrick Lambrix}, \bibinfo{person}{Daniel Faria},
  \bibinfo{person}{Ernesto Jim{\'e}nez-Ruiz}, {and} \bibinfo{person}{Catia
  Pesquita}.} \bibinfo{year}{2016}\natexlab{}.
\newblock \showarticletitle{User validation in ontology alignment}. In
  \bibinfo{booktitle}{\emph{International Semantic Web Conference}}. Springer,
  \bibinfo{pages}{200--217}.
\newblock


\bibitem[\protect\citeauthoryear{Ebraheem, Thirumuruganathan, Joty, Ouzzani,
  and Tang}{Ebraheem et~al\mbox{.}}{2018}]%
        {ebraheem2017deeper}
\bibfield{author}{\bibinfo{person}{Muhammad Ebraheem},
  \bibinfo{person}{Saravanan Thirumuruganathan}, \bibinfo{person}{Shafiq Joty},
  \bibinfo{person}{Mourad Ouzzani}, {and} \bibinfo{person}{Nan Tang}.}
  \bibinfo{year}{2018}\natexlab{}.
\newblock \showarticletitle{Distributed Representations of Tuples for Entity
  Resolution}.
\newblock \bibinfo{journal}{\emph{PVLDB}} \bibinfo{volume}{11},
  \bibinfo{number}{11} (\bibinfo{year}{2018}).
\newblock


\bibitem[\protect\citeauthoryear{Euzenat, Shvaiko, et~al\mbox{.}}{Euzenat
  et~al\mbox{.}}{2007}]%
        {EUZENAT2007a}
\bibfield{author}{\bibinfo{person}{J{\'e}r{\^o}me Euzenat},
  \bibinfo{person}{Pavel Shvaiko}, {et~al\mbox{.}}}
  \bibinfo{year}{2007}\natexlab{}.
\newblock \bibinfo{booktitle}{\emph{Ontology matching}}.
  Vol.~\bibinfo{volume}{18}.
\newblock \bibinfo{publisher}{Springer}.
\newblock


\bibitem[\protect\citeauthoryear{Falconer and Storey}{Falconer and
  Storey}{2007}]%
        {Falconer2007}
\bibfield{author}{\bibinfo{person}{Sean~M. Falconer} {and}
  \bibinfo{person}{Margaret{-}Anne~D. Storey}.}
  \bibinfo{year}{2007}\natexlab{}.
\newblock \showarticletitle{A Cognitive Support Framework for Ontology
  Mapping}.
\newblock In \bibinfo{booktitle}{\emph{International Semantic Web Conference,
  {ISWC}}}. \bibinfo{series}{Lecture Notes in Computer Science},
  Vol.~\bibinfo{volume}{4825}. \bibinfo{publisher}{Springer Berlin Heidelberg},
  \bibinfo{pages}{114--127}.
\newblock
\showISBNx{978-3-540-76297-3}


\bibitem[\protect\citeauthoryear{Fan, Lu, Ooi, Tan, and Zhang}{Fan
  et~al\mbox{.}}{2014}]%
        {fan2014hybrid}
\bibfield{author}{\bibinfo{person}{Ju Fan}, \bibinfo{person}{Meiyu Lu},
  \bibinfo{person}{Beng~Chin Ooi}, \bibinfo{person}{Wang-Chiew Tan}, {and}
  \bibinfo{person}{Meihui Zhang}.} \bibinfo{year}{2014}\natexlab{}.
\newblock \showarticletitle{A hybrid machine-crowdsourcing system for matching
  web tables}. In \bibinfo{booktitle}{\emph{2014 IEEE 30th International
  Conference on Data Engineering}}. IEEE, \bibinfo{pages}{976--987}.
\newblock


\bibitem[\protect\citeauthoryear{Fernandez, Mansour, Qahtan, Elmagarmid, Ilyas,
  Madden, Ouzzani, Stonebraker, and Tang}{Fernandez et~al\mbox{.}}{2018}]%
        {fernandez2018seeping}
\bibfield{author}{\bibinfo{person}{Raul~Castro Fernandez},
  \bibinfo{person}{Essam Mansour}, \bibinfo{person}{Abdulhakim~A Qahtan},
  \bibinfo{person}{Ahmed Elmagarmid}, \bibinfo{person}{Ihab Ilyas},
  \bibinfo{person}{Samuel Madden}, \bibinfo{person}{Mourad Ouzzani},
  \bibinfo{person}{Michael Stonebraker}, {and} \bibinfo{person}{Nan Tang}.}
  \bibinfo{year}{2018}\natexlab{}.
\newblock \showarticletitle{Seeping semantics: Linking datasets using word
  embeddings for data discovery}. In \bibinfo{booktitle}{\emph{2018 IEEE 34th
  International Conference on Data Engineering (ICDE)}}. IEEE,
  \bibinfo{pages}{989--1000}.
\newblock


\bibitem[\protect\citeauthoryear{Gal}{Gal}{2011}]%
        {GAL2011}
\bibfield{author}{\bibinfo{person}{Avigdor Gal}.}
  \bibinfo{year}{2011}\natexlab{}.
\newblock \bibinfo{booktitle}{\emph{Uncertain Schema Matching}}.
\newblock \bibinfo{publisher}{Morgan {\&} Claypool Publishers}.
\newblock


\bibitem[\protect\citeauthoryear{Gal, Roitman, and Shraga}{Gal
  et~al\mbox{.}}{2019}]%
        {roee2018icdm}
\bibfield{author}{\bibinfo{person}{Avigdor Gal}, \bibinfo{person}{Haggai
  Roitman}, {and} \bibinfo{person}{Roee Shraga}.}
  \bibinfo{year}{2019}\natexlab{}.
\newblock \showarticletitle{Learning to Rerank Schema Matches}.
\newblock \bibinfo{journal}{\emph{IEEE Transactions on Knowledge and Data
  Engineering (TKDE)}} (\bibinfo{year}{2019}).
\newblock
\showISSN{1041-4347}


\bibitem[\protect\citeauthoryear{Gawinecki}{Gawinecki}{2009}]%
        {wordnet1}
\bibfield{author}{\bibinfo{person}{Maciej Gawinecki}.}
  \bibinfo{year}{2009}\natexlab{}.
\newblock \showarticletitle{Abbreviation expansion in lexical annotation of
  schema}.
\newblock \bibinfo{journal}{\emph{Camogli (Genova), Italy June 25th, 2009
  Co-located with SEBD}} (\bibinfo{year}{2009}), \bibinfo{pages}{61}.
\newblock


\bibitem[\protect\citeauthoryear{Gers, Schmidhuber, and Cummins}{Gers
  et~al\mbox{.}}{2000}]%
        {gers2000learning}
\bibfield{author}{\bibinfo{person}{Felix~A Gers}, \bibinfo{person}{J{\"u}rgen
  Schmidhuber}, {and} \bibinfo{person}{Fred Cummins}.}
  \bibinfo{year}{2000}\natexlab{}.
\newblock \showarticletitle{Learning to forget: Continual prediction with
  LSTM}.
\newblock \bibinfo{journal}{\emph{Neural computation}} \bibinfo{volume}{12},
  \bibinfo{number}{10} (\bibinfo{year}{2000}), \bibinfo{pages}{2451--2471}.
\newblock


\bibitem[\protect\citeauthoryear{Halevy and Madhavan}{Halevy and
  Madhavan}{2003}]%
        {halevy2003corpus}
\bibfield{author}{\bibinfo{person}{Alon~Y Halevy} {and} \bibinfo{person}{Jayant
  Madhavan}.} \bibinfo{year}{2003}\natexlab{}.
\newblock \showarticletitle{Corpus-based knowledge representation}. In
  \bibinfo{booktitle}{\emph{IJCAI}}, Vol.~\bibinfo{volume}{3}.
  \bibinfo{pages}{1567--1572}.
\newblock


\bibitem[\protect\citeauthoryear{Hammer, Stonebraker, and Topsakal}{Hammer
  et~al\mbox{.}}{2005}]%
        {HammerST05}
\bibfield{author}{\bibinfo{person}{Joachim Hammer}, \bibinfo{person}{Michael
  Stonebraker}, {and} \bibinfo{person}{Oguzhan Topsakal}.}
  \bibinfo{year}{2005}\natexlab{}.
\newblock \showarticletitle{{THALIA:} Test Harness for the Assessment of Legacy
  Information Integration Approaches}. In \bibinfo{booktitle}{\emph{ICDE}}.
  \bibinfo{pages}{485--486}.
\newblock


\bibitem[\protect\citeauthoryear{He and Chang}{He and Chang}{2005}]%
        {he2005making}
\bibfield{author}{\bibinfo{person}{Bin He} {and} \bibinfo{person}{Kevin
  Chen-Chuan Chang}.} \bibinfo{year}{2005}\natexlab{}.
\newblock \showarticletitle{Making holistic schema matching robust: an ensemble
  approach}. In \bibinfo{booktitle}{\emph{Proceedings of the eleventh ACM
  SIGKDD international conference on Knowledge discovery in data mining}}.
  \bibinfo{pages}{429--438}.
\newblock


\bibitem[\protect\citeauthoryear{Kendall}{Kendall}{1938}]%
        {kendall1938new}
\bibfield{author}{\bibinfo{person}{Maurice~G Kendall}.}
  \bibinfo{year}{1938}\natexlab{}.
\newblock \showarticletitle{A new measure of rank correlation}.
\newblock \bibinfo{journal}{\emph{Biometrika}} \bibinfo{volume}{30},
  \bibinfo{number}{1/2} (\bibinfo{year}{1938}), \bibinfo{pages}{81--93}.
\newblock


\bibitem[\protect\citeauthoryear{Kolyvakis, Kalousis, and Kiritsis}{Kolyvakis
  et~al\mbox{.}}{2018}]%
        {kolyvakis2018deepalignment}
\bibfield{author}{\bibinfo{person}{Prodromos Kolyvakis},
  \bibinfo{person}{Alexandros Kalousis}, {and} \bibinfo{person}{Dimitris
  Kiritsis}.} \bibinfo{year}{2018}\natexlab{}.
\newblock \showarticletitle{Deepalignment: Unsupervised ontology matching with
  refined word vectors}. In \bibinfo{booktitle}{\emph{Proceedings of the 2018
  Conference of the North American Chapter of the Association for Computational
  Linguistics: Human Language Technologies, Volume 1 (Long Papers)}}.
  \bibinfo{pages}{787--798}.
\newblock


\bibitem[\protect\citeauthoryear{Lambrix and Edberg}{Lambrix and
  Edberg}{2003}]%
        {Lambrix2003}
\bibfield{author}{\bibinfo{person}{Patrick Lambrix} {and} \bibinfo{person}{Anna
  Edberg}.} \bibinfo{year}{2003}\natexlab{}.
\newblock \showarticletitle{Evaluation of Ontology Merging Tools in
  Bioinformatics}. In \bibinfo{booktitle}{\emph{Proceedings of the 8th Pacific
  Symposium on Biocomputing, {PSB} 2003, Lihue, Hawaii, USA, January 3-7,
  2003}}, Vol.~\bibinfo{volume}{8}. \bibinfo{pages}{589--600}.
\newblock


\bibitem[\protect\citeauthoryear{LeCun, Bengio, and Hinton}{LeCun
  et~al\mbox{.}}{2015}]%
        {lecun2015deep}
\bibfield{author}{\bibinfo{person}{Yann LeCun}, \bibinfo{person}{Yoshua
  Bengio}, {and} \bibinfo{person}{Geoffrey Hinton}.}
  \bibinfo{year}{2015}\natexlab{}.
\newblock \showarticletitle{Deep learning}.
\newblock \bibinfo{journal}{\emph{nature}} \bibinfo{volume}{521},
  \bibinfo{number}{7553} (\bibinfo{year}{2015}), \bibinfo{pages}{436}.
\newblock


\bibitem[\protect\citeauthoryear{Li}{Li}{2017}]%
        {li2017human}
\bibfield{author}{\bibinfo{person}{Guoliang Li}.}
  \bibinfo{year}{2017}\natexlab{}.
\newblock \showarticletitle{Human-in-the-loop data integration}.
\newblock \bibinfo{journal}{\emph{Proceedings of the VLDB Endowment}}
  \bibinfo{volume}{10}, \bibinfo{number}{12} (\bibinfo{year}{2017}),
  \bibinfo{pages}{2006--2017}.
\newblock


\bibitem[\protect\citeauthoryear{Li, Dragisic, Faria, Ivanova,
  Jim{\'e}nez-Ruiz, Lambrix, and Pesquita}{Li et~al\mbox{.}}{2019}]%
        {li2019user}
\bibfield{author}{\bibinfo{person}{Huanyu Li}, \bibinfo{person}{Zlatan
  Dragisic}, \bibinfo{person}{Daniel Faria}, \bibinfo{person}{Valentina
  Ivanova}, \bibinfo{person}{Ernesto Jim{\'e}nez-Ruiz},
  \bibinfo{person}{Patrick Lambrix}, {and} \bibinfo{person}{Catia Pesquita}.}
  \bibinfo{year}{2019}\natexlab{}.
\newblock \showarticletitle{User validation in ontology alignment: functional
  assessment and impact}.
\newblock \bibinfo{journal}{\emph{The Knowledge Engineering Review}}
  \bibinfo{volume}{34} (\bibinfo{year}{2019}).
\newblock


\bibitem[\protect\citeauthoryear{Li, Li, Suhara, Wang, Hirota, and Tan}{Li
  et~al\mbox{.}}{2021}]%
        {li2021deep}
\bibfield{author}{\bibinfo{person}{Yuliang Li}, \bibinfo{person}{Jinfeng Li},
  \bibinfo{person}{Yoshihiko Suhara}, \bibinfo{person}{Jin Wang},
  \bibinfo{person}{Wataru Hirota}, {and} \bibinfo{person}{Wang-Chiew Tan}.}
  \bibinfo{year}{2021}\natexlab{}.
\newblock \showarticletitle{Deep Entity Matching: Challenges and
  Opportunities}.
\newblock \bibinfo{journal}{\emph{Journal of Data and Information Quality
  (JDIQ)}} \bibinfo{volume}{13}, \bibinfo{number}{1} (\bibinfo{year}{2021}),
  \bibinfo{pages}{1--17}.
\newblock


\bibitem[\protect\citeauthoryear{McCann, Shen, and Doan}{McCann
  et~al\mbox{.}}{2008}]%
        {McCann2008}
\bibfield{author}{\bibinfo{person}{Robert McCann}, \bibinfo{person}{Warren
  Shen}, {and} \bibinfo{person}{AnHai Doan}.} \bibinfo{year}{2008}\natexlab{}.
\newblock \showarticletitle{Matching schemas in online communities: A web 2.0
  approach}. In \bibinfo{booktitle}{\emph{2008 IEEE 24th international
  conference on data engineering}}. IEEE, \bibinfo{pages}{110--119}.
\newblock


\bibitem[\protect\citeauthoryear{Melnik, Garcia-Molina, and Rahm}{Melnik
  et~al\mbox{.}}{2002}]%
        {MELNIK2002}
\bibfield{author}{\bibinfo{person}{Sergey Melnik}, \bibinfo{person}{Hector
  Garcia-Molina}, {and} \bibinfo{person}{Erhard Rahm}.}
  \bibinfo{year}{2002}\natexlab{}.
\newblock \showarticletitle{Similarity flooding: A versatile graph matching
  algorithm and its application to schema matching}. In
  \bibinfo{booktitle}{\emph{ICDE}}. IEEE, \bibinfo{pages}{117--128}.
\newblock


\bibitem[\protect\citeauthoryear{Metcalfe and Finn}{Metcalfe and Finn}{2008}]%
        {Metcalfe2008}
\bibfield{author}{\bibinfo{person}{Janet Metcalfe} {and}
  \bibinfo{person}{Bridgid Finn}.} \bibinfo{year}{2008}\natexlab{}.
\newblock \showarticletitle{Evidence that judgments of learning are causally
  related to study choice}.
\newblock \bibinfo{journal}{\emph{Psychonomic Bulletin \& Review}}
  \bibinfo{volume}{15}, \bibinfo{number}{1} (\bibinfo{year}{2008}),
  \bibinfo{pages}{174--179}.
\newblock


\bibitem[\protect\citeauthoryear{Modica, Gal, and Jamil}{Modica
  et~al\mbox{.}}{2001}]%
        {MODICA2001}
\bibfield{author}{\bibinfo{person}{Giovanni Modica}, \bibinfo{person}{Avigdor
  Gal}, {and} \bibinfo{person}{Hasan~M Jamil}.}
  \bibinfo{year}{2001}\natexlab{}.
\newblock \showarticletitle{The use of machine-generated ontologies in dynamic
  information seeking}. In \bibinfo{booktitle}{\emph{International Conference
  on Cooperative Information Systems}}. Springer, \bibinfo{pages}{433--447}.
\newblock


\bibitem[\protect\citeauthoryear{Mudgal, Li, Rekatsinas, Doan, Park, Krishnan,
  Deep, Arcaute, and Raghavendra}{Mudgal et~al\mbox{.}}{2018}]%
        {mudgal2018deep}
\bibfield{author}{\bibinfo{person}{Sidharth Mudgal}, \bibinfo{person}{Han Li},
  \bibinfo{person}{Theodoros Rekatsinas}, \bibinfo{person}{AnHai Doan},
  \bibinfo{person}{Youngchoon Park}, \bibinfo{person}{Ganesh Krishnan},
  \bibinfo{person}{Rohit Deep}, \bibinfo{person}{Esteban Arcaute}, {and}
  \bibinfo{person}{Vijay Raghavendra}.} \bibinfo{year}{2018}\natexlab{}.
\newblock \showarticletitle{Deep Learning for Entity Matching: A Design Space
  Exploration}. In \bibinfo{booktitle}{\emph{Proceedings of the 2018
  International Conference on Management of Data}}. ACM,
  \bibinfo{pages}{19--34}.
\newblock


\bibitem[\protect\citeauthoryear{Nguyen, Nguyen, Mikl{\'o}s, Aberer, Gal, and
  Weidlich}{Nguyen et~al\mbox{.}}{2014}]%
        {Hung2014}
\bibfield{author}{\bibinfo{person}{Quoc Viet~Hung Nguyen},
  \bibinfo{person}{Thanh~Tam Nguyen}, \bibinfo{person}{Zolt{\'a}n Mikl{\'o}s},
  \bibinfo{person}{Karl Aberer}, \bibinfo{person}{Avigdor Gal}, {and}
  \bibinfo{person}{Matthias Weidlich}.} \bibinfo{year}{2014}\natexlab{}.
\newblock \showarticletitle{Pay-as-you-go reconciliation in schema matching
  networks}. In \bibinfo{booktitle}{\emph{ICDE}}. IEEE,
  \bibinfo{pages}{220--231}.
\newblock


\bibitem[\protect\citeauthoryear{Noy, Mortensen, Musen, and Alexander}{Noy
  et~al\mbox{.}}{2013}]%
        {NoyMMA13}
\bibfield{author}{\bibinfo{person}{Natalya~Fridman Noy},
  \bibinfo{person}{Jonathan Mortensen}, \bibinfo{person}{Mark~A. Musen}, {and}
  \bibinfo{person}{Paul~R. Alexander}.} \bibinfo{year}{2013}\natexlab{}.
\newblock \showarticletitle{Mechanical turk as an ontology engineer?: using
  microtasks as a component of an ontology-engineering workflow}. In
  \bibinfo{booktitle}{\emph{Web Science 2013, WebSci '13}}.
  \bibinfo{pages}{262--271}.
\newblock


\bibitem[\protect\citeauthoryear{Noy and Musen}{Noy and Musen}{2002}]%
        {Noy2002}
\bibfield{author}{\bibinfo{person}{Natalya~F Noy} {and} \bibinfo{person}{Mark~A
  Musen}.} \bibinfo{year}{2002}\natexlab{}.
\newblock \showarticletitle{Evaluating ontology-mapping tools: Requirements and
  experience}. In \bibinfo{booktitle}{\emph{Workshop on Evaluation of Ontology
  Tools at EKAW}}, Vol.~\bibinfo{volume}{2}. \bibinfo{pages}{p1--14}.
\newblock


\bibitem[\protect\citeauthoryear{Peukert, Eberius, and Rahm}{Peukert
  et~al\mbox{.}}{2011}]%
        {PEUKERT2011}
\bibfield{author}{\bibinfo{person}{Eric Peukert}, \bibinfo{person}{Julian
  Eberius}, {and} \bibinfo{person}{Erhard Rahm}.}
  \bibinfo{year}{2011}\natexlab{}.
\newblock \showarticletitle{AMC-A framework for modelling and comparing
  matching systems as matching processes}. In \bibinfo{booktitle}{\emph{2011
  IEEE 27th International Conference on Data Engineering}}. IEEE,
  \bibinfo{pages}{1304--1307}.
\newblock


\bibitem[\protect\citeauthoryear{Pinkel, Binnig, Kharlamov, and Haase}{Pinkel
  et~al\mbox{.}}{2013}]%
        {pinkel2013incmap}
\bibfield{author}{\bibinfo{person}{Christoph Pinkel}, \bibinfo{person}{Carsten
  Binnig}, \bibinfo{person}{Evgeny Kharlamov}, {and} \bibinfo{person}{Peter
  Haase}.} \bibinfo{year}{2013}\natexlab{}.
\newblock \showarticletitle{IncMap: pay as you go matching of relational
  schemata to OWL ontologies.}. In \bibinfo{booktitle}{\emph{OM}}. Citeseer,
  \bibinfo{pages}{37--48}.
\newblock


\bibitem[\protect\citeauthoryear{Rahm and Bernstein}{Rahm and
  Bernstein}{2001}]%
        {RAHM2001}
\bibfield{author}{\bibinfo{person}{Erhard Rahm} {and} \bibinfo{person}{Philip~A
  Bernstein}.} \bibinfo{year}{2001}\natexlab{}.
\newblock \showarticletitle{A survey of approaches to automatic schema
  matching}.
\newblock \bibinfo{journal}{\emph{the VLDB Journal}} \bibinfo{volume}{10},
  \bibinfo{number}{4} (\bibinfo{year}{2001}), \bibinfo{pages}{334--350}.
\newblock


\bibitem[\protect\citeauthoryear{Ross, Irani, Silberman, Zaldivar, and
  Tomlinson}{Ross et~al\mbox{.}}{2010}]%
        {Ross2010}
\bibfield{author}{\bibinfo{person}{Joel Ross}, \bibinfo{person}{Lilly Irani},
  \bibinfo{person}{M Silberman}, \bibinfo{person}{Andrew Zaldivar}, {and}
  \bibinfo{person}{Bill Tomlinson}.} \bibinfo{year}{2010}\natexlab{}.
\newblock \showarticletitle{Who are the crowdworkers?: shifting demographics in
  mechanical turk}. In \bibinfo{booktitle}{\emph{CHI'10 Extended Abstracts on
  Human Factors in Computing Systems}}. ACM, \bibinfo{pages}{2863--2872}.
\newblock


\bibitem[\protect\citeauthoryear{Sarasua, Simperl, and Noy}{Sarasua
  et~al\mbox{.}}{2012}]%
        {Crowdmap}
\bibfield{author}{\bibinfo{person}{C. Sarasua}, \bibinfo{person}{E. Simperl},
  {and} \bibinfo{person}{N.~F Noy}.} \bibinfo{year}{2012}\natexlab{}.
\newblock \showarticletitle{Crowdmap: Crowdsourcing ontology alignment with
  microtasks}. In \bibinfo{booktitle}{\emph{ISWC}}.
\newblock


\bibitem[\protect\citeauthoryear{Shraga, Gal, and Roitman}{Shraga
  et~al\mbox{.}}{2018}]%
        {HILDA18}
\bibfield{author}{\bibinfo{person}{Roee Shraga}, \bibinfo{person}{Avigdor Gal},
  {and} \bibinfo{person}{Haggai Roitman}.} \bibinfo{year}{2018}\natexlab{}.
\newblock \showarticletitle{What Type of a Matcher Are You?: Coordination of
  Human and Algorithmic Matchers}. In \bibinfo{booktitle}{\emph{Proceedings of
  the Workshop on Human-In-the-Loop Data Analytics, HILDA@SIGMOD}}.
  \bibinfo{pages}{12:1--12:7}.
\newblock


\bibitem[\protect\citeauthoryear{Shraga, Gal, and Roitman}{Shraga
  et~al\mbox{.}}{2020}]%
        {shraga2020}
\bibfield{author}{\bibinfo{person}{Roee Shraga}, \bibinfo{person}{Avigdor Gal},
  {and} \bibinfo{person}{Haggai Roitman}.} \bibinfo{year}{2020}\natexlab{}.
\newblock \showarticletitle{ADnEV: Cross-Domain Schema Matching using Deep
  Similarity Matrix Adjustment and Evaluation}.
\newblock \bibinfo{journal}{\emph{Proceedings of the VLDB Endowment}}
  \bibinfo{volume}{13}, \bibinfo{number}{9} (\bibinfo{year}{2020}),
  \bibinfo{pages}{1401--1415}.
\newblock


\bibitem[\protect\citeauthoryear{Singh, Meduri, Elmagarmid, Madden, Papotti,
  Quian{\'e}-Ruiz, Solar-Lezama, and Tang}{Singh et~al\mbox{.}}{2017}]%
        {singh2017synthesizing}
\bibfield{author}{\bibinfo{person}{Rohit Singh},
  \bibinfo{person}{Venkata~Vamsikrishna Meduri}, \bibinfo{person}{Ahmed
  Elmagarmid}, \bibinfo{person}{Samuel Madden}, \bibinfo{person}{Paolo
  Papotti}, \bibinfo{person}{Jorge-Arnulfo Quian{\'e}-Ruiz},
  \bibinfo{person}{Armando Solar-Lezama}, {and} \bibinfo{person}{Nan Tang}.}
  \bibinfo{year}{2017}\natexlab{}.
\newblock \showarticletitle{Synthesizing entity matching rules by examples}.
\newblock \bibinfo{journal}{\emph{Proceedings of the VLDB Endowment}}
  \bibinfo{volume}{11}, \bibinfo{number}{2} (\bibinfo{year}{2017}),
  \bibinfo{pages}{189--202}.
\newblock


\bibitem[\protect\citeauthoryear{Thirumuruganathan, Tang, Ouzzani, and
  Doan}{Thirumuruganathan et~al\mbox{.}}{2020}]%
        {thirumuruganathan2020data}
\bibfield{author}{\bibinfo{person}{Saravanan Thirumuruganathan},
  \bibinfo{person}{Nan Tang}, \bibinfo{person}{Mourad Ouzzani}, {and}
  \bibinfo{person}{AnHai Doan}.} \bibinfo{year}{2020}\natexlab{}.
\newblock \showarticletitle{Data Curation with Deep Learning}. In
  \bibinfo{booktitle}{\emph{EDBT}}. \bibinfo{pages}{277--286}.
\newblock


\bibitem[\protect\citeauthoryear{Wang, Shea, Wang, and Wu}{Wang
  et~al\mbox{.}}{2019}]%
        {wang2019progressive}
\bibfield{author}{\bibinfo{person}{Pei Wang}, \bibinfo{person}{Ryan Shea},
  \bibinfo{person}{Jiannan Wang}, {and} \bibinfo{person}{Eugene Wu}.}
  \bibinfo{year}{2019}\natexlab{}.
\newblock \showarticletitle{Progressive Deep Web Crawling Through Keyword
  Queries For Data Enrichment}. In \bibinfo{booktitle}{\emph{Proceedings of the
  2019 International Conference on Management of Data}}.
  \bibinfo{pages}{229--246}.
\newblock


\bibitem[\protect\citeauthoryear{Willmott and Matsuura}{Willmott and
  Matsuura}{2005}]%
        {willmott2005advantages}
\bibfield{author}{\bibinfo{person}{Cort~J Willmott} {and}
  \bibinfo{person}{Kenji Matsuura}.} \bibinfo{year}{2005}\natexlab{}.
\newblock \showarticletitle{Advantages of the mean absolute error (MAE) over
  the root mean square error (RMSE) in assessing average model performance}.
\newblock \bibinfo{journal}{\emph{Climate research}} \bibinfo{volume}{30},
  \bibinfo{number}{1} (\bibinfo{year}{2005}), \bibinfo{pages}{79--82}.
\newblock


\bibitem[\protect\citeauthoryear{Zhang, Chen, Jagadish, Zhang, and Tong}{Zhang
  et~al\mbox{.}}{2018}]%
        {zhang2018reducing}
\bibfield{author}{\bibinfo{person}{Chen Zhang}, \bibinfo{person}{Lei Chen},
  \bibinfo{person}{HV Jagadish}, \bibinfo{person}{Mengchen Zhang}, {and}
  \bibinfo{person}{Yongxin Tong}.} \bibinfo{year}{2018}\natexlab{}.
\newblock \showarticletitle{Reducing Uncertainty of Schema Matching via
  Crowdsourcing with Accuracy Rates}.
\newblock \bibinfo{journal}{\emph{IEEE Transactions on Knowledge and Data
  Engineering}} (\bibinfo{year}{2018}).
\newblock


\bibitem[\protect\citeauthoryear{Zhang, Chen, Jagadish, and Cao}{Zhang
  et~al\mbox{.}}{2013}]%
        {Zhang2013}
\bibfield{author}{\bibinfo{person}{Chen~Jason Zhang}, \bibinfo{person}{Lei
  Chen}, \bibinfo{person}{H.~V. Jagadish}, {and} \bibinfo{person}{Caleb~Chen
  Cao}.} \bibinfo{year}{2013}\natexlab{}.
\newblock \showarticletitle{Reducing Uncertainty of Schema Matching via
  Crowdsourcing}.
\newblock \bibinfo{journal}{\emph{{PVLDB}}} \bibinfo{volume}{6},
  \bibinfo{number}{9} (\bibinfo{year}{2013}), \bibinfo{pages}{757--768}.
\newblock


\bibitem[\protect\citeauthoryear{Zhang and Ives}{Zhang and Ives}{2020}]%
        {zhang2020finding}
\bibfield{author}{\bibinfo{person}{Yi Zhang} {and} \bibinfo{person}{Zachary~G
  Ives}.} \bibinfo{year}{2020}\natexlab{}.
\newblock \showarticletitle{Finding Related Tables in Data Lakes for
  Interactive Data Science}. In \bibinfo{booktitle}{\emph{Proceedings of the
  2020 ACM SIGMOD International Conference on Management of Data}}.
  \bibinfo{pages}{1951--1966}.
\newblock


\end{thebibliography}
\appendix

\section{Monotonic Evaluation and Section~\ref{sec:matching_eval} Proofs}
\label{sec:proofs}
The appendix is devoted to the proofs of Section~\ref{sec:matching_eval}.

\setcounter{theorem}{0}
\begin{theorem}\label{thm:MIEMapp}
	Recall ($R$) is a MIEM over $\Sigma^{\subseteq}$, Precision ($P$) is a MIEM over $\Sigma^P$, and \fmspace ($F$) is a MIEM over $\Sigma^F$.
\end{theorem}

For proving Theorem~\ref{thm:MIEM}, we use two lemmas stating that recall is a MIEM over all match pairs in $\Sigma^{\subseteq}$ and, since precision and \fmspace are not monotonic over the full set of pairs in $\Sigma^{\subseteq}$, the conditions under which monotonicity can be guaranteed for both measures. 

\begin{lemma}\label{lemma:r_miem}
	Recall ($R$) is a MIEM over $\Sigma^{\subseteq}$.
\end{lemma}

\begin{proof}[Proof of Lemma \ref{lemma:r_miem}]
	Let $(\sigma, \sigma^{\prime})\in\Sigma^{\subseteq}$ be a match pair in $\Sigma^{\subseteq}$. Using Eq.~\ref{eq:PandR}, one can compute recall of $\sigma$ and $\sigma^{\prime}$, as follows.
	$$R(\sigma)=\frac{\mid\sigma\cap \sigma^{*}\mid}{\mid \sigma^{*}\mid}, R(\sigma^{\prime})=\frac{\mid\sigma^{\prime}\cap \sigma^{*}\mid}{\mid \sigma^{*}\mid}$$
	
	$\sigma\subseteq\sigma^{\prime}$ and thus $(\sigma\cap \sigma^{*})\subseteq(\sigma^{\prime}\cap \sigma^{*})$ and $|\sigma\cap \sigma^{*}|\leq|\sigma^{\prime}\cap \sigma^{*}|$. Noting that the denominator is not affected by the addition, we obtain:
	$$R(\sigma) = \frac{\mid\sigma\cap \sigma^{*}\mid}{\mid \sigma^{*}\mid}\leq \frac{\mid\sigma^{\prime}\cap \sigma^{*}\mid}{\mid \sigma^{*}\mid} = R(\sigma^{\prime})$$ and thus $R(\sigma) \leq R(\sigma^{\prime})$.
\end{proof}
\begin{lemma}\label{lemma:p_f_ineq}
	For $(\sigma, \sigma^{\prime})\in\Sigma^{\subseteq}$:
	\begin{itemize}
		\item $P(\sigma) \leq P(\sigma^{\prime})$ iff $P(\sigma)\leq P(\deltacorr)$
		\item $F(\sigma) \leq F(\sigma^{\prime})$ iff $0.5\cdot F(\sigma)\leq P(\deltacorr)$
	\end{itemize}
\end{lemma}

\begin{proof}[Proof of Lemma \ref{lemma:p_f_ineq}]\label{proof:p_f_ineq}
	Let $(\sigma, \sigma^{\prime})\in\Sigma^{\subseteq}$ be a match pair in $\Sigma^{\subseteq}$. 
	We shall begin with two extreme cases in which the denominator of a precision calculation is zero. First, in case $\sigma = \emptyset$, $P(\sigma)$ is undefined. Accordingly, for this work, we shall define $P(\sigma) = P(\deltacorr)$, \emph{i.e.,} the prior precision value does not change. Second, in the case $\sigma^{\prime} = \sigma$, we obtain that $\sigma^{\prime}\setminus\sigma = \emptyset$ resulting in an undefined $P(\deltacorr)$. Similarly, we define $P(\deltacorr) = P(\sigma)$, \emph{i.e.,} expanding a match with nothing does not ``harm'' its precision. In both cases we obtain $P(\sigma) = P(\deltacorr) = P(\sigma^{\prime})$ and the first statement of the lemma holds.
	
	Next, we refer to the general case where $\sigma\neq\emptyset$ and $\deltacorr\neq \emptyset$ and note that $|\sigma^{\prime}| = |\sigma\cup\deltacorr|$. Assuming $\sigma\neq\sigma^{\prime}$ and recalling that $\sigma\subseteq\sigma^{\prime}$ ($(\sigma, \sigma^{\prime})\in\Sigma^{\subseteq}$), we obtain that $\sigma\cap\deltacorr = \emptyset$ and 
	\begin{equation}\label{eq:m}
		|\sigma^{\prime}| = |\sigma\cup\deltacorr| = |\sigma| + |\deltacorr|
	\end{equation}
	Similarly, we obtain 
	\begin{equation}
		|\sigma^{\prime}\cap \sigma^{*}| = |(\sigma\cap \sigma^{*})\cup(\deltacorr\cap \sigma^{*})|
	\end{equation}
	Recalling that $\sigma\subseteq\sigma^{\prime}$, we have $(\sigma\cap \sigma^{*})\cap(\deltacorr\cap \sigma^{*}) = \emptyset$ and
	\begin{equation}\label{eq:tp}
		|\sigma^{\prime}\cap \sigma^{*}| = |\sigma\cap \sigma^{*}| + |\deltacorr\cap \sigma^{*}|
	\end{equation}
	
	\begin{itemize}
		\item $\mathbf{P(\sigma) \leq P(\sigma^{\prime})}$ \textbf{iff} $\mathbf{P(\sigma)\leq P(\deltacorr)}$:\\
		Using Eq.~\ref{eq:PandR}, the precision values of $\sigma$, $\sigma^{\prime}$, and $\deltacorr$ are computed as follows:
		\begin{equation}\label{eq:p_for_proof}
			P(\sigma) = \frac{|\sigma\cap \sigma^{*}|}{|\sigma|},  P(\sigma^{\prime}) = \frac{|\sigma^{\prime}\cap \sigma^{*}|}{|\sigma^{\prime}|}, P(\deltacorr) = \frac{|\deltacorr\cap \sigma^{*}|}{|\deltacorr|}
		\end{equation}
		
		$\Rightarrow$: Let $P(\sigma) \leq P(\sigma^{\prime})$. Using Eq.~\ref{eq:p_for_proof}, we obtain that 
		$$\frac{|\sigma\cap \sigma^{*}|}{|\sigma|} \leq \frac{|\sigma^{\prime}\cap \sigma^{*}|}{|\sigma^{\prime}|}$$
		and accordingly (Eq.~\ref{eq:tp} and Eq.~\ref{eq:m}), 
		$$\frac{|\sigma\cap \sigma^{*}|}{|\sigma|}  \leq \frac{|\sigma\cap \sigma^{*}| + |\deltacorr\cap \sigma^{*}|}{|\sigma| + |\deltacorr|}$$
		Then, we obtain 
		$$|\sigma\cap \sigma^{*}|(|\sigma| + |\deltacorr|) \leq |\sigma|(|\sigma\cap \sigma^{*}| + |\deltacorr\cap \sigma^{*}|)$$ 
		and following, 
		$$|\sigma\cap \sigma^{*}|\cdot |\sigma| + |\sigma\cap \sigma^{*}|\cdot|\deltacorr| \leq |\sigma|\cdot |\sigma\cap \sigma^{*}| + |\sigma|\cdot|\deltacorr\cap \sigma^{*}|$$ 
		Finally, we get 
		$$|\sigma\cap \sigma^{*}|\cdot|\deltacorr| \leq |\sigma|\cdot|\deltacorr\cap \sigma^{*}|$$ 
		and conclude
		$$P(\sigma) = \frac{|\sigma\cap \sigma^{*}|}{|\sigma|} \leq \frac{|\deltacorr\cap \sigma^{*}|}{|\deltacorr|} = P(\deltacorr)$$ 
		
		$\Leftarrow$: Let $P(\sigma) \leq P(\deltacorr)$. Using Eq.~\ref{eq:p_for_proof}, we obtain that 
		$$\frac{|\sigma\cap \sigma^{*}|}{|\sigma|} \leq \frac{|\deltacorr\cap \sigma^{*}|}{|\deltacorr|} $$ 
		and following 
		$$|\sigma\cap \sigma^{*}|\cdot |\deltacorr| \leq |\sigma|\cdot |\deltacorr\cap \sigma^{*}|$$
		By adding $|\sigma\cap \sigma^{*}|\cdot |\sigma|$ on both sides we get 
		$$|\sigma\cap \sigma^{*}|\cdot |\deltacorr| + |\sigma\cap \sigma^{*}|\cdot |\sigma|\leq |\sigma|\cdot |\deltacorr\cap \sigma^{*}| + |\sigma\cap \sigma^{*}|\cdot |\sigma|$$
		After rewriting, we obtain 
		$$|\sigma\cap \sigma^{*}|\cdot(|\deltacorr| + |\sigma|)\leq |\sigma|\cdot (|\deltacorr\cap \sigma^{*}| + |\sigma\cap \sigma^{*}|)$$ 
		and in what follows (Eq.~\ref{eq:tp} and Eq.~\ref{eq:m}):
		$$P(\sigma) = \frac{|\sigma\cap \sigma^{*}|}{|\sigma|} \leq \frac{|\sigma^{\prime}\cap \sigma^{*}|}{|\sigma^{\prime}|} = P(\sigma^{\prime})$$
		
		\item $\mathbf{F(\sigma) \leq F(\sigma^{\prime})}$ \textbf{iff} $\mathbf{0.5F(\sigma)\leq P(\deltacorr)}$:\\
		The F1 measure value of $\sigma$ is given by 
		$$F(\sigma) = 2\cdot\frac{P(\sigma)\cdot R(\sigma)}{P(\sigma) + R(\sigma)}$$
		Then, using Eq.~\ref{eq:PandR}, we have 
		$$F(\sigma) = 2\cdot\frac{\frac{|\sigma\cap \sigma^{*}|}{|\sigma|}\cdot \frac{|\sigma\cap \sigma^{*}|}{|\sigma^{*}|}}{\frac{|\sigma\cap \sigma^{*}|}{|\sigma|} + \frac{|\sigma\cap \sigma^{*}|}{|\sigma^{*}|}}$$ 
		and after rewriting we obtain 
		\begin{equation}\label{eq:f_sigma}
			F(\sigma) = \frac{2\cdot |\sigma\cap \sigma^{*}|}{|\sigma| + |\sigma^{*}|}
		\end{equation}
		Similarly, we can compute 
		$$F(\sigma^{\prime}) = \frac{2\cdot |\sigma^{\prime}\cap \sigma^{*}|}{|\sigma^{\prime}| + |\sigma^{*}|}$$
		Using Eq.~\ref{eq:tp} and Eq.~\ref{eq:m} , we further obtain 
		\begin{equation}\label{eq:f_sigma_prime}
			F(\sigma^{\prime}) = \frac{2\cdot (|\sigma\cap \sigma^{*}| + |\deltacorr\cap \sigma^{*}|)}{(|\sigma| + |\deltacorr|) + |\sigma^{*}|}
		\end{equation}
		
		$\Rightarrow$: Let $F(\sigma) \leq F(\sigma^{\prime})$. Using Eq.~\ref{eq:f_sigma} and Eq.~\ref{eq:f_sigma_prime}, we obtain
		$$\frac{2\cdot |\sigma\cap \sigma^{*}|}{|\sigma| + |\sigma^{*}|} \leq \frac{2\cdot (|\sigma\cap \sigma^{*}| + |\deltacorr\cap \sigma^{*}|)}{(|\sigma| + |\deltacorr|) + |\sigma^{*}|}$$
		after rewriting we get 
		$$|\sigma\cap \sigma^{*}|\cdot (|\sigma|+|\sigma^{*}|+|\deltacorr|) \leq (|\sigma\cap \sigma^{*}| + |\deltacorr\cap \sigma^{*}|)\cdot(|\sigma| + |\sigma^{*}|)$$
		and following
		$$|\sigma\cap \sigma^{*}|\cdot|\deltacorr| \leq |\deltacorr\cap \sigma^{*}|\cdot(|\sigma| + |\sigma^{*}|) \rightarrow \frac{|\sigma\cap \sigma^{*}|}{|\sigma| + |\sigma^{*}|}\leq\frac{|\deltacorr\cap \sigma^{*}|}{|\deltacorr|}$$
		Recalling that $P(\deltacorr) = \frac{|\deltacorr\cap \sigma^{*}|}{|\deltacorr|}$ (Eq.~\ref{eq:p_for_proof}) we conclude 
		$$0.5F(\sigma) =\frac{|\sigma\cap \sigma^{*}|}{|\sigma| + |\sigma^{*}|}\leq\frac{|\deltacorr\cap \sigma^{*}|}{|\deltacorr|} = P(\deltacorr)$$
		
		$\Leftarrow$: Let $0.5F(\sigma)\leq P(\deltacorr)$. Using Eq.~\ref{eq:f_sigma} and Eq.~\ref{eq:p_for_proof}, we obtain
		$$0.5\frac{2|\sigma\cap \sigma^{*}|}{|\sigma| + |\sigma^{*}|}\leq\frac{|\deltacorr\cap \sigma^{*}|}{|\deltacorr|}$$
		and accordingly,
		$$|\sigma\cap \sigma^{*}|\cdot |\deltacorr| \leq |\deltacorr\cap \sigma^{*}|\cdot(|\sigma| + |\sigma^{*}|)$$
		By adding $|\sigma\cap \sigma^{*}|\cdot (|\sigma| + |\sigma^{*}|)$ on both sides we get
		
		\begin{footnotesize}
			$$|\sigma\cap \sigma^{*}|\cdot |\deltacorr| + |\sigma\cap \sigma^{*}|\cdot (|\sigma| + |\sigma^{*}|)\leq |\deltacorr\cap \sigma^{*}|\cdot(|\sigma| + |\sigma^{*}|) + |\sigma\cap \sigma^{*}|\cdot (|\sigma| + |\sigma^{*}|)$$
		\end{footnotesize} 
		After rewriting, we obtain 
		$$|\sigma\cap \sigma^{*}|\cdot (|\sigma| + |\sigma^{*}| + |\deltacorr|) \leq (|\sigma\cap \sigma^{*}| + |\deltacorr\cap \sigma^{*}|)\cdot(|\sigma| + |\sigma^{*}|)$$
		Then, dividing both sides by $0.5\cdot(|\sigma| + |\sigma^{*}| + |\deltacorr|)\cdot(|\sigma| + |\sigma^{*}|)$, we get
		$$F(\sigma) = \frac{2\cdot |\sigma\cap \sigma^{*}|}{|\sigma| + |\sigma^{*}|} \leq \frac{2\cdot (|\sigma\cap \sigma^{*}| + |\deltacorr\cap \sigma^{*}|)}{(|\sigma| + |\deltacorr|) + |\sigma^{*}|} = F(\sigma^{\prime})$$
	\end{itemize}
	which concludes the proof.		
\end{proof}

\setcounter{prop}{0}
\begin{prop}\label{prop:annealer}
	Let $G$ be an evaluation measure. If $G$ is a MIEM over $\Sigma^2\subseteq\Sigma^{\subseteq_1}$, then $\forall(\sigma,\sigma^{\prime})\in\Sigma^2$~: $\deltacorr=\sigma^{\prime}\setminus\sigma$ is a local annealer with respect to $G$ over $\Sigma^2\subseteq\Sigma^{\subseteq_1}$.
\end{prop}


\begin{proof}[Proof of Proposition \ref{prop:annealer}]
	Let $G$ be an MIEM over $\Sigma^2\subseteq\Sigma^{\subseteq_1}$. By Definition~\ref{def:monotonicity}, 
	$G(\sigma) \leq G(\sigma^{\prime})$
	holds for every match pair $(\sigma, \sigma^{\prime})\in\Sigma^2$.
	
	Assume, by way of contradiction, that exists some $\deltacorr$ that is not a local annealer with respect to $G$ over $\Sigma^2$. Thus, there exists some match pair $(\sigma, \sigma^{\prime})\in\Sigma^2$ such that $\deltacorr=\deltacorr_{(\sigma,\sigma^{\prime})}$ and $G(\sigma) > G(\sigma^{\prime})$, in contradiction to the fact that $G$ is a MIEM over $\Sigma^2$.
\end{proof}

\setcounter{corol}{0}
\begin{corol}\label{corol:annealer}
	Any singleton correspondence set $\deltacorr$ ($|\deltacorr|=1$) is a local annealer with respect to 1) $R$ over $\Sigma^{\subseteq_1}$, 2) $P$ over $\Sigma^P\cap\Sigma^{\subseteq_1}$, and 3) $F$ over $\Sigma^F\cap\Sigma^{\subseteq_1}$.
\end{corol}
\noindent where $\Sigma^P$ and $\Sigma^F$ are the subspaces for which precision and \fmspace are monotonic as defined in Section~\ref{sec:monotonicEvaluation}.

\begin{proof}[Proof of Corollary \ref{corol:annealer}]
	Note that $\Sigma^{\subseteq_1}\subseteq\Sigma^{\subseteq}$, and accordingly also $\Sigma^P\cap\Sigma^{\subseteq_1}\subseteq\Sigma^P$ and $\Sigma^F\cap\Sigma^{\subseteq_1}\subseteq\Sigma^F$. Using Theorem~\ref{thm:MIEM} we can therefore say that recall ($R$) is a MIEM over $\Sigma^{\subseteq_1}$, precision ($P$) is a MIEM over $\Sigma^P\cap\Sigma^{\subseteq_1}$, and F1 measure ($F$) is a MIEM over $\Sigma^F\cap\Sigma^{\subseteq_1}$.
	
	Then using Proposition~\ref{prop:annealer}, we can conclude that for all $\deltacorr_{\sigma,\sigma^{\prime}}$ s.t. $(\sigma,\sigma^{\prime})\in\Sigma^{\subseteq_1}/\Sigma^P\cap\Sigma^{\subseteq_1}/\Sigma^F\cap\Sigma^{\subseteq_1}$, $\deltacorr_{\sigma,\sigma^{\prime}}$ is a local annealer with respect to $R$ over $\Sigma^{\subseteq_1}$/$P$ over $\Sigma^P\cap\Sigma^{\subseteq_1}$/$F$ over $\Sigma^F\cap\Sigma^{\subseteq_1}$. For any other singleton $\deltacorr$, the claim is vacuously satisfied.
\end{proof}

\setcounter{lemma}{2}
\begin{lemma}\label{lemma:p_f_ineqProb}
	For $(\sigma, \sigma^{\prime})\in\Sigma^{\subseteq_1}$:
	\begin{itemize}
		\item $E(P(\sigma)) \leq E(P(\sigma^{\prime}))$ iff $E(P(\sigma))\leq Pr\{\deltacorr\in\sigma^*\}$
		\item $E(F(\sigma)) \leq E(F(\sigma^{\prime}))$ iff $0.5\cdot E(F(\sigma))\leq Pr\{\deltacorr\in\sigma^*\}$
	\end{itemize}
\end{lemma}

\begin{proof}[Proof of Lemma \ref{lemma:p_f_ineqProb}]
	
	Let $(\sigma, \sigma^{\prime})\in\Sigma^{\subseteq_1}$ be a match pair in $\Sigma^{\subseteq_1}$. 
	
	We first address an extreme case, where the denominator of a precision calculation is zero. Here, this case occurs only when $\sigma = \sigma^{\prime} = \emptyset$. For this work, we shall define $E(P(\sigma)) = E(P(\sigma^{\prime})) = -1$, ensuring the validity of the first statement of the lemma.
	
	Next, we analyze the expected values of the evaluation measures. We shall assume that the size of the current match $\sigma$ is deterministically known and therefore $E(|\sigma|) = |\sigma|$. Let $|\sigma^{*}|$ and $|\sigma\cap \sigma^{*}|$ be random variables with expected values of $E(|\sigma^{*}|)$ and $E(|\sigma\cap \sigma^{*}|)$, respectively. 
	\begin{itemize}
		\item $\mathbf{E(P(\sigma)) \leq E(P(\sigma^{\prime}))}$ \textbf{iff} $\mathbf{E(P(\sigma))\leq Pr\{\deltacorr\in\sigma^*\}}$:
		Similar to Eq.~\ref{eq:PandR}, we compute the expected precision value of $\sigma$ as follows:
		\begin{equation}\label{eq:EP}
			E(P(\sigma)) = \frac{E(|\sigma\cap \sigma^{*}|)}{|\sigma|}
		\end{equation}
		Now, we are ready to compute the expected precision value of $\sigma^{\prime}$. Note that the value of the denominator is deterministic, $E(|\sigma^{\prime}|) = |\sigma| + 1$ since $(\sigma, \sigma^{\prime})\in\Sigma^{\subseteq_1}$. Then, using $Pr\{\deltacorr\in\sigma^*\}$ and Eq.~\ref{eq:EP}, we obtain
		$${\small E(P(\sigma^{\prime})) = Pr\{\deltacorr\in\sigma^*\}\cdot \frac{E(|\sigma\cap \sigma^{*}|) + 1}{|\sigma| + 1} + (1-Pr\{\deltacorr\in\sigma^*\})\cdot \frac{E(|\sigma\cap \sigma^{*}|)}{|\sigma| + 1}}$$
		While the denominator remains unchanged, the numerator increases by one only if $\deltacorr$ is part of the reference match. After rewriting we obtain 
		$${\small  E(P(\sigma^{\prime})) = \frac{Pr\{\deltacorr\in\sigma^*\}\cdot E(|\sigma\cap \sigma^{*}|) + Pr\{\deltacorr\in\sigma^*\} + E(|\sigma\cap \sigma^{*}|) - Pr\{\deltacorr\in\sigma^*\}\cdot E(|\sigma\cap \sigma^{*}|)}{|\sigma| + 1}}$$
		and conclude that
		\begin{equation}\label{eq:EP_prime}
			E(P(\sigma^{\prime})) = \frac{E(|\sigma\cap \sigma^{*}|) + Pr\{\deltacorr\in\sigma^*\}}{|\sigma| + 1}
		\end{equation} 
		
		$\Rightarrow$: Let $E(P(\sigma)) \leq E(P(\sigma^{\prime}))$. Using Eq.~\ref{eq:EP} and Eq.~\ref{eq:EP_prime} we obtain
		$$\frac{E(|\sigma\cap \sigma^{*}|)}{|\sigma|}\leq \frac{E(|\sigma\cap \sigma^{*}|) + Pr\{\deltacorr\in\sigma^*\}}{|\sigma| + 1}$$
		Multiplying by $|\sigma|\cdot(|\sigma| + 1)$, we get
		$$E(|\sigma\cap \sigma^{*}|)\cdot (|\sigma| + 1) \leq (E(|\sigma\cap \sigma^{*}|) + Pr\{\deltacorr\in\sigma^*\})\cdot|\sigma|$$
		Using rewriting we obtain 
		$$E(|\sigma\cap \sigma^{*}|)\cdot|\sigma| + E(|\sigma\cap \sigma^{*}|) \leq E(|\sigma\cap \sigma^{*}|)\cdot|\sigma| + Pr\{\deltacorr\in\sigma^*\}\cdot|\sigma|\rightarrow$$
		$$E(|\sigma\cap \sigma^{*}|)\leq Pr\{\deltacorr\in\sigma^*\}\cdot|\sigma|$$
		and conclude that
		$$E(P(\sigma)) = \frac{E(|\sigma\cap \sigma^{*}|)}{|\sigma|}\leq Pr\{\deltacorr\in\sigma^*\}$$
		
		$\Leftarrow$: Let $E(P(\sigma)) \leq Pr\{\deltacorr\in\sigma^*\}$. Using Eq.~\ref{eq:EP} we obtain
		$$\frac{E(|\sigma\cap \sigma^{*}|)}{|\sigma|}\leq Pr\{\deltacorr\in\sigma^*\}\rightarrow E(|\sigma\cap \sigma^{*}|)\leq Pr\{\deltacorr\in\sigma^*\}\cdot |\sigma|$$
		by adding $E(|\sigma\cap \sigma^{*}|)\cdot|\sigma|$ on both sides and rewriting we conclude that
		$$E(P(\sigma)) = \frac{E(|\sigma\cap \sigma^{*}|)}{|\sigma|}\leq \frac{E(|\sigma\cap \sigma^{*}|) + Pr\{\deltacorr\in\sigma^*\}}{|\sigma| + 1} = E(P(\sigma^{\prime}))$$
		
		\item $\mathbf{E(F(\sigma)) \leq E(F(\sigma^{\prime}))}$ \textbf{iff} $\mathbf{0.5\cdot E(F(\sigma))\leq Pr\{\deltacorr\in\sigma^*\}}$:
		
		The expected F1 measure value of $\sigma$ is given by
		
		\begin{equation}\label{eq:expected_f_sigma}
			E(F(\sigma)) = \frac{2\cdot E(|\sigma\cap \sigma^{*}|)}{|\sigma| + E(|\sigma^{*}|)}
		\end{equation}
		Similar to the computation of the precision value, since the size of the match is deterministic, $E(|\sigma^{\prime}|) = |\sigma| + 1$ and 
		$$E(F(\sigma^{\prime})) = \frac{2\cdot E(|\sigma^{\prime}\cap \sigma^{*}|)}{E(|\sigma^{\prime}|) + E(|\sigma^{*}|)} = \frac{2\cdot E(|\sigma^{\prime}\cap \sigma^{*}|)}{|\sigma| + 1 + E(|\sigma^{*}|)}$$
		The denominator value is deterministically affected by the addition of $\deltacorr$ (and increased by one). However, the nominator depends on $Pr\{\deltacorr\in\sigma^*\}$ and accordingly, we can rewrite it as
		$$Pr\{\deltacorr\in\sigma^*\}\cdot\frac{2\cdot (E(|\sigma\cap \sigma^{*}|) + 1)}{|\sigma| + 1 + E(|\sigma^{*}|)} + (1 - Pr\{\deltacorr\in\sigma^*\})\cdot\frac{2\cdot E(|\sigma\cap \sigma^{*}|)}{|\sigma| + 1 + E(|\sigma^{*}|)}$$
		after rewriting we obtain 
		\begin{equation}\label{eq:expected_f_sigma_prime}
			E(F(\sigma^{\prime})) = \frac{2\cdot(E(|\sigma\cap \sigma^{*}|) + Pr\{\deltacorr\in\sigma^*\})}{|\sigma| + 1 + E(|\sigma^{*}|}
		\end{equation}
		
		$\Rightarrow$: Let $E(F(\sigma)) \leq E(F(\sigma^{\prime}))$. Using Eq.~\ref{eq:expected_f_sigma} and Eq.~\ref{eq:expected_f_sigma_prime} we obtain
		$$\frac{2\cdot E(|\sigma\cap \sigma^{*}|)}{|\sigma| + E(|\sigma^{*}|)} \leq \frac{2\cdot(E(|\sigma\cap \sigma^{*}|) + Pr\{\deltacorr\in\sigma^*\})}{|\sigma| + 1 + E(|\sigma^{*}|)}$$
		multiplying by $(|\sigma| + E(|\sigma^{*}|))\cdot (|\sigma| + E(|\sigma^{*}|) + 1)$ yields
		$${\small 2\cdot E(|\sigma\cap \sigma^{*}|)\cdot (|\sigma| + E(|\sigma^{*}|) + 1) \leq 2\cdot(E(|\sigma\cap \sigma^{*}|) + Pr\{\deltacorr\in\sigma^*\})\cdot (|\sigma| + E(|\sigma^{*}|))}$$
		and by rewriting we get
		$$E(|\sigma\cap \sigma^{*}|) \leq Pr\{\deltacorr\in\sigma^*\})\cdot |\sigma| + Pr\{\deltacorr\in\sigma^*\})\cdot E(|\sigma^{*}|)$$
		and conclude
		$$0.5\cdot E(F(\sigma)) = 0.5\cdot\frac{2\cdot E(|\sigma\cap \sigma^{*}|)}{|\sigma| + E(|\sigma^{*}|)} \leq Pr\{\deltacorr\in\sigma^*\}$$
		
		$\Leftarrow$: Let $0.5\cdot E(F(\sigma))\leq Pr\{\deltacorr\in\sigma^*\}$. Using Eq.~\ref{eq:expected_f_sigma} we obtain
		$$0.5\cdot\frac{2\cdot E(|\sigma\cap \sigma^{*}|)}{|\sigma| + E(|\sigma^{*}|)}\leq Pr\{\deltacorr\in\sigma^*\}$$
		by multiplying by $|\sigma| + E(|\sigma^{*}|)$ we get
		$$E(|\sigma\cap \sigma^{*}|)\leq Pr\{\deltacorr\in\sigma^*\}\cdot |\sigma| + Pr\{\deltacorr\in\sigma^*\}\cdot E(|\sigma^{*}|)$$
		Similar to above, we add $E(|\sigma\cap \sigma^{*}|)\cdot (|\sigma| + E(|\sigma^{*}|))$ on both sides and get
		$$E(|\sigma\cap \sigma^{*}|) + E(|\sigma\cap \sigma^{*}|)\cdot (|\sigma| + E(|\sigma^{*}|))\leq $$
		$$Pr\{\deltacorr\in\sigma^*\}\cdot |\sigma| + Pr\{\deltacorr\in\sigma^*\}\cdot E(|\sigma^{*}|) + E(|\sigma\cap \sigma^{*}|)\cdot (|\sigma| + E(|\sigma^{*}|))$$
		by rewriting and multiplying by $\frac{2}{(|\sigma| + E(|\sigma^{*}|)\cdot (|\sigma| + E(|\sigma^{*}|) + 1)}$ we get
		$$\frac{2\cdot E(|\sigma\cap \sigma^{*}|)}{|\sigma| + E(|\sigma^{*}|)} \leq \frac{2\cdot(E(|\sigma\cap \sigma^{*}|) + Pr\{\deltacorr\in\sigma^*\})}{|\sigma| + E(|\sigma^{*}|) + 1}$$
		and conclude that
		$$E(F(\sigma)) \leq E(F(\sigma^{\prime}))$$
	\end{itemize}
	which finalizes the proof.	
\end{proof}

\begin{proof}[Proof of Theorem~\ref{thm:MIEM}]
	\begin{sloppypar}		
		The first part of the theorem follows directly from Lemma~\ref{lemma:r_miem}. For the remainder of the proof we rely on Lemma~\ref{lemma:p_f_ineq}.
	\end{sloppypar}
	Let $(\sigma, \sigma^{\prime})\in\Sigma^{P}$ be a match pair in $\Sigma^{P}$. By definition, since $(\sigma, \sigma^{\prime})\in\Sigma^{P}$, 
	then $P(\sigma)\leq P(\deltacorr)$ and by Lemma~\ref{lemma:p_f_ineq}, we can conclude that
	$P(\sigma) \leq P(\sigma^{\prime})$.
	
	Similarly, let $(\sigma, \sigma^{\prime})\in\Sigma^{F}$ be a match pair in $\Sigma^{F}$. By definition, since $(\sigma, \sigma^{\prime})\in\Sigma^{F}$, then $0.5\cdot F(\sigma)\leq P(\deltacorr)$ and by Lemma~\ref{lemma:p_f_ineq}, we can infer that
	$F(\sigma) \leq F(\sigma^{\prime})$, which concludes the proof.
\end{proof}

\begin{theorem}\label{thm:probLocalAnnealapp}
	Let $R/P/F$ be a random variable, whose values are taken from the domain of $[0,1]$, and $\deltacorr$ be a singleton correspondence set ($|\deltacorr|=1$). $\deltacorr$ is a \emph{probabilistic local annealer with respect to $R/P/F$ over $\Sigma^{\subseteq_1}/\Sigma^{E(P)}/\Sigma^{E(F)}$}.
\end{theorem}

\begin{proof}[Proof of Theorem~\ref{thm:probLocalAnneal}]
	Let $\deltacorr$ be a singleton correspondence set ($|\deltacorr|=1$).
	Let $R$ be a random variable and let $\deltacorr=\deltacorr_{(\sigma,\sigma^{\prime})}$ such that 	$(\sigma,\sigma^{\prime})\in\Sigma^{\subseteq_1}$. 
	
	According to Corollary~\ref{corol:annealer}, $\deltacorr$ is a local annealer with respect to $R$ over $\Sigma^{\subseteq_1}$ and therefore $R(\sigma)\leq R(\sigma^{\prime})$, regardless of $Pr\{\deltacorr\in\sigma^*\}$. Therefore, for any $p=Pr\{\deltacorr\in\sigma^*\}, p\cdot R(\sigma)\leq p\cdot R(\sigma^{\prime})$ and by definition of expectation, $E(R(\sigma))\leq E(R(\sigma^{\prime}))$. 
	
	Let $P$ be a random variable and let $\deltacorr=\deltacorr_{(\sigma,\sigma^{\prime})}$ such that $(\sigma,\sigma^{\prime})\in\Sigma^{E(P)}$. By definition of $\Sigma^{E(P)}$, $E(P(\sigma))\leq Pr\{\deltacorr\in\sigma^*\}$ and using Lemma~\ref{lemma:p_f_ineqProb} we obtain $E(P(\sigma)) \leq E(P(\sigma^{\prime}))$. 
	
	Let $F$ be a random variable and let $\deltacorr=\deltacorr_{(\sigma,\sigma^{\prime})}$ such that $(\sigma,\sigma^{\prime})\in\Sigma^{E(F)}$. By definition of $\Sigma^{E(F)}$, $E(F(\sigma))\leq 0.5\cdot Pr\{\deltacorr\in\sigma^*\}$ and using Lemma~\ref{lemma:p_f_ineqProb} we obtain $E(F(\sigma)) \leq E(F(\sigma^{\prime}))$. 
	
	We can therefore conclude, by Definition~\ref{def:probAnnealer}, that $\deltacorr$ is a probabilistic local annealer with respect to	$R/P/F$ over $\Sigma^{\subseteq_1}/\Sigma^{E(P)}/\Sigma^{E(F)}$.
\end{proof}

\setcounter{equation}{6}
\begin{equation}\label{eq:unbiasedMatchingEstim}
	Pr\{M_{ij} \in\sigma^*\} = M_{ij}, \text{ }E(P(\sigma)) = \frac{\sum_{M_{ij}\in\sigma}M_{ij}}{|\sigma|}, \text{ }E(F(\sigma)) =\frac{2\cdot\sum_{M_{ij}\in\sigma}M_{ij}}{|\sigma| + |\sigma^{*}|}
\end{equation}

The details of the computation of Eq.~\ref{eq:unbiasedMatchingEstim} are as follows:
\begin{calc}[Computation of Eq.~\ref{eq:unbiasedMatchingEstim}]
	We first look into the main component in both expressions $E(|\sigma\cap\sigma^{*}|)$, that is, the expected number of correct correspondences in a match $\sigma$.
	
	By rewriting we get   
	$$E(|\sigma\cap\sigma^{*}|) = E\left(\sum_{M_{ij}\in\sigma} {\rm I\!I}_{\{M_{ij}\in\sigma^*\}}\right)$$
	and based on the linearity of expectation, we obtain
	$$E(|\sigma\cap\sigma^{*}|) = \sum_{M_{ij}\in\sigma} E({\rm I\!I}_{\{M_{ij}\in\sigma^*\}})$$
	The expected value of an indicator equals the probability of an event and thus,  
	$$\sum_{M_{ij}\in\sigma} E({\rm I\!I}_{\{M_{ij}\in\sigma^*\}}) = \sum_{M_{ij}\in\sigma} Pr\{M_{ij}\in\sigma^*\}$$
	Assuming unbiased matching (Definition~\ref{def:unbiasedMatching}), we conclude
	$$E(|\sigma\cap\sigma^{*}|) = \sum_{M_{ij}\in\sigma} M_{ij}$$
	
	Using Eq.\ref{eq:PandR} and recalling that $|\sigma|$ is deterministic, we obtain
	$$E(P(\sigma)) = \frac{\sum_{M_{ij}\in\sigma} M_{ij}}{|\sigma|}$$
	Similarly, using Eq.~\ref{eq:f_sigma} and assuming that $|\sigma^{*}|$ is also deterministic, we obtain
	$$E(F(\sigma))=\frac{2\cdot\sum_{M_{ij}\in\sigma}M_{ij}}{|\sigma| + |\sigma^{*}|}$$
\end{calc}

\section{Human Matching Biases}
\label{app:biases}
We now provide a of possible human biases based on~\cite{ackerman2019cognitive}.

\subsection{Temporal Bias}
\label{sec:dcm}

The temporal bias is rooted in the \emph{Diminishing Criterion Model (DCM)}, that models a common bias in human confidence judgment. DCM stipulates that the stopping criterion of human decisions is relaxed over time. Thus, a human matcher is more willing to accept a low confidence level after investing some time and effort on finding a correspondence. 

\begin{wrapfigure}{r}{0.5\textwidth}
	\centering 
	\includegraphics[width=\linewidth]{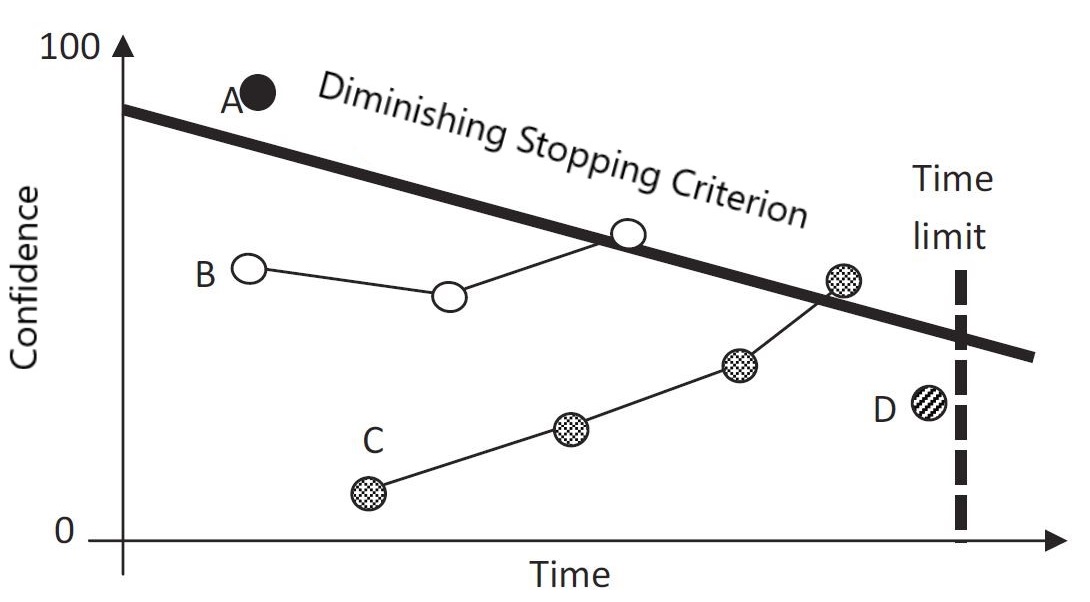}
	\caption[The Diminishing Criterion Model]{DCM with hypothetical confidence ratings for four items and a self-imposed time limit}
	\label{fig:DCM}
\end{wrapfigure}

To demonstrate DCM, consider Figure~\ref{fig:DCM}, which presents hypothetical confidence ratings while performing a schema matching task. Each dot in the figure represents a possible solution to a matching decision ({\em e.g.}, attributes $a_i$ and $b_j$ correspond), and its associated confidence, which changes over time. The search for a solution starts at time $t=0$ and the first dot for each of the four examples represents the first solution a human matcher reaches. As time passes, human matchers continuously evaluate their confidence. In case A, the matcher has a sufficiently high confidence after a short investigation, thus decides to accept it right away. In case B, a candidate correspondence is found quickly but fails to meet the sufficient confidence level. As time passes, together with more comprehensive examination, the confidence level (for the same or a different solution) becomes satisfactory (although the confidence value itself does not change much) and thus it is accepted. In Case C, no immediate candidate stands out, and even when found its confidence is too low to pass the confidence threshold. Therefore, a slow exploration is preformed until the confidence level is sufficiently high. In Case D, an unsatisfactory correspondence is found after a long search process, which fails to meet the stopping criterion before the individual deadline  passes. Thus, the human matcher decides to reject the correspondence.

\subsection{Consensuality Dimension}
\label{sec:consensualDimension}
While the temporal bias relates to the compromises an individual makes over time, the consensuality bias is concerned with multiple matchers and their agreement. Metacognitive studies suggest that a strong predictor for people's confidence is the frequency in which a particular answer is given by a group of people. 

Consensuality serves as a strong motivation to use crowd sourcing for matching and examines whether the number of people who chose a particular match can be used as a predictor of the chance of a match to be correct. Although consensuality in the cognitive literature does not ensure accuracy, the study in~\cite{ackerman2019cognitive} shows strong corrlation between the ampunt of agreement among matchers regarding a correspondence and its correctness. This can also support the use of majority vote to determine correctness. 

\subsection{Control Dimension}
\label{sec:controlDimension}
The third bias, which we ruled out in this work, analyzes the consistency of human matchers when provided the result of an algorithmic solution. 

Metacognitive control is the use of metacognitive monitoring (\emph{i.e.,} judgment of success) to alter behavior. Specifically, control decisions are the decisions pepole make based of their subjective self-assessment of their chance for success. In the context of the present study, the use of algorithmic output for helping the matcher in her task is taken as a control decision, as the human matcher chooses to reassure her decision based on her judgment. Variability in this dimension may be attributed to the predicted tendency of participants who do not use system suggestions to be more engaged in the task and recruit more mental effort than those who use suggestions as a way to ease their cognitive load.

\end{document}